\documentclass[11pt]{article}
\def\tp{0}
\def\camr{0}
\usepackage{ifthen}
\usepackage{mdwlist}
\usepackage{amsmath,amssymb,amsfonts,amsthm}
\usepackage{bm}
\usepackage{hyperref}
\usepackage{enumitem}
\usepackage{graphicx}
\usepackage{xspace}
\usepackage{verbatim}
\usepackage{algorithm}
\usepackage{algpseudocode}
\usepackage[margin=1in]{geometry}
\usepackage{color}
\usepackage{thm-restate}

\usepackage{latexsym}
\usepackage{epsfig}

\def\a{\alpha}

\def\d{\delta}

\renewcommand{\epsilon}{\ve}
\def\ve{\varepsilon}

\def\z{\zeta}

\def\l{\lambda}

\def\m{\mu}

\def\p{\pi}

\def\r{\rho}

\def\s{\sigma}
\def\S{\Sigma}

\def\E{{\bf E}}

\newcommand{\var}{\mbox{\bf Var}}

\newcommand{\erf}{\mbox{\text erf}}

\newcommand{\prob}[2][]{\text{\bf Pr}\ifthenelse{\not\equal{}{#1}}{_{#1}}{}\!\left[#2\right]}
\newcommand{\expect}[2][]{\text{\bf E}\ifthenelse{\not\equal{}{#1}}{_{#1}}{}\!\left[#2\right]}

\newcommand{\dtv}{d_{\mathrm {TV}}}
\newcommand{\dk}{d_{\mathrm K}}
\newcommand{\dkl}{d_{\mathrm {KL}}}
\newcommand{\Tr}{{\mathrm {Tr}}}

\newtheorem{theorem}{Theorem}

\newtheorem{proposition}{Proposition}
\newtheorem{observation}{Observation}

\newtheorem{fact}{Fact}
\newtheorem{lemma}{Lemma}
\newtheorem{claim}{Claim}
\newtheorem{corollary}{Corollary}

\newtheorem{definition}{Definition}

\newcommand{\ignore}[1]{}
\providecommand{\poly}{\operatorname*{poly}}

\newcommand{\bg}[1]{\medskip\noindent{\bf #1}}

\definecolor{Red}{rgb}{1,0,0}

\newcommand{\oldbound}[1]{{}}

\newcommand{\sparse}{\mathrm{sparse}}
\newcommand{\lrg}{\mathrm{large}}

\title{A Size-Free CLT for Poisson Multinomials and its Applications}

\author {
Constantinos Daskalakis\thanks{Supported by a Microsoft Research Faculty Fellowship, and NSF Award CCF-0953960 (CAREER) and CCF-1551875. This work was done in part while the author was visiting the Simons Institute for the Theory of Computing.}\\
EECS, MIT \\
\tt{costis@mit.edu}
\and
Anindya De\\
Northwestern University \\
\tt{anindya.de1@northwestern.edu}
\and
Gautam Kamath\thanks{Supported by NSF Award CCF-0953960 (CAREER) and ONR grant N00014-12-1-0999. 
This work was done in part while the author was an intern at Microsoft Research Cambridge and visiting the Simons Institute for the Theory of Computing.} \\
EECS, MIT\\
\tt{g@csail.mit.edu}
\and
Christos Tzamos\thanks{Supported by NSF Award CCF-0953960 (CAREER), ONR grant N00014-12-1-0999, and a Simons Award for Graduate Students in Theoretical Computer Science. This work was done in part while the author was visiting the Simons Institute for the Theory of Computing.}\\
EECS, MIT\\
\tt{tzamos@mit.edu}
}
\begin{document}
\addtocounter{page}{-1}
\maketitle
\thispagestyle{empty}

\begin{abstract}
An $(n,k)$-{\em Poisson Multinomial Distribution} (PMD) is the distribution of the sum of $n$ independent random vectors supported on the set ${\cal B}_k=\{e_1,\ldots,e_k\}$ of standard basis vectors in~$\mathbb{R}^k$. We show that any $(n,k)$-PMD is ${\rm poly}({k\over \sigma})$-close in total variation distance to the (appropriately discretized) multi-dimensional Gaussian with the same first two moments, removing the dependence on $n$ from the Central Limit Theorem of Valiant and Valiant~\cite{ValiantV11}. Interestingly, our CLT is obtained by bootstrapping the Valiant-Valiant CLT itself through the structural characterization of PMDs shown in recent work~\cite{DaskalakisKT15}. In turn, our stronger CLT can be leveraged to obtain an efficient PTAS for approximate Nash equilibria in anonymous games, significantly improving the state of the art~\cite{DaskalakisP08}, and matching qualitatively the running time dependence on $n$ and $1/\epsilon$ of the best known algorithm for two-strategy anonymous games~\cite{DaskalakisP09}. Our new CLT also enables the construction of  covers for the set of $(n,k)$-PMDs, which are proper and whose size is shown to be essentially optimal.
Our cover construction combines our CLT with the Shapley-Folkman theorem and recent sparsification results for Laplacian matrices~\cite{BatsonSS12}. Our cover size lower bound is based on an algebraic geometric construction. Finally, leveraging the structural properties of the Fourier spectrum of PMDs we show that these distributions can be learned from $O_k(1/\epsilon^2)$ samples in ${\rm poly}_k(1/\epsilon)$-time, removing the quasi-polynomial dependence of the running time on $1/\epsilon$ from~\cite{DaskalakisKT15}.
\end{abstract}

\newpage
\section{Introduction}
The Poisson Multinomial Distribution (PMD) is the multi-dimensional generalization of the more familiar Poisson Binomial Distribution (PBD). To illustrate its meaning, consider a city of $n$ people and $k$ newspapers. Suppose that person $i$ has his own proclivity to buy each newspaper, so that his purchase each day can be modeled as a random vector $X_i$ -- also called a Categorical Random Variable (CRV) -- taking values in the set ${\cal B}_k=\{e_1,\ldots,e_k\}$ of standard basis vectors in $\mathbb{R}^k$.\footnote{Of course, we can always add a dummy newspaper to account for the possibility that somebody may decide not to buy a newspaper.} If people buy their newspapers independently, the total circulation of newspapers is the sum $X=\sum_i X_i$. The distribution of $X$ is a $(n,k)$-PMD, and we need $n\cdot (k-1)$ parameters to describe it. When $k=2$, the distribution is called an $n$-PBD. When people have identical proclivities to buy the different newspapers, the distribution degenerates to the more familiar {Multinomial} (general $k$) or {Binomial} ($k=2$) distribution.\footnote{It is customary to project Binomial and Poisson Binomial distributions to one of their coordinates. In multiple dimensions, it will be convenient to call a distribution resulting from the projection of a PMD to all but one coordinates a Generalized Multinomial distribution (GMD).} In other words, $n$-PBDs are distributions of sums of $n$ independent, not necessarily identically distributed Bernoullis, while $(n,k)$-PMDs are their multi-dimensional generalization, where we are summing independent categorical random variables. As such, these distributions are one of the most widely studied multi-dimensional families of distributions. 

In Probability theory, a large body of literature aims at approximating PMDs via simpler distributions. The Central Limit Theorem (CLT) informs us that the limiting behavior of an appropriately normalized PMD, as $n \rightarrow \infty$, is a multi-dimensional Gaussian, under conditions on the eigenvalues of the summands' covariance matrices; see e.g.~\cite{VanderVaart00}. The rate of convergence in the CLT is quantified by multi-dimensional Berry-Esseen theorems. As PMDs are discrete, while Gaussians are continuous distributions, such theorems typically bound the maximum difference in probabilities assigned by the two distributions to convex subsets of $\mathbb{R}^k$. Again, these bounds degrade as the PMD's covariance matrix tends to singularity; see e.g.~\cite{Bentkus05, ChenST14}. Similarly, approximations of PMDs via multivariate Poisson~\cite{Barbour88,DeheuvelsP88}, multinomial~\cite{Loh92}, and other discrete distributions has been intensely studied, often using Stein's method.

In theoretical computer science, PMDs are commonly used in the analysis of randomized algorithms, often through large deviation inequalities. They have also found applications in algorithmic problems where one is looking for a collection of random vectors optimizing a certain probabilistic objective, or satisfying probabilistic constraints. For example, understanding the behavior of PMDs has led to polynomial-time approximation schemes for anonymous games~\cite{Milchtaich96, Blonski99, Blonski05, Kalai05, DaskalakisP07, DaskalakisP08, DaskalakisP09}, despite the {\tt PPAD}-completeness of their exact equilibria~\cite{ChenDO2015}. Anonymous games are games where a large number $n$ of players share the same $k$ strategies, and each player's utility only depends on his own choice of strategy and the number of other players that chose each of the $k$ strategies. In particular, the expected payoff of each player depends on the PMD resulting from the mixed strategies of the other players. It turns out that understanding the behavior of PMDs provides a handle on the computation of approximate Nash equilibria. One of our main contributions is to advance the state of the art for computing approximate Nash equilibria in anonymous games. We will come to this contribution shortly.

\paragraph{A New CLT\ifnum\camr=0
.
\fi} Recently Valiant and Valiant have used PMDs to obtain sample complexity lower bounds for testing symmetric properties of distributions~\cite{ValiantV11}. The workhorse in their lower bounds is a new CLT bounding the total variation distance between a $(n,k)$-GMD and a multidimensional Gaussian with the same mean vector and covariance matrix. Since they are comparing a discrete to a continuous distribution under the total variation distance, they need to discretize the Gaussian by rounding its coordinates to their closest point in the integer lattice. If $X$ is distributed according to some $(n,k)$-GMD with mean vector $\mu$ and covariance matrix $\Sigma$, and $Y$ is distributed according to the multi-dimensional Gaussian ${\cal N}(\mu, \Sigma)$, \cite{ValiantV11} shows that:
\begin{align}
\dtv(X,\lfloor Y \rceil) \le {k^{4/3} \over \sigma^{1/3}}\cdot 2.2 \cdot (3.1+0.83 \log n)^{2/3}, \label{eq:VV CLT}
\end{align}
where $\sigma^2$ is the minimum eigenvalue of $\Sigma$ and $\lfloor Y \rceil$ denotes the rounding of $Y$ to the closest point in the integer lattice. The dependence of the bound on the dimension $k$ and the minimum eigenvalue $\sigma^2$ is necessary, and quite typical of Berry-Esseen type bounds. 
Answering a question raised in~\cite{ValiantV11}, we prove a qualitatively stronger CLT by showing that the explicit dependence of the bound on $n$ can be removed (hence, the CLT is ``size-free''). 
\begin{theorem}[Size-free CLT]
  \label{thm:newclt}
  Suppose that $X$ is distributed according to some $(n,k)$-GMD with mean $\m$ and covariance matrix $\S$, and $Y \sim {\cal N}(\m,\S)$.  There exists some constant $C>0$ such that
  \begin{align}\dtv\left(X,\lfloor Y \rceil \right) \leq C\frac{k^{7/2}}{\s^{1/10}}, \label{eq:our CLT bound}
	\end{align}
  where $\s^2$ is the minimum eigenvalue of $\S$.
\end{theorem}

Interestingly, Theorem~\ref{thm:newclt} is proven by bootstrapping the Valiant-Valiant CLT itself. 
Indeed, this CLT was used as one of the key ingredients in a recent structural characterization of PMDs~\cite{DaskalakisKT15}, where it was shown that any $(n,k)$-Poisson multinomial random vector is $\epsilon$-close in total variation distance to the sum of an (appropriately discretized) Gaussian and a $({\rm poly}(k/\epsilon),k)$-Poisson multinomial random vector; see Theorem~\ref{thm:struct}. In turn, we prove Theorem~\ref{thm:newclt} by using Theorem~\ref{thm:struct} as a black box.

We start with an invocation of the structural characterization for some $\epsilon={\rm poly}(k/\sigma)$. 
With a judicious such choice of $\epsilon$, the structural result approximates an arbitrary $(n,k)$-Poisson multinomial random vector $X$ (to within ${\rm poly}(k/\sigma)$ in total variation distance) by the sum $G+P$ of a discretized Gaussian $G$ and a $(o(\sigma),k)$-Poisson multinomial random vector $P$. 
As $P$ has too few components, namely $o(\sigma)$, we show that $G$ must account for the variance of $X$, which is at least $\sigma^2$ in all directions. 
Next, since $G$ has variance $\Omega(\sigma^2)$ in all directions and $P$ has variance $o(\sigma^2)$, we can show that $G$ swamps $P$, in that $\dtv(G,G+P)$ is small, using Proposition~\ref{prop:subtractsparse}.
So $\dtv(X,G)$ is also small by triangle inequality. 
The remaining step is to argue that $G$ can be replaced by a discretized multidimensional Gaussian with the same first two moments as $X$. 
This is done in two parts.
First, since $X$ and $G$ are close in total variation distance, we can argue that their first two moments are close using Proposition~\ref{prop:kdparamclose}.
Then, we relate $G$ to a discretized Gaussian with the same mean and covariance as $X$ using Lemma~\ref{lem:dtvgaussian}, which bounds the total variation distance between two Gaussians with similar moments.
Finally, we need to argue that the resulting Gaussian can be trivially discretized to the integer lattice, obviating the need for a more sophisticated structure preserving rounding.

For more details on our proof's approach, see Section~\ref{sec:body_clt}.

\bigskip
\noindent In the remainder of this section we discuss the algorithmic applications of our CLT, concluding with our improved algorithms for learning PMDs using Fourier analysis.

\vspace{-5pt}\paragraph{Anonymous Games\ifnum\camr=0
.
\fi} We have already discussed anonymous games earlier in this section, where we have also explained their relation to PMDs. In particular, the expected utility $u_i$ of some player $i$ in a $n$-player $k$-strategy anonymous game only depends on his own choice of mixed strategy $X_i$ and the $(n-1,k)$-Poisson multinomial random vector $\sum_{j \neq i}X_j$ aggregating the mixed strategies of his opponents. It is therefore natural to expect that a better understanding of the structure of PMDs could lead to improved algorithms for computing Nash equilibria in these games. Indeed, earlier work~\cite{DaskalakisP08,DaskalakisP15b} has exploited this connection to obtain algorithms for approximate Nash equilibria, whose running time is
$$n^{O\left(2^{k^2}\cdot \left({f(k) \over \ve} \right)^{6 \cdot k}\right)}\text{, where }f(k) \le 2^{3k-1} k^{k^2+1} k!$$
While clearly of theoretical interest, this bound shows that anonymous games are one of the few classes of games where \emph{approximate} equilibria can be efficiently computed, while \emph{exact} equilibria are {\tt PPAD}-hard~\cite{ChenDO2015}, even for $n$-player $7$-strategy anonymous games.
Exploiting our CLT we obtain a significant improvement over~\cite{DaskalakisP08}.
\begin{theorem}[Equilibria in Anonymous Games]
  \label{thm:PTAS for anonymous}
  An $\epsilon$-approximately well supported Nash equilibrium of an $n$-player $k$-strategy anonymous games whose utilities are in $[0,1]$ can be computed in time:\footnote{As it is customary in Nash equilibrium algorithms, approximate Nash equilibria are defined with respect to additive approximations and the player utilities are normalized to $[0,1]$ to make these approximations meaningful.}
	 \begin{align}
	n^{O(k^2)} \cdot 2^{O(k^{5k} \cdot \log^{k+2} (1/\ve))}. \label{eq:runtime anonymous}
	\end{align}
\end{theorem}
\noindent The salient feature of Theorem~\ref{thm:PTAS for anonymous} is the polynomial dependence of the running time on $n$ and its quasi-polynomial dependence on $\epsilon^{-1}$. In terms of these dependencies our algorithm matches the best known algorithm for $2$-strategy anonymous games~\cite{DaskalakisP09}, where much more is known given the single-dimensional nature of $(n,2)$-PMDs. 

Moreover, the recent hardness results for anonymous games \cite{ChenDO2015} establish that not only finding an exact but also a $2^{n^a}$-approximate Nash equilibrium is {\tt PPAD}-hard. An interesting corollary of Theorem~\ref{thm:PTAS for anonymous} is that this cannot be pushed to ${\rm poly}(1/n)$-approximations, unless {\tt PPAD} can be solved in quasi-polynomial time.
\begin{corollary}[Non-{\tt PPAD} Hardness of FPTAS] 
  \label{thm:no PPAD hardness}
{\ifnum\camr=1 ~\newline \fi} Unless {\tt PPAD} $\subseteq$ {\tt Quasi-PTIME}, it is {\em not} {\tt PPAD}-hard to find a ${\rm poly}(1/n)$-approximately well supported Nash equilibrium in anonymous games, for any ${\rm poly}(\cdot)$.
\end{corollary}
\noindent It is interesting to contrast this corollary with normal-form games where it is known that computing inverse polynomial approximations {\em is} {\tt PPAD}-hard~\cite{DaskalakisGP09,ChenDT09}.

From a technical standpoint, our algorithm for anonymous games uses the structural understanding of PMDs as follows.
Since every player views the aggregate strategies of the other players as a PMD, one approach would be to guess each player's view using a cover as developed in \cite{DaskalakisKT15}.
However, this approach gives a runtime which is exponential in $n$, since it requires us to enumerate the cover for each player.
An alternative approach is to guess the overall PMD which occurs at a Nash equilibrium, and guess appropriate ``corrections'' that allow us to infer each player's view.
To do this, we must find an alternative PMD which approximately matches the PMD at Nash in the following sense:
\begin{itemize}
\item The PMD that results by removing the CRV corresponding to a player should be close to the view that the player observes;
\item A player's CRV must only assign probability to strategies which are approximate best responses to his view.
\end{itemize}
It turns out that these conditions can be satisfied by using a careful dynamic program together with the structural understanding provided by \cite{DaskalakisKT15} and the CLT of Theorem~\ref{thm:newclt}. 
According to this structural result, we can partition the players into a ``sparse'' and a ``Gaussian'' component.
Moreover, our CLT implies that matching the first two moments of the Gaussian suffices to approximate this component.
This allows us to perform guesses at a different granularity for the sparse and Gaussian components.
Roughly speaking, our dynamic program guesses a succinct representation of the two components and tries to compute CRVs which obey this representation and satisfy the conditions outlined above.

For more details on our PTAS, refer to Section~\ref{sec:body_anon}.

\paragraph{Proper Covers\ifnum\camr=0
.
\fi} The second application of our CLT is to obtain proper covers for the set ${\cal S}_{n,k}$ of $(n,k)$-PMDs. A proper $\epsilon$-cover of ${\cal S}_{n,k}$, in total variation distance, is a subset ${\cal S}_{n,k,\epsilon} \subseteq {\cal S}_{n,k}$ such that for all $(X_1,\ldots,X_n) \in {\cal S}_{n,k}$ there exists some  $(Y_1,\ldots,Y_n) \in {\cal S}_{n,k,\epsilon}$ such that $\dtv\left(\sum_i X_i,\sum_i Y_i\right) \le \epsilon$. We show the following:

\begin{theorem}[Proper Cover]
  \label{thm:propercover}
  For all $n, k \in \mathbb{N}$, and $\ve > 0$, there exists a proper $\ve$-cover, in total variation distance, of the set of all $(n,k)$-PMDs whose size is 
  \begin{align}
	n^{O(k)} \cdot \min \left\{2^{\poly(k/\ve)}, 2^{O(k^{5k}\log^{k+2}(1/\ve))}\right\}. \label{eq:size of our cover}
	\end{align}
  Moreover, we can efficiently enumerate this cover in time polynomial in its size.
\end{theorem}

\noindent It is important to contrast Theorem~\ref{thm:propercover} with Theorem 2 in~\cite{DaskalakisKT15}, which provides a non-proper cover whose size is similar, albeit with a leading factor of $n^{O(k^2)}$. Instead, our cover is proper, which is important for approximation algorithms that require searching over PMDs. Its dependence on $n$ is also optimal, as the number of $(n,k)$-PMDs whose summands are deterministic is already $n^{\Omega(k)}$. Moreover, we provide a lower bound for the dependence on $1/\ve$, establishing that the quasi-polynomial dependence is also essentially optimal.

\begin{theorem}[Cover Size Lower Bound] \label{thm:lower-bound}
{\ifnum\camr=1 ~\newline \fi} For any $n, k \in \mathbb{Z}$,  $\epsilon>0$ such that $n > 2 \log^k(1/\epsilon)$,  there exist $(n,k)$-PMDs $Z_1, \ldots, Z_s$ such that for $1 \le i < j \le s$, $d_{TV}(Z_i, Z_j) \ge \epsilon$ and $s = \Omega_k(n^{k-1} \cdot 2^{\tilde{\Omega}(\log^{k-1}(1/\epsilon))})$. The $\tilde{\Omega}$ in the exponent hides factors of $\poly(\log \log (1/\epsilon))$ and dependence on $k$. 
\end{theorem}

We describe our proper cover construction in two parts.
First, we give details on how to construct a non-proper cover of size $n^{O(k)}$.
The main tool we use is the existence of spectral sparsifiers for Laplacian matrices.
\ifnum\tp=0
Our non-proper cover sparsifies the non-proper cover of~\cite{DaskalakisKT15}, showing how its leading factor of $n^{O(k^2)}$ can be reduced to $n^{O(k)}$. Roughly speaking, the factor of $n^{O(k^2)}$ was due to spectrally approximating all possible covariance matrices $\S$, whose $O(k^2)$ entries are bounded by $n$. These covariance matrices corresponded to covariance matrices of $(n,k)$-PMDs, and the cover maintained for each such $\S$ some $\S'$ such that  $|v^T(\S-\S') v| \le {\rm poly}(\epsilon/k) \cdot v^T\S v, \forall v$. (We call this guarantee a ``${\rm poly}(\ve/k)$-spectral approximation.'') The realization leading to our sparsification result is that covariance matrices of PMDs are in fact graph Laplacians. Indeed, a $(n,k)$-PMD, $X=\sum_i X_i$, has covariance matrix, ${\rm \bf cov}(X)=\sum_i {\rm \bf cov}(X_i)$, corresponding to the sum of the covariance matrices of its summands.  Now the covariance matrix of a $k$-CRV, $X_i$, is actually the Laplacian of a graph that has one node $j$ per dimension, along with an edge from node $j$ to node $j'$ of weight $\E[X_{ij}]\cdot \E[X_{ij'}]$; and the covariance matrix of a $(n,k)$-PMD is the Laplacian of the graph with the sum of the weights from each constituent $k$-CRV{\ifnum\camr=0 -- see Observation~\ref{obs:pmdlaplacian}\fi}. We show that Laplacians corresponding to $(n,k)$-PMDs can be ${\rm poly}(\ve/k)$-spectrally covered with a set of covariance matrices of size $n^{O(k)} \cdot \left({ k \over \epsilon}\right)^{O(k^3)}$.

We appeal to recent results in spectral sparsification of Laplacian matrices \cite{SpielmanT11, SpielmanS11, BatsonSS12, BatsonSST13}.
In particular, we use the result of Batson, Spielman, and Srivastava~\cite{BatsonSS12} {\ifnum\camr=0 (Theorem~\ref{thm:spectralsparse}) \fi}to argue that the underlying graph can be sparsified to linearly many edges in the dimension $k$. We do this in the hopes that we would have fewer parameters in the covariance matrix to guess. Unfortunately, the \cite{BatsonSS12} sparsification theorem has polynomial dependence in the accuracy. So applying it with a ${\rm poly}(\ve/k)$-approximation error, which is what we need, gives a meaningless result (namely no sparsification at all). Instead, we only use this theorem to get a rough $O(1)$-spectral cover of $(n,k)$-PMD covariance matrices. Around every covariance matrix in this rough cover we grow a local ${\rm poly}(\ve/k)$-spectral cover. Roughly speaking, as the $O(1)$-spectral cover provides multiplicative approximation to the variance in every direction $v$, every covariance matrix in this cover gives us a multiplicative handle on the eigenvalues of the matrices approximated by it. This is sufficient information to cover these matrices to ${\rm poly}(\ve/k)$-spectral error with a ``local'' spectral cover of size $(k/\ve)^{O(k^2)}${\ifnum\camr=0 -- see Lemma~\ref{lem:psdcover}\fi}. Putting everything together, we get a ${\rm poly}(\ve/k)$-spectral cover of all covariance matrices of $(n,k)$-PMDs of size $n^{O(k)} \cdot \left({ k \over \epsilon}\right)^{O(k^3)}${\ifnum\camr=0 -- see Section~\ref{subsec:counting size of cover}\fi}. As covering these matrices was the bottleneck in the size of the non-proper cover, this completes the construction of a non-proper cover whose size is~\eqref{eq:size of our cover}.

\fi

Further details on our non-proper construction are provided in Section~\ref{sec:body_upperb}.

We then show how to convert each element of this improper cover back to a PMD.
We bypass the difficulty involved with a non-convex optimization problem by exploiting the ``almost convexity'' of the Minkowski sum as guaranteed by the Shapley-Folkman lemma.
\ifnum\tp=0
The cover {\ifnum\camr=0 provided by Theorem~\ref{thm:newcover} \else described above \fi}is non-proper. 
It utilizes the structural result of~\cite{DaskalakisKT15} (see Theorem~\ref{thm:struct}) to cover the set of $(n,k)$-PMDs by hypotheses which take the form of the convolution of a discretized multidimensional Gaussian with a $({\rm poly}(k/\ve),k)$-PMD. 
The benefit of this class of hypotheses is that they have only ${\rm poly}(k/\ve)$ parameters.
This allows us to efficiently enumerate over them, resulting in a cover size of~\eqref{eq:size of our cover}.
To convert this cover into a proper one, we need an algorithm which, given a convolution of a discretized Gaussian with some $(\kappa\triangleq {\rm poly}(k/\ve),k)$-PMD, finds a $(n,k)$-PMD that is $O(\ve)$-close to this distribution, if such a PMD exists. As the $(\kappa,k)$-PMD is already a PMD, this boils down to answering whether a given discretized Gaussian with parameters $(\m,\S)$ is $O(\ve)$-close to a $(n-\kappa,k)$-PMD. To answer this question, we exploit our new CLT (Theorem~\ref{thm:newclt}) and the fact that the discretized Gaussians that arise in the cover have an extra property: all their non-zero eigenvalues are at least ${\rm poly}(k/\ve)$-large. Exploiting this we argue that (i) if there exists an $(n-\kappa,k)$-PMD that is close to the discretized Gaussian with parameters $(\m,\S)$, then its mean $\m'$ should be close to $\m$ and its covariance matrix $\S'$ should be spectrally close to $\S$; and (ii) if we can find any  $(n-\kappa,k)$-PMD with with these properties, then it will be close to the discretized Gaussian. With (i) and (ii), our task becomes a convex geometry question: Let $\cal M$ be all possible first two moments $(\E[Y],{\rm \bf cov}(Y))$, of $k$-CRVs $Y$ whose parameters have been finely discretized. As the first two moments of a $(n-\kappa,k)$-PMD are sums of the first two moments of its constituent $k$-CRVs, we can reduce our problem to finding a point in the Minkowski sum ${\cal M}^{\oplus n - \kappa}$ that (spectrally) approximates the target $(\m,\S)$. We write an LP to find a point in the convex hull of ${\cal M}^{\oplus n-\kappa}$ with this property, and the Shapley-Folkman theorem to ``round'' it into a point in ${\cal M}^{\oplus n-\kappa}$ that is only a little worse. The Shapley-Folkman theorem comes in handy because  $\cal M$ lives in $\mathbb{R}^{O(k^2)}$, i.e. much smaller dimension than $n-\kappa$. The whole approximation can be carried out in time $n^{O(k)}${\ifnum\camr=0 -- see Lemma~\ref{lem:spectralsf}\fi}.

\fi
Details on this conversion process are provided in Section~\ref{sec:body_proper}.

Our lower bound is described further in Section~\ref{sec:body_lowerb}.
Our technique shows a lower bound on the metric entropy of a polynomial map of the moments of  PMDs using an extension of B\'ezout's theorem and other tools from algebraic geometry.

\paragraph{Learning\ifnum\camr=0
.
\fi}  
Finally, we give a new learning algorithm for PMDs:
\begin{theorem}\label{thm:learning}
For all $n, k \in \mathbb{N}$ and $\ve>0$, there is a learning algorithm for $(n,k)$-PMDs with the following properties: Let $X=\sum_{i=1}^n X_i$ be any $(n,k)$-Poisson multinomial random vector. The algorithm uses 
$\frac{\poly(k, \log (1/\ve))^k }{\ve^2} $
samples from $X$, runs in time\footnote{We work in the standard ``word RAM'' model in which basic arithmetic operations on $O(\log n)$-bit integers are assumed to take constant time.} 
{$\poly\left(\frac{k}{\ve}\right)^{k^2}$}
and with probability at least $9/10$ outputs a (succinct description of a) random vector $\tilde{X}$ such that $\dtv(X, \tilde{X}) \le \ve$.
\end{theorem}
This improves the learning algorithm from \cite{DaskalakisKT15} by eliminating the superpolynomial dependence on $\epsilon$ in the running time that was obtained in that paper.
Our algorithm exploits properties of the continuous Fourier transform of a PMD, as opposed to recent work by Diakonikolas, Kane and Stewart on learning univariate sums of independent integer random variables, which uses the discrete Fourier transform \cite{DiakonikolasKS16a}.
They also apply similar discrete Fourier techniques in their simultaneous work on PMDs \cite{DiakonikolasKS16c}.

We note that such Fourier-based learning algorithms may simply output a description of the Fourier transform of a distribution.
This allows one to compute the PMF of the distribution at any point of interest, but it is not obvious how to sample from such a description.
Our algorithm outputs an explicit description of a distribution, which allows one to efficiently (i.e., in time independent of $n$) draw samples from the distribution.
In contrast, they output the Fourier transform of a distribution and describe how to sample from it. 

For more details on our learning algorithm, refer to Section~\ref{sec:body_learning}.

\subsection{Comparison of Results with \cite{DiakonikolasKS16c}}
Simultaneous to our work, Diakonikolas, Kane, and Stewart also studied Poisson Multinomial distributions \cite{DiakonikolasKS16c}.
In this section, we describe and compare their results with ours. 
While both papers independently prove many qualitatively similar results, the techniques are quite different, and thus both may be of independent interest.

Both papers prove new CLTs, which manage to remove the dependence on $n$ which is found in the CLT of \cite{ValiantV11}, while the dependence on $k$ and $1/\s$ remains polynomial.
Additionally, both works improve upon the previous best covers for PMDs \cite{DaskalakisKT15}.
First, both manage to reduce the size of the cover -- interestingly, the two improvements seem to be orthogonal.
Our result improves the dependence on $n$ from $n^{k^2}$ to $n^{O(k)}$, while theirs improves the dependence on $k$ and $1/\ve$ from $(1/\ve)^{O(k^{5k} \log^{k+1}(1/\ve))}$ to $(1/\ve)^{O(k \log(k/\ve)/\log \log (k/\ve))^{k-1}}$.\footnote{
We note that this upper bound holds for $k>2$: for $k = 2$, \cite{DiakonikolasKS16a} proves the tight cover size bound of $n \cdot (1/\ve)^{\Theta(\log(1/\ve))}$.
}
Furthermore, both papers describe how to efficiently achieve a proper cover of this size.
These cover sizes are asymptotically optimal, as shown by lower bounds in both papers.
In particular, the double-exponential dependence in $k$ is necessary.
Both works also consider the problem of finding approximate Nash equilibria in anonymous games.
The complexity of both algorithms is roughly comparable to the PMD cover size.
Finally, both papers study the learning of PMDs, obtaining algorithms with sample complexity $\poly(k, \log(1/\ve))^k/\ve^2$.
The runtime of our algorithm is $\poly(k/\ve)^{k^2}$, and the runtime of their algorithm is $\poly(k, \log(1/\ve))^{k}/\ve^2 \cdot \log n$, both in the standard word RAM model.

\section{Preliminaries}
\ifnum\tp=0
\ifnum\tp=0
\subsection{Definitions}
\else
\subsection{Additional Definitions}
\fi
\label{sec:moredefs}
We more formally define several of the distribution classes we consider.

\begin{definition}
A $k$-\emph{Categorical Random Variable} ($k$-CRV) is a random variable that takes values in $\{e_1,\dots,e_k\}$ where $e_j$ is the $k$-dimensional unit vector along direction $j$.
$\p(i)$ is the probability of observing $e_i$.
\end{definition}

\begin{definition}
An $(n,k)$-\emph{Poisson Multinomial Distribution} ($(n,k)$-PMD) is given by the law of the sum of $n$ independent but not necessarily identical $k$-CRVs.
An $(n,k)$-PMD is parameterized by a nonnegative matrix $\p \in [0,1]^{n \times k}$ each of whose rows sum to $1$ is denoted by $M^\p$, and is defined by the following random process:
for each row $\p(i,\cdot)$ of matrix $\p$ interpret it as a probability distribution over the columns of $\p$ and draw a column index from this distribution. Finally, return a row vector recording the total number of samples falling into each column (the histogram of the samples).
\end{definition}

We note that a sample from an $(n,k)$-PMD is redundant  -- given $k-1$ coordinates of a sample, we can recover the final coordinate by noting that the sum of all $k$ coordinates is $n$.
For instance, while a Binomial distribution is over a support of size $2$, a sample is $1$-dimensional since the frequency of the other coordinate may be inferred given the parameter $n$.
With this inspiration in mind, we define the Generalized Multinomial Distribution, which is the primary object of study in \cite{ValiantV11}. 

\begin{definition}
  A \emph{Truncated $k$-Categorical Random Variable} is a random variable that takes values in $\{0, e_1,\dots,e_{k-1}\}$ where $e_j$ is the $(k-1)$-dimensional unit vector along direction $j$, and $0$ is the $(k-1)$ dimensional zero vector.
  $\r(0)$ is the probability of observing the zero vector, and $\r(i)$ is the probability of observing $e_i$.
\end{definition}

\begin{definition}
  \label{def:GMD}
  An $(n,k)$-\emph{Generalized Multinomial Distribution} ($(n,k)$-GMD) is given by the law of the sum of $n$ independent but not necessarily identical truncated $k$-CRVs.
  A GMD is parameterized by a nonnegative matrix $\r \in [0,1]^{n \times (k-1)}$ each of whose rows sum to at most $1$ is denoted by $G^\r$, and is defined by the following random process:
  for each row $\r(i,\cdot)$ of matrix $\r$ interpret it as a probability distribution over the columns of $\r$ -- including, if $\sum_{j=1}^k \r(i,j) <1$, an ``invisible'' column $0$ -- and draw a column index from this distribution. Finally, return a row vector recording the total number of samples falling into each column (the histogram of the samples).
\end{definition}
  For both $(n,k)$-PMDs and $(n,k)$-GMDs, we will refer to $n$ and $k$ as the \emph{size} and \emph{dimension}, respectively.

We note that a PMD corresponds to a GMD where the ``invisible'' column is the zero vector, and thus the definition of GMDs is more general than that of PMDs.
However, whenever we refer to a GMD in this paper, it will explicitly have a non-zero invisible column.

While we will approximate the Multinomial distribution with Gaussian distributions, it does not make sense to compare discrete distributions with continuous distributions, since the total variation distance is always $1$.
As such, we must discretize the Gaussian distributions.
We will use the notation $\lfloor x \rceil$ to say that $x$ is rounded to the nearest integer (with ties being broken arbitrarily).
If $x$ is a vector, we round each coordinate independently to the nearest integer.
\ifnum\tp=0
\begin{definition}
  The $k$-dimensional \emph{Discretized Gaussian Distribution} with mean $\m$ and covariance matrix $\S$, denoted $\lfloor\mathcal{N}(\m,\S) \rceil$, is the distribution with support $\mathbb{Z}^k$ obtained by sampling according to the $k$-dimensional Gaussian $\mathcal{N}(\m,\S)$, and then rounding each coordinate to the nearest integer.
\end{definition}
\else
\discretizedgaussian*
\fi

As seen in the definition of an $(n,k)$-GMD, we have one coordinate which is equal to $n$ minus the sum of the other coordinates.
We define a similar notion for a discretized Gaussian.
However, we go one step further, to take care of when there are several such Gaussians which live in disjoint dimensions.
By this, we mean that given two Gaussians, the set of directions in which they have a non-zero variance are disjoint.
Without loss of generality (because we can simply relabel the dimensions), we assume all of a Gaussian's non-zero variance directions are consecutive, i.e., the covariance matrix is all zeros, except for a single block on the diagonal.
Therefore, when we add the covariance matrices, the result is block diagonal.
The resulting distribution is described in the following definition.

\ifnum\tp=0
\begin{definition}
  The \emph{structure preserving rounding} of a multidimensional Gaussian Distribution takes as input a multi-dimensional Gaussian $\mathcal{N}(\m,\S)$ with $\S$ in block-diagonal form.
  It chooses one coordinate as a ``pivot'' in each block, samples from the Gaussian ignoring these pivots and rounds each value to the nearest integer.
  Finally, the pivot coordinate of each block is set by taking the difference between the sum of the means and the sum of the values sampled within the block.
\end{definition}
\else
\structuregaussian*
\fi

Finally, we formally define the notion of a cover.

\begin{definition}
An \emph{$\ve$-cover} for a set of distributions $\mathcal{S}$ is a set of distributions $\mathcal{S'}$ such that for any distribution $X \in \mathcal{S}$, there exists some distribution $Y \in \mathcal{S'}$ such that $\dtv(X,Y) \leq \ve$. 
A cover is \emph{proper} if $\mathcal{S'} \subseteq \mathcal{S}$.
\end{definition}

\subsection{Probability Metrics}
To compare probability distributions, we will require the total variation and Kolmogorov distances:
\begin{definition}
The \emph{total variation distance} between two probability measures $P$ and $Q$ on a $\s$-algebra $F$ is defined by
$$\dtv(P,Q) = \sup_{A \in F} |P(A) - Q(A)| = \frac12\|P - Q\|_1.$$
\end{definition}

Unless explicitly stated otherwise, in this paper, when two distributions are said to be $\ve$-close, we mean in total variation distance.

\begin{definition}
  The \emph{Kolmogorov distance} between two probability measures $P$ and $Q$ with CDFs $F_P$ and $F_Q$ is defined by
  $$\dk(P,Q) = \sup_{x \in \mathbb{R}} |F_P(x) - F_Q(x)|.$$
\end{definition}

We note that Kolmogorov distance is, in general, weaker than total variation distance. 
In particular, total variation distance between two distributions is lower bounded by the Kolmogorov distance.
\begin{fact}
\label{fct:drelation}
$\dk(P,Q) \leq \dtv(P,Q)$
\end{fact}

\else
\subsection{Definitions}
We consider several distribution classes, including $k$-CRVs, $(n,k)$-PMDs, and $(n,k)$-GMDs.
Informal definitions have already been provided above, for formal definitions see 
\ifnum\camr=0
Section~\ref{sec:moredefs}.
\else
the full version of this paper.
\fi
In this section, we highlight the Gaussian distributions which are used to approximate PMDs.

First, since we wish to approximate a PMD in total variation distance, it is clear that we must use a discrete distribution.
As such, we require a discretized Gaussian, which rounds samples from a Gaussian to the nearest integer point.

\begin{restatable}{definition}{discretizedgaussian}
  The $k$-dimensional \emph{Discretized Gaussian Distribution} with mean $\m$ and covariance matrix $\S$, denoted $\lfloor\mathcal{N}(\m,\S) \rceil$, is the distribution with support $\mathbb{Z}^k$ obtained by sampling according to the $k$-dimensional Gaussian $\mathcal{N}(\m,\S)$, and then rounding each coordinate to the nearest integer.
\end{restatable}

We note that this is insufficient to approximate a PMD, due to two main issues.
First, a sample from an $(n,k)$-PMD always has coordinates which sum to $n$, which is in general not true for a sample from a discretized Gaussian.
Furthermore, subsets of the CRVs which comprise a PMD may have disjoint supports, which give us multiple such ``sum constraints.''
As such, we use a class of distributions which may accomodate these requirements: Gaussians with a structure preserving rounding, as introduced in \cite{DaskalakisKT15}. 

\begin{restatable}{definition}{structuregaussian}
  The \emph{structure preserving rounding} of a multidimensional Gaussian Distribution takes as input a multi-dimensional Gaussian $\mathcal{N}(\m,\S)$ with $\S$ in block-diagonal form.
  It chooses one coordinate as a ``pivot'' in each block, samples from the Gaussian ignoring these pivots and rounds each value to the nearest integer.
  Finally, the pivot coordinate of each block is set by taking the difference between the sum of the means and the sum of the values sampled within the block.
\end{restatable}

\fi
\subsection{Miscellaneous Lemmata}
We will use the following tools for bounding total variation distance between various random variables.

\begin{lemma}[Data Processing Inequality{\ifnum\camr=0 ~for Total Variation Distance\fi}]\label{lem:DPI}
  Let $X, X'$ be two random variables over a domain $\Omega$.
  Fix any (possibly randomized) function $F$ on $\Omega$ (which may be viewed as a distribution over deterministic functions on $\Omega$) and let $F(X)$ be the random variable such that a draw from $F(X)$ is obtained by drawing independently $x$ from $X$ and $f$ from $F$ and then outputting $f(x)$ (likewise for $F(X')$).
  Then we have
  $$\dtv\left(F(X),F(X')\right) \leq \dtv\left(X,X'\right).$$
\end{lemma}

\begin{proposition}[Berry-Esseen {\ifnum\camr=0 theorem \fi}\cite{Berry41, Esseen42, Shevtsova10}]
\label{prop:berryesseen}
Let $X_1, \dots, X_n$ be independent random variables, with $E[X_i] = 0, E[X_i^2] = \s_i^2 > 0, E[|X_i|^3] = \r_i < \infty$, and define $X = \sum_{i=1}^n X_i, \s^2 = \sum_{i=1}^n \s_i^2, \r = \sum_{i=1}^n \r_i$.
Then for an absolute constant $C_0 \leq 0.56$,
$$\dk(X, \mathcal{N}(0,\s^2)) \leq \frac{C_0 \r}{\s^{3}}.$$
\end{proposition}

\begin{proposition}[Proposition 32 in \cite{ValiantV10a}]\label{prop:gaussapprox}
    Given two $k$-dimensional Gaussians $\mathcal{N}_1 = \mathcal{N}(\m_1,\S_1), \mathcal{N}_2 = \mathcal{N}(\m_2,\S_2)$ such that for all $i,j \in [k]$, $|\S_1(i,j) - \S_2(i,j)| \leq \a$, and the minimum eigenvalue of $\S_1$ is at least $\s^2 \ge \alpha$,
    $$\dtv\left(\mathcal{N}_1, \mathcal{N}_2\right) \leq \frac{\| \m_1 - \m_2\|_2}{\sqrt{2\p\s^2}} + \frac{k\a}{\sqrt{2\p e}(\s^2 - \a)}.$$
\end{proposition}

In addition, we prove the following general purpose lemma showing that two multivariate Gaussians with spectrally-close moments are close in total variation distance.
This is intended to be a multivariate version of Proposition B.4 of \cite{DaskalakisDOST13}, which proves a similar statement for univariate Gaussians.
The proof appears in 
\ifnum\camr=0
Section~\ref{sec:tvgaussian}.
\else
the full version of this paper.
\fi
\begin{lemma}
  \label{lem:dtvgaussian}
  Suppose there exist two $k$-dimensional Gaussians, $X \sim \mathcal{N}(\m_1,\S_1)$ and $Y \sim \mathcal{N}(\m_2,\S_2)$, such that for all unit vectors $v$,
  $$|v^T(\m_1 - \m_2)| \leq \ve s_v,$$
  $$|v^T(\S_1 - \S_2)v| \leq \frac{\ve s_v^2}{2\sqrt{k}};$$
  where  $s_v^2 = \max\{v^T\S_1v, v^T\S_2v\}$.
  Then $\dtv(X,Y) \leq \ve$. 
\end{lemma}

\subsection{Results on PMDs from \cite{DaskalakisKT15}}
\label{sec:prevpmd}
Our work builds upon recent structural results on PMDs \cite{DaskalakisKT15}.
We recall some of the key results which we will refer to in this paper.

Two key parameters used in this paper are $c = c(\ve, k) = \poly(\ve/k)$ and $t = t(\ve,k) = \poly(k/\ve)$, set as $c = \left(\frac{\ve^2}{k^5}\right)^{1 + \d_c}$ and $ t = \left(\frac{k^{19}}{c\ve^6}\right)^{1 + \d_t}$, for constants $\d_c, \d_t > 0$. 

The main tool from this paper we will use is the structural characterization, stating that every PMD is close to the sum of an appropriately discretized Gaussian and a ``sparse'' PMD.

\begin{theorem}[Theorem 5 from \cite{DaskalakisKT15}] \label{thm:struct}
  For parameters $c$ and $t$ as described above, every $(n,k)$-Poisson multinomial random vector is $\ve$-close to the sum of a Gaussian with a structure preserving rounding and a $(tk^2, k)$-Poisson multinomial random vector.
  For each block of the Gaussian, the minimum non-zero eigenvalue of $\S_i$ is at least $\frac{tc}{2k^4}$.
\end{theorem}

Finally, we will also use their rounding procedure, which relates a PMD to a nearby PMD with all parameters either equal to or sufficiently far from $0$ and $1$:
\begin{lemma}[Lemma 1 from \cite{DaskalakisKT15}]\label{lem:round}
  For any $c \leq \frac1{2k}$, given access to the parameter matrix $\r$ for an $(n,k)$-PMD $M^\r$, we can efficiently construct another $(n,k)$-PMD $M^{\hat \r}$, such that, for all $i,j$, $\hat \r(i,j) \not\in (0,c)$, and $\dtv\left(M^{\r},M^{\hat{\r}}\right) < O\left(c^{1/2} k^{5/2} \log^{1/2}\left(\frac{1}{ck}\right)\right)$.
\end{lemma}

\section{A Size-Free CLT}
\label{sec:body_clt}

\label{sec:improved_vv}
\label{sec:cltoverview}
%
%

We overview our proof of Theorem~\ref{thm:newclt}. Recall that the Central Limit Theorem of Valiant and Valiant, \eqref{eq:VV CLT}, has a poly-logarithmic dependence on the size parameter of the GMD. Their work raised the question whether this CLT could be made size-independent, and we resolve this conjecture by showing that it can be. This qualitative improvement comes at a quantitative loss in the polynomial dependence of the bound on the parameters $k$ and $\s^2$. 

Our CLT builds off of the structural result of \cite{DaskalakisKT15}, Theorem \ref{thm:struct}, which we use as a black box. 
This structural result says that every $(n,k)$-PMD is $\ve$-close to the sum of an appropriately discretized Gaussian and a $(\poly(k/\ve),k)$-PMD.
We note that the statement of Theorem \ref{thm:struct} does not tell us anything about the moments of this Gaussian and sparse PMD, while our new CLT requires that the discretized Gaussian has the same moments as the original PMD.
We prove this CLT in two steps.
First, we show that the original PMD $X$ and the discretized Gaussian from the cover $G$ are close in total variation distance, i.e., we show that we can ``drop'' the sparse PMD component from Theorem \ref{thm:struct} in the relevant approximation regime.
Then, we bound the distance between the discretized Gaussian from the cover, $G$, and a discretized Gaussian with the same mean and covariance as the original PMD, $G_X$.
The proof is concluded by combining these two bounds using the triangle inequality.

To bound the distance between the original PMD $X$ and the discretized Gaussian from the cover $G$, we start by invoking Theorem \ref{thm:struct} with parameter $\ve = \poly(k/\s)$.
This tells us that the PMD is close to the sum of a discretized Gaussian with a structure preserving rounding $G$ and a ``sparse'' PMD $P$, which has size parameter at most some $\poly(\s) = o(\s)$.
We first show that the structure preserving rounding only has a single block in its structure.
This is proved by contradiction.
If there were multiple blocks in the structure, there would exist some direction $v$ in which $G$ contributes $0$ variance.
Since $P$ is sparse, it can contribute at most $o(\s)$ variance when projected in direction $v$.
However, we know that $X$ had at least $\s^2$ variance in direction $v$.
By projecting both $X$ and $P$ in direction $v$ and applying Berry-Esseen's theorem, we can show that such a large discrepancy in the variance implies large Kolmogorov distance between the projections, see Proposition~\ref{prop:dtvvarsoff}.
This acts as a certificate demonstrating a large total variation distance, contradicting our invocation of Theorem \ref{thm:struct}, and thus the Gaussian has a single block in its structure.
By a similar contradiction argument, we can also argue that $G$ has a large variance ($\Omega(\s^2)$) when projected in any direction.
Since $G$'s variance is at least $\Omega(\s^2)$ in any direction, while $P$ is only supported over $\{0, \dots, o(\s)\}^k$, it can be shown that $P$'s contribution to the distribution is negligible using Proposition~\ref{prop:subtractsparse}, and thus we can remove it at low cost; i.e. $\dtv(G + P, G)$ is small.
Since Theorem \ref{thm:struct} implied that $\dtv(X, G+P)$ was small, by triangle inequality, we have shown that the original PMD $X$ and the discretized Gaussian from the cover $G$ are close in total variation distance. 

Next, we bound the distance between the discretized Gaussian from the cover, $G$, and a discretized Gaussian with the same moments as the original PMD, $G_X$.
At this point, we know that $X$ and $G$ are close in total variation distance.
By projecting both distributions in some direction and considering true Gaussians with the same moments as $X$ and $G$, it can be shown that the first two moments are similar in this direction -- otherwise, the true Gaussians would be far from each other in the Kolmogorov metric.
This implies that the first two moments of $X$ and $G$ are close in \emph{every} direction, as guaranteed by Proposition~\ref{prop:kdparamclose}.
Applying Lemma \ref{lem:dtvgaussian} tells us that bona-fide Gaussians with moments which are close in every direction are therefore close in total variation distance.
The proof is concluded by applying the Data Processing inequality, which shows that the corresponding discretized Gaussians $G$ and $G_X$ are close as well.

We state and prove many useful lemmas in Section \ref{sec:cltlemmas}, which we combine to complete the proof of Theorem \ref{thm:newclt} in Section \ref{sec:cltproof}.

\subsection{Useful Lemmas}
\label{sec:cltlemmas}
The following two propositions bound the Kolmogorov distance between a univariate Gaussian and the projection of a GMD or a discretized Gaussian, respectively.

\begin{proposition}
  \label{prop:dkgmd}
  Suppose that there exists an $(n,k)$-generalized multinomial random vector $X$, with mean vector $\m$ and covariance matrix $\S$.
  Then for any unit vector $v$, 
  $$\dk(v^TX,\mathcal{N}(v^T\m,v^T\S v)) \leq \frac{1}{\s} , $$
  where $\s^2$ is the minimum eigenvalue of $\S$.
\end{proposition}
\begin{proof}
  We apply the Berry-Esseen theorem (Proposition \ref{prop:berryesseen}).
  Let $Y_i = X_i - E[X_i]$ to recenter the random variables, and we will now compare $Y = \sum_i Y_i$ with $\mathcal{N}(0,v^T\S v)$.
  We note that $v^T Y_i \in [-\sqrt{2},\sqrt{2}]$.
  Letting $\s_i^2 = \var(v^TY_i)$ and $\r_i = E\left[\left|v^TY_i\right|^3\right]$,
  this implies that $\r_i \leq \sqrt{2}\s_i^2$, and thus the Berry-Esseen bound gives
  \ifnum\camr=0
  $$\dk(v^TY,\mathcal{N}(0,v^T\S v)) \leq \frac{0.56\left(\sum_i \r_i\right)}{ \left(\sum_i \s_i^2\right)^{3/2}} \leq \frac{\left( \sum_i \s_i^2\right)}{\left(\sum_i \s_i^2\right)^{3/2}} = \frac{1}{\left(\sum_i \s_i^2\right)^{1/2}} \leq \frac{1}{\sigma}.$$
  \else
  \begin{align*}
    \dk(v^TY,\mathcal{N}(0,v^T\S v)) &\leq \frac{0.56\left(\sum_i \r_i\right)}{ \left(\sum_i \s_i^2\right)^{3/2}} \leq \frac{\left( \sum_i \s_i^2\right)}{\left(\sum_i \s_i^2\right)^{3/2}} \\ &= \frac{1}{\left(\sum_i \s_i^2\right)^{1/2}} \leq \frac{1}{\sigma}.
  \end{align*}
  \fi\end{proof}

\begin{proposition}
  \label{prop:dkdgaussian}
  Suppose there exists a random variable $X \sim \lfloor \mathcal{N}(\m,\S)\rceil$.
  Then for any unit vector $v$,
  $$\dk(v^TX,\mathcal{N}(v^T\m,v^T\S v)) \leq \frac{\sqrt{k}}{\sqrt{2\p}\s} , $$
  where $\s^2$ is the minimum eigenvalue of $\S$.
\end{proposition}
\begin{proof}
  Let $Y \sim \mathcal{N}(\m,\S)$.
  We first show $|v^T(Y - \lfloor Y \rceil)| \leq \frac{\sqrt{k}}{2}$, which holds by Cauchy-Schwarz:
  $\|v\|_2 = 1$ and $\|Y - \lfloor Y \rceil\|_2 \leq \sqrt{k} \cdot \|Y - \lfloor Y \rceil\|_\infty \leq \frac{\sqrt{k}}{2}$. 
  Thus,
  $$v^TY - \frac{\sqrt{k}}{2} \leq v^T \lfloor Y \rceil \leq v^TY + \frac{\sqrt{k}}{2}.$$
  Using $F$ to denote the corresponding CDFs, this stochastic dominance condition implies that for any $y \in \mathbb{R}$,
  $$F_{v^TY - \frac{\sqrt{k}}{2}}(y) \leq F_{v^T \lfloor Y \rceil}(y) \leq F_{v^TY + \frac{\sqrt{k}}{2}}(y).$$
  Furthermore, 
  $$F_{v^TY - \frac{\sqrt{k}}{2}}(y) \leq F_{v^T Y }(y) \leq F_{v^TY + \frac{\sqrt{k}}{2}}(y)$$
  and
  $$F_{v^TY + \frac{\sqrt{k}}{2}}(y) - F_{v^TY - \frac{\sqrt{k}}{2}}(y) \leq \sqrt{k} \cdot \frac{1}{\sqrt{2\p}\s},$$
  because the two distributions are univariate Gaussians with the same variance (which is at least $\s^2)$ and means shifted by $\sqrt{k}$.
  This implies 
  $$|F_{v^TY}(y) - F_{v^T\lfloor Y \rceil}(y) | \leq \frac{\sqrt{k}}{\sqrt{2\p}\s},$$
  as desired.
\end{proof}

The following proposition compares a Gaussian $X$ and an arbitrary distribution $Y$.
It shows that if $Y$'s variance is much smaller than $X$'s, then they must be far in Kolmogorov distance.
\begin{proposition}
  \label{prop:dtvvarsoff}
  Suppose there exists a univariate Gaussian $X$ with variance $\s_X^2$, and a distribution $Y$ with variance $\s_Y^2 < \s_X^2$.
  Then the Kolmogorov distance between $X$ and $Y$ is at least $\frac12 - \left(\frac{\s_Y}{\s_X}\right)^{2/3}$.
\end{proposition}
\begin{proof}
  We consider the event that a sample falls in an interval of width $2k$ centered at $E[Y]$.
  As a certificate of a large Kolmogorov distance between $X$ and $Y$, we show that the probability assigned to this interval is very different for $X$ versus $Y$.

  First, by Chebyshev's inequality, we know that
  $$\Pr\left[\left|Y - E[Y]\right| \leq k\right] \geq 1 - \frac{\s_Y^2}{k^2}.$$

  On the other hand, we know that 
  \ifnum\camr=0
  $$\Pr\left[\left|X - E[Y]\right| \leq k\right] \leq \Pr\left[\left|X - E[X]\right| \leq k\right] = \erf\left(\frac{k}{\sqrt{2}\s_X}\right) 
  \leq \frac{k}{\sqrt{2\p}\s_X},$$
  \else
  \begin{align*}
  \Pr\left[\left|X - E[Y]\right| \leq k\right] &\leq \Pr\left[\left|X - E[X]\right| \leq k\right] \\
  &= \erf\left(\frac{k}{\sqrt{2}\s_X}\right) \leq \frac{k}{\sqrt{2\p}\s_X},
  \end{align*}
  \fi
  where the last inequality uses the Taylor expansion of the error function.

  The difference in probability assigned to this interval is at least
  $$1 - \frac{\s_Y^2}{k^2} - \frac{k}{\sqrt{2\p}\s_X}.$$
  Setting $k = \s_Y^{2/3} \s_X^{1/3}$ gives
  \ifnum\camr=0
    $$\dk(X,Y) \geq \frac12\left(1 - \left(\frac{\s_Y}{\s_X}\right)^{2/3} - \frac{1}{\sqrt{2\p}}\left(\frac{\s_Y}{\s_X}\right)^{2/3}\right) \geq \frac12 - \left(\frac{\s_Y}{\s_X}\right)^{2/3}, $$
  \else
  \begin{align*}
  \dk(X,Y) &\geq \frac12\left(1 - \left(\frac{\s_Y}{\s_X}\right)^{2/3} - \frac{1}{\sqrt{2\p}}\left(\frac{\s_Y}{\s_X}\right)^{2/3}\right) \\
           &\geq \frac12 - \left(\frac{\s_Y}{\s_X}\right)^{2/3}, 
  \end{align*}
  \fi

  as desired.
\end{proof}

The following proposition tells us if we are considering the sum of two random variables, one being a Gaussian with a large variance and one being an arbitrary distribution with a small support, we can remove all contribution from the distribution with small support and not pay a large cost in total variation distance.
\begin{proposition}
  \label{prop:subtractsparse}
  Suppose $X$ and $Y$ are independent random variables, where $X \sim \lfloor\mathcal{N}(\m,\S)\rceil \in \mathbb{R}^k$ and $Y$ is supported on $S = \{0,\dots,m\}^k$.
  Then $\dtv(X, X+Y) \leq \frac{m\sqrt{k}}{\sqrt{2\p}\s}$, where $\s$ is the minimum eigenvalue of $\S$.
\end{proposition}
\begin{proof}
  We start by applying a law of total probability for total variation distance:
  \ifnum\camr=0
  $$\dtv(X, X+Y) \leq \sum_{v \in S} \Pr(Y = v) \dtv(X, X + v) = \sum_{v \in S} \Pr(Y = v) \dtv(\lfloor\mathcal{N}(\m,\S)\rceil,\lfloor\mathcal{N}(\m + v,\S)\rceil).$$
  \else
  \begin{align*}
    &\dtv(X, X+Y) \\    
    &\leq \sum_{v \in S} \Pr(Y = v) \dtv(X, X + v) \\ 
                 &= \sum_{v \in S} \Pr(Y = v) \dtv(\lfloor\mathcal{N}(\m,\S)\rceil,\lfloor\mathcal{N}(\m + v,\S)\rceil).
  \end{align*}  
  \fi
  Using the data processing inequality for total variation distance (Lemma \ref{lem:DPI}): 
  \ifnum\camr=0
    $$ \dtv(\lfloor\mathcal{N}(\m,\S)\rceil,\lfloor\mathcal{N}(\m + v,\S)\rceil) \leq \dtv(\mathcal{N}(\m,\S),\mathcal{N}(\m + v,\S)) 
 \leq \frac{\|v\|}{\sqrt{2\pi}\s} \leq \frac{m\sqrt{k}}{\sqrt{2\p}\s},$$
  \else
    \begin{align*}
    \dtv(\lfloor\mathcal{N}(\m,\S)\rceil,\lfloor\mathcal{N}(\m + v,\S)\rceil) &\leq \dtv(\mathcal{N}(\m,\S),\mathcal{N}(\m + v,\S)) \\
 &\leq \frac{\|v\|}{\sqrt{2\pi}\s} \leq \frac{m\sqrt{k}}{\sqrt{2\p}\s},
    \end{align*}
  \fi 
 where the second last inequality follows from Proposition \ref{prop:gaussapprox}. 
 We conclude by observing that $\dtv(X, X + Y)$ is a convex combination of such terms.
\end{proof}

The next proposition tells us that Kolmogorov closeness implies parameter closeness for univariate Gaussians.
\begin{proposition}
  \label{prop:1dparamclose}
  Consider two univariate Gaussians $X \sim \mathcal{N}(\m_1, \s_1^2)$ and $Y \sim \mathcal{N}(\m_2, \s_2^2)$ where $\s_1 \leq \s_2$.
  For any $\a \in (0,1)$, if $\dk(X,Y) \leq \frac{\a}{10}$, then $|\m_2 - \m_1| \leq \a\s_1$ and $|\s_2^2 - \s_1^2| \leq 3\a\s_1^2$ .
\end{proposition}
\begin{proof}
  We start by proving the following statement:
  For any $\a \in (0,1)$, if $|\m_2 - \m_1| \geq \a\s_1$ or $|\s_2 - \s_1| \geq \a\s_1$ , then $\dk(X,Y) \geq \frac{\a}{10}$.
  The proof follows by contraposition, and observing that multiplying both sides of $|\s_2 - \s_1| \leq \a\s_1$ by $(\s_2 + \s_1)$, bounding $\s_2 \leq (1 + \a)\s_1$, and $\a \leq 1$ imply $|\s_2^2 - \s_1^2| \leq 3\a\s_1^2$.

  Without loss of generality, assume $\m_1 \leq \m_2$.
  We will first show the conclusion assuming the means are separated, and then assuming the variances are separated.

  Suppose $|\m_2 - \m_1| \geq \a\s_1$.
  Consider the point $x = \m_2$.
  At this point, the CDF of the second Gaussian is equal to $\frac{1}{2}$. 
  The CDF of the first Gaussian is $\frac{1}{2}\left(1 + \erf\left(\frac{\m_2 - \m_1}{\sqrt{2} \s_1}\right)\right) \geq \frac12\left(1 + \erf\left(\frac{\a}{\sqrt{2}}\right)\right)$.
  Therefore, $\dk(\mathcal{N}_1, \mathcal{N}_2) \geq \frac12 \erf\left(\frac{\a}{\sqrt{2}}\right) \geq \frac{\a}{10}$, where the last inequality holds for all $\a \in (0,1)$.

  Now, suppose $|\s_2 - \s_1| \geq \a\s_1$.
  Consider the point $x = \m_1 + \sqrt{2}\s_1$. 
  At this point, the CDF of the first Gaussian equal to $\frac{1}{2}(1 + \erf(1))$. 
  Similarly, the CDF of the second Gaussian is at most $\frac{1}{2}\left(1 + \erf\left(\frac{\s_1}{\s_2}\right)\right) \leq \frac{1}{2}\left(1 + \erf\left(\frac{1}{1 + \a}\right)\right)$.
  Therefore, $\dk(\mathcal{N}_1, \mathcal{N}_2) \geq \frac{\erf(1) - \erf\left(\frac{1}{1+\a}\right)}{2} \ge \frac \a {10}$ where the last inequality holds for all $\a \in (0,1)$.
\end{proof}

Our final proposition in this section applies the previous proposition, showing that total variation closeness implies parameter closeness (in any projection) when considering a GMD and a discretized Gaussian.
\begin{proposition}
  \label{prop:kdparamclose}
  Suppose $X$ is an $(n,k)$-GMD, and $Y$ is a $k$-dimensional discretized Gaussian such that $\dtv(X,Y) \leq \a$.
  Let $\m_X$ and $\S_X$ be the mean vector and covariance matrix (respectively) of $X$, and define $\m_Y$ and $\S_Y$ similarly for $Y$.
  For a unit vector $v$, let $\s_v^2 = \min\{v^T\S_Xv, v^T\S_Yv\}$, and let $\s^2 = \min_v \s_v^2$.
  If $\a + \frac{2\sqrt{k}}{\s} \leq 1/10$, then for all unit vectors $v$
  $$|v^T(\m_X - \m_Y)| \leq 10\left(\a + \frac{2\sqrt{k}}{\s}\right)\s_v;$$
  $$|v^T(\S_X - \S_Y)v| \leq 30\left(\a + \frac{2\sqrt{k}}{\s}\right)\s_v^2.$$
\end{proposition}
\begin{proof}
Consider the projections of $X$ and $Y$ onto $v$.
By Propositions \ref{prop:dkgmd} and \ref{prop:dkdgaussian} and the triangle inequality, the Kolmogorov distance between the univariate Gaussians with the same mean and variance is at most $\a + \frac{2\sqrt{k}}{\s}$.
Applying Proposition \ref{prop:1dparamclose} implies the desired result.
\end{proof}

\subsection{Proof of Theorem \ref{thm:newclt}}
\label{sec:cltproof}
  We will prove the statement for a sufficiently large constant $C$.
  Thus we only need examine the case 
  \begin{equation}
    \label{eq:cassn}
    \frac{k^3}{\s^{1/10}} \leq \frac{1}{C},
  \end{equation}
  otherwise the conclusion of the theorem statement is vacuous since total variation distance is at most $1$.

  As a starting point, we convert from a GMD to the corresponding $(n,k)$-Poisson multinomial random vector $X$ and apply Theorem \ref{thm:struct} with $\ve = \frac{k^3}{\s^{1/10}}$.
  This gives us that
  $$\dtv(X, G + P) \leq \frac{k^3}{\s^{1/10}},$$
  where $G$ is a Gaussian with a structure preserving rounding and $P$ is a $(tk^2,k)$-Poisson multinomial random vector.
  By the definition of $t$ in Section~\ref{sec:prevpmd}, we have that $t \leq \frac{C'\s^{9/10}}{k^2}$ for some constant $C'$.
  Thus, $P$ is a $(C'\s^{9/10}, k)$-Poisson multinomial random vector.
    
  First, we argue that the Gaussian component $G$ only has a single block in its structure.
  We prove this by contradiction -- suppose there exist multiple blocks in its structure.
  Let one of the pivots be the pivot coordinate for the GMD, and ignore this dimension.
  If there are multiple blocks, the rounding procedure implies that there exists a direction $v$ in which the variance of the resulting covariance matrix of the Gaussian is $0$.
  In direction $v$, the maximum possible value for the variance of $P$ is $\frac{C'\s^{9/10}}{4}$, giving us an upper bound for the variance of $G + P$.
  However, we know that the variance of $X$ in direction $v$ is at least $\s^2$, by the assumption in the theorem statement.
  By Proposition \ref{prop:dkgmd}, projecting $X$ in direction $v$ and converting to a univariate Gaussian $X_g$ with the same mean and variance incurs a cost of at most $\frac{1}{\s}$ in Kolmogorov distance.
  Also projecting $G + P$ in direction $v$, Proposition \ref{prop:dtvvarsoff} tells us that $\dk(v^TX, v^T(G + P)) \geq \dk(X_g, v^T(G+P)) - \dk(v^TX, X_g) \geq \frac12 - \left(\frac{C'}{4\s^{11/10}}\right)^{1/3} - \frac{1}{\s}$.
  Because $\s \geq C^{10}$ (as assumed in (\ref{eq:cassn})), we have that $\dk(v^TX, v^T(G + P)) > \frac{1}{3}$.
  Since we know $\dtv(X, G + P) \leq \frac{k^3}{\s^{1/10}}  $, this implies that $\dk(v^TX, v^T(G + P)) \leq \dtv(v^TX, v^T(G + P)) \leq \frac{k^3}{\s^{1/10}} $ should also hold, which is a contradiction for large $C$, as $\frac{k^3}{\s^{1/10}} \leq \frac{1}{C} < \frac{1}{3}$.
  Therefore, the Gaussian component $G$ only has a single block in its structure.
  
  Since we have established that the Gaussian component $G$ only has a single block, we will convert back to the original GMD domain for the remainder of the proof.
  Recall that the original GMD is $M^\r$, and we let $D$ be the discretized Gaussian and $S$ be the $(C'\s^{9/10},k)$-Generalized multinomial random vector with the same pivot coordinate as $M^\r$.
  Now, we wish to upper bound $\dtv(M^\r, D)$, i.e., we want to eliminate the sparse GMD from our statement.
  First, we wish to argue that $D$ has a large variance in every direction, and thus removing $S$ will not have a large effect.
  This is done by the same method in the above paragraph.
  Let the minimum variance of $D$ in any direction be $\z^2$.
  Then to avoid the same contradiction as above, we require that
  $$\frac12 - \left(\frac{ \frac{C'\s^{9/10}}{4}+ \z^2}{\s^2}\right)^{1/3} - \frac{1}{\s}\leq \frac{1}{C}.$$
  This can be manipulated to show that 
  \begin{equation}
    \label{eq:vargaussian}
    \z^2 \geq \frac{1}{16} \sigma^2.
  \end{equation}
  Now, applying Proposition \ref{prop:subtractsparse} and the triangle inequality, we get
  \begin{equation}
    \label{eq:gmddgclose}
    \dtv(M^\r, D) \leq \frac{k^3}{\s^{1/10}} + \frac{4C'\sqrt{k}}{\sqrt{2\p}\s^{1/10}}.
  \end{equation}
  
  Finally, to conclude, we must compare $D$ with a discretized Gaussian with the same moments as $M^\r$, i.e., we wish to upper bound $\dtv(D, \lfloor \mathcal{N}(\m,\S)\rceil)$.
  Recall that $\m$ and $\S$ are the mean and covariance of $M^\r$, and let $\m_D$ and $\S_D$ be the mean and covariance of $D$.
  Apply Proposition \ref{prop:kdparamclose} to $M^\r$ and $D$ using the guarantees of Equations (\ref{eq:vargaussian}) and  (\ref{eq:gmddgclose}).
  This implies that their moments are close:
  $$|v^T(\m - \m_D)| \leq 10\left(\frac{k^3}{\s^{1/10}} + \frac{4C'\sqrt{k}}{\sqrt{2\p}\s^{1/10}} +  \frac{8\sqrt{k}}{\s}\right)\s_v;$$
  $$|v^T(\S - \S_D)v| \leq 30\left(\frac{k^3}{\s^{1/10}} + \frac{4C'\sqrt{k}}{\sqrt{2\p}\s^{1/10}} +  \frac{8\sqrt{k}}{\s}\right)\s_v^2,$$
  where $\s_v^2 = \min\{v^T \S v, v^T\S_D v\}$.

  We use the Data Processing Inequality (Lemma \ref{lem:DPI}) followed by Lemma \ref{lem:dtvgaussian} with these guarantees to give:
  \ifnum\camr=0
    $$\dtv\left(D, \lfloor \mathcal{N}(\m, \S)\rceil\right) \leq \dtv\left(\mathcal{N}(\m_D,\S_D), \mathcal{N}(\m, \S)\right)
  \leq 60\left(\frac{k^3}{\s^{1/10}} + \frac{4C'\sqrt{k}}{\sqrt{2\p}\s^{1/10}} +  \frac{8\sqrt{k}}{\s}\right)\sqrt{k}.$$
  \else
  \begin{align*}
    \dtv\left(D, \lfloor \mathcal{N}(\m, \S)\rceil\right) &\leq \dtv\left(\mathcal{N}(\m_D,\S_D), \mathcal{N}(\m, \S)\right) \\
  &\leq 60\left(\frac{k^3}{\s^{1/10}} + \frac{4C'\sqrt{k}}{\sqrt{2\p}\s^{1/10}} +  \frac{8\sqrt{k}}{\s}\right)\sqrt{k}.
  \end{align*}
  \fi
  
  Finally, applying the triangle inequality with Equation (\ref{eq:gmddgclose}) gives
  \ifnum\camr=0
  $$\dtv\left(M^\r, \lfloor \mathcal{N}(\m, \S)\rceil\right) \leq
  \dtv(M^\r, D) \leq 61\left(\frac{k^{7/2}}{\s^{1/10}} + \frac{4C'k}{\sqrt{2\p}\s^{1/10}} +  \frac{8k}{\s}\right).$$
  \else
  \begin{align*}
  \dtv\left(M^\r, \lfloor \mathcal{N}(\m, \S)\rceil\right) &\leq\dtv(M^\r, D) \\
  &\leq 61\left(\frac{k^{7/2}}{\s^{1/10}} + \frac{4C'k}{\sqrt{2\p}\s^{1/10}} +  \frac{8k}{\s}\right).
  \end{align*}
  \fi
  Choosing the constant $C$ sufficiently large completes the proof.

\section{A PTAS for Anonymous Games}
\label{sec:body_anon}
Here, we overview the algorithm of Theorem~\ref{thm:PTAS for anonymous}.
The algorithm starts with a guess of $X=\sum_i X_i$ at a Nash equilibrium $V_X=(X_1,\ldots,X_n)$ of the game, where $X_i$ represents the mixed strategy of player $i$. While there are infinitely many $X$'s to guess, our proper cover theorem {(Theorem~\ref{thm:propercover})} implies that every $X$ can be approximated by some $Y=\sum_iY_i$, where $V_Y=(Y_1,\ldots,Y_n) \in S_\epsilon$, $\dtv(X,Y) \le \epsilon$ and $|S_\epsilon|$ is of the order of at most~\eqref{eq:runtime anonymous}. What we would like to claim is that if $Y$ approximates $X$, then $V_Y$ is an approximate Nash equilibrium of the game up to a permutation of the $Y_i$'s. This is unfortunately not necessarily true, but the following guarantees would suffice:
\ifnum\camr=0
\begin{align}\forall i: {\rm support}(Y_i) \subseteq {\rm support}(X_i)~~\wedge~~\dtv\left(\sum_{j\neq i}X_j,\sum_{j\neq i}Y_j\right)\le \epsilon. \label{eq:ideal guarantee}
\end{align}
\else
\begin{align}\forall i: {\rm support}(Y_i) \subseteq {\rm support}(X_i) \nonumber \\~~\wedge~~\dtv\left(\sum_{j\neq i}X_j,\sum_{j\neq i}Y_j\right)\le \epsilon. \label{eq:ideal guarantee}
\end{align}
\fi
Indeed, if the above guarantee held, then the expected payoff of every player $i$ from any pure strategy $\sigma$ would not change by more than an additive $O(\epsilon)$ if we changed the strategies of all other players from $(X_j)_{j\neq i}$ to $(Y_j)_{j \neq i}$. So, if $V_X$ were a Nash equilibrium and ${\rm support}(Y_i) \subseteq {\rm support}(X_i)$, it would follow that $Y_i$ is an approximate best response of player $i$ to $(Y_j)_{j \neq i}$. So $V_Y$ would be an approximate equilibrium.

Unfortunately, we do not know how to construct a proper $\epsilon$-cover $S_\epsilon$ of all $(n,k)$-PMDs that has size of order~\eqref{eq:runtime anonymous} and such that for any $V_X$ there exists some $V_Y \in S_\epsilon$ satisfying Condition~\eqref{eq:ideal guarantee}. Nevertheless, we can exploit our CLT and the structural result of~\cite{DaskalakisKT15} (restated as Theorem~\ref{thm:struct} in this paper) to bypass this difficulty. Roughly speaking~\cite{DaskalakisKT15} approximate a given $V_X=(X_1,\ldots,X_n)$ by first discretizing the parameters of all $X_i$'s into fine enough accuracy (this is shown to only cost some $O(\epsilon)$ in total variation distance), then partitioning the $X_i$'s into a small group ${\cal L}$ of size ${\rm poly}(k/\epsilon)$ that are left intact, and a large group whose sum is approximated by a discretized multidimensional Gaussian (up to another cost of $O(\epsilon)$ in total variation distance).  It is further shown that the distribution of the sum of variables in ${\cal L}$ can be summarized through the vector $\vec{m}$ of its first $O(\log{1/\epsilon})$ moments (at a loss of an additional $O(\epsilon)$ in total variation distance), while the discretized Gaussian through its first two moments $(\m,\S)$. Moreover, it is shown that the Gaussian has at least ${\rm poly}(k/\epsilon)$ variance in all directions where it has non-zero variance. By enumerating over all possible summary statistics $(\vec{m},\m,\S)$, a non-proper cover of all $(n,k)$-PMDs can be obtained, whose size is of the order of~\eqref{eq:runtime anonymous}.

Suppose now that $V_X=(X_1,\ldots,X_n)$ is a Nash equilibrium whose approximating statistic in the non-proper cover is some $(\vec{m},\mu, \Sigma)$. Given a correct guess for this statistic, our goal is to uncover an approximate Nash equilibrium $V_Y=(Y_1,\ldots,Y_n)$ of the game.  By the construction of the cover, we know that every player $i$ either contributed his discretized $X_i$ to the discretized Gaussian with parameters $(\m,\S)$, or to the small group of variables with moments $\vec{m}$. So, letting  ${\cal C}$ be the set of $k$-CRVs whose parameters have the discretization accuracy used in the construction of the cover, we need to assign some $Y_i \in {\cal C}$ to each player $i$ such that:
\begin{enumerate}
\item[(a)] There exists a ${\rm poly}(k/\epsilon)$-size subset ${\cal L}$ of players such that $\sum_{i \in {\cal L}}Y_i$ has vector of moments $\vec{m}$, while $\sum_{i \notin {\cal L}}Y_i$ has first two moments $(\m,\S)$.
\item[(b)] For all $i$, $Y_i$ is a best response to $\sum_{j \neq i} Y_j$.
\end{enumerate}

To find a good assignment, we first construct a compatibility graph between players and mixed strategies in ${\cal C}$. We add an edge between some $i$ and some $Y_i \in {\cal C}$ iff at least one of the following two conditions is met. We also annotate the edge with all conditions that are met:
\begin{enumerate}
\item ($Y_i$ is compatible with $i \in {\cal L}$): $Y_i$ is an approximate best response to the ``environment'' $i$ would observe if $i$ contributed to $\vec{m}$. If $i$ contributed to $\vec{m}$ and Condition (a) were met, then we can deduce what PMD player $i$ would see in his environment. Indeed, this would be within some $O(\epsilon)$ in total variation distance to a the sum of a Gaussian random vector with parameters $(\mu,\Sigma)$ and a PMD whose first $O(\log (1/\epsilon))$ moments are the same as $\vec{m}$ after removing the contribution of $Y_i$. The updated moment vector can be computed from $\vec{m}$ and $Y_i$ as moments are symmetric polynomials of the underlying parameters. Given the updated moment vector, the PMD is determined to within $\epsilon$ in total variation distance, so its sum with the discretized Gaussian is also determined, and we can also efficiently determine whether $Y_i$ is an approximate best response of player $i$ to that distribution.

\item ($Y_i$ is compatible with $i \in \bar{{\cal L}}$): $Y_i$ is an approximate best response to the ``environment'' $i$ would observe if $i$ contributed to the discretized Gaussian with parameters $(\m,\S)$. First, for this to be the case $Y_i$ must be ``compatible'' with $\S$, i.e. not correlating uncorrelated pairs of dimensions/adding variance in zero-variance dimensions (or in other words, the block structure of $\S$ should be preserved). Moreover, since all non-zero eigenvalues of $\S$ are at least ${\rm poly}(k/\epsilon)$-large, the discretized Gaussian with parameters $(\m,\S)$ and $(\m-\E[Y_i],\S-{\rm \bf cov}(Y_i))$ are approximately the same{\ifnum\camr=0 (Proposition~\ref{prop:gaussapprox})\fi}. At the same time, due to the largeness of the non-zero eigenvalues of $\Sigma$, if condition (a) were eventually true, then our CLT (Theorem~\ref{thm:newclt}) would imply that $\sum_{j \in \bar{{\cal L}}\setminus \{i\}}Y_j$ is well-approximated by the discretized Gaussian with parameters $(\m-\E[Y_i],\S-{\rm \bf cov}(Y_i))$, and hence by that with parameters $(\m,\S)$. So, if $i \in \bar{\cal L}$, $i$ is assigned $Y_i$, and Condition (a) is eventually met, then the PMD that player $i$ sees in his environment is pinned down to within $O(\epsilon)$ in total variation distance: it is approximately the sum of the discretized Gaussian with parameters $(\m,\S)$ and a PMD with moments $\vec{m}$. We can therefore check if $Y_i$ is an approximate best response to that distribution.
\end{enumerate}

After constructing the compatibility graph as above, we need to see if there is an assignment of players to compatible mixed strategies from $\cal C$ so that (a) is satisfied. This looks non-trivial, but it can be done using dynamic programming. We sweep through the players, maintaining as state all possible leftover moments $(\vec{m}', \m',\S')$ that may arise from assignments of a prefix of players to compatible mixed strategies. Given the discretization of $\cal C$, the set of possible states is bounded by~\eqref{eq:runtime anonymous}. Importantly, the compatibility graph has the property that player $i$ is happy when given a compatible strategy as long as the overall assignment matches $(\vec{m}, \m, \S )$.
\ifnum\tp=1
For more details on our PTAS, see 
\ifnum\camr=0
Section \ref{sec:anon}.
\else
the full version of this paper.
\fi
\fi

\ifnum\tp=1
\section{Details from Section \ref{sec:body_anon}}
\label{sec:anon}
\fi
\subsection{Preliminaries for Anonymous Games}

\begin{definition}
 An {\em anonymous game} is a triple
$G=(n,k,\{u^i_j\})$ where $[n]=\{1,\ldots,n\}$, $n\geq 2$, is the set
of players, $[k]=\{1,\ldots,k\}$, $k\geq 2$, is the set of strategies,
and $u^i_j$ with $i \in [n]$ and $j\in[k]$ is the utility of player
$i$ when she plays strategy $j$, a function mapping the set of
partitions $\Pi^k_{n-1}=\{(x_1,\ldots,x_k): x_i \in \mathbb{N}_0
\hbox{{\rm~for all~}} i\in[k], \sum_{i=1}^k x_i = n-1\}$ to the
interval $[0,1]$.
\end{definition}

A {\em mixed strategy profile} $\r$ is a set of $n$ distributions $\{\r_i \in \Delta^k\}_{i\in[n]}$, where by $\Delta^k$ we denote the $(k-1)$-dimensional simplex, or, equivalently, the set of distributions over $[k]$.
A mixed strategy profile $\r$ is an {\em $\ve$-approximately well supported Nash equilibrium} (or an {\em $\ve$-Nash equilibrium}, for short) if, for all $i \in [n]$ and $j\in[k]$,
$$E_{x \sim \r_{-i}} [ u^i_j(x)]  < \max_{j' \in [k]} E_{x \sim \r_{-i}} [ u^i_{j'}(x) ] - \ve \Rightarrow \r_{ij}=0,$$
where $\r_{-i}$ is the distribution over $\Pi^k_{n-1}$ obtained by
drawing $n-1$ random samples from $[k]$ independently according to
the distributions $\r_{i'}, i'\neq i$, and forming the induced
partition. 
We note that this an $\ve$-Nash equilibrium is stronger than the related concept of an $\ve$-approximate Nash equilibrium (see, i.e., \cite{DaskalakisGP09} for further discussion of this distinction).
Throughout this paper, we solely consider the harder problem of computing an $\ve$-Nash equilibrium.

A $0$-Nash equilibrium is simply called a {\em Nash equilibrium} and it is always guaranteed to exist by Nash's theorem.

\subsection{An Algorithm for Anonymous Games}

In a Nash equilibrium $\r$ of an anonymous game every player uses a mixed strategy $\r_i$ selecting strategy $j$ with probability $\r_{ij}$. The distribution of the number of players which select each of the strategies is an $(n,k)$-PMD.
Using the fact that there exist small size $\ve$-covers for PMDs, we can efficiently search over the space of all strategies and identify a mixed strategy profile that produces an $\ve$-Nash equilibrium. 
We show that there exists an efficient polynomial time approximation scheme (EPTAS) for computing an $\ve$-Nash equilibrium, thus proving Theorem \ref{thm:PTAS for anonymous}.

The algorithm works by guessing an aggregate statistic $(m,\m,\S)$ that describes the overall behavior of all players. This statistic is based on the structural theorem shown in~\cite{DaskalakisKT15}, which shows that the overall PMD that describes the mixed strategy profile can be approximately written as sum of a discretized Gaussian and a sparse PMD with only $\poly(k/\ve)$ components. Moreover, for the sparse PMD knowledge of the $\log(1/\ve)$ moments (which is equivalent to knowing the power-sums of all the summands up to $\poly(1/\ve)$, suffices to describe it within $\ve$ in total variation distance. Thus, the algorithm requires guessing the power-sums $m$ of the sparse PMD and the mean $\m$ and covariance $\S$ of the discretized Gaussian.

As we will show, knowledge of an individual's strategy together with the aggregate statistic $(m,\m,\S)$ for the overall mixed strategy profile, allows us to compute an approximate distribution $D_i$ that discribes the player's view about the aggregate strategy of everyone else. If we manage to assign strategies $\r_i$ to every player so that $\r_{-i}$ approximately matched $D_i$ and additionally each player only chooses strategies that corresponds to approximate best responses with respect to his view $D_i$ we will obtain an $\ve$-Nash equilibrium. The following lemma formalizes this intuition and is the main tool we use in the proof of Theorem~\ref{thm:PTAS for anonymous}.



\begin{lemma}\label{lem:epsnash}
Consider the anonymous game $G=(n,k,\{u^i_j\})$ and let $D_1, D_2,...,D_n$ be arbitrary distributions over $\mathbb{Z}^k$. If there exists an $(n,k)$-PMD $\r$ such that:
\begin{itemize}
  \item For all $i \in [n]$, $\dtv(\r_{-i},D_{i}) \le \ve_1$
  \item For all $i \in [n]$ and $j \in [k]$, $E_{x \sim D_{i}} [ u^i_j(x)]  < \max_{j' \in [k]} E_{x \sim D_{i}} [ u^i_{j'}(x) ] - \ve_2 \Rightarrow \r_{ij}=0,$
  \end{itemize}
Then, $\r$ is an $(2 \ve_1 + \ve_2)$-Nash equilibrium for the game $G$.
\end{lemma}

\begin{proof}
  For any $i \in [n]$ and $j \in [k]$, we have that $| E_{x \sim D_{i}} [ u^i_j(x)] - E_{x \sim \r_{-i}} [ u^i_j(x)] | \le \ve_1$, since $i \in [n]$, $\dtv(\r_{-i},D_{i}) \le \ve_1$.
  Therefore,
\begin{align*}
  \max_{j' \in [k]} E_{x \sim \r_{-i}} [ u^i_{j'}(x) ] - E_{x \sim \r_{-i}} [ u^i_j(x)]  &>  \ve_2 + 2 \ve_1 \Rightarrow \\
    \max_{j' \in [k]} E_{x \sim D_{i}} [ u^i_{j'}(x) ] - E_{x \sim D_{i}} [ u^i_j(x)]  &>  \ve_2 
\Rightarrow \r_{ij} = 0
\end{align*}
\end{proof}

\begin{proof}[Proof of Theorem~\ref{thm:PTAS for anonymous}]
 Consider the game $G=(n,k,\{u^i_j\})$. By Nash's theorem there always exists a Nash equilibrium. Let $\r$ be such an equilibrium where every player uses a mixed strategy $\r_i$ selecting strategy $j$ with probability $\r_{ij}$. The distribution of vectors which give the number of players which select each of the strategies is an $(n,k)$-PMD. 
 
 To get an efficient algorithm, we need to search over a restricted set of strategies for each player. To be able to do that we must show that an $\ve$-Nash equilibrium exists in a more restricted space. To argue that, we begin by a Nash equilibrium $\r$ and perform a series of operations that maintain the property that the resulting mixed strategy profile is an $\ve$-equilibrium.
\begin{enumerate}
  \item
    We first proceed by rounding the probabilities $\r_{ij}$ so that they are either $0$ or at least $c$ as done in Lemma~\ref{lem:round}. This gives a PMD $\r^{(1)}$ that is $O( c^{1/2} k^{5/2} \log^{1/2}(1/ck))$-close in total variation distance to $\r$ . Moreover, if we consider the PMD $\r^{(1)}_{-i}$, which is the $(n-1,k)$-PMD obtained by removing the $i$-th component from the rounded PMD $\r^{(1)}$, this is also $O( c^{1/2} k^{5/2} \log^{1/2}(1/ck))$-close in total variation to $\r_{-i}$, i.e. the PMD obtained after removing the $i$-th component from the original PMD $\r$. The proof of this statement is almost identical to the proof in~\cite{DaskalakisKT15} and is omitted. That proof uses Poisson approximations to bound the total variation between the rounded and the unrounded PMDs and uses the fact that the means of the two PMDs can differ by at most $c$ in each coordinate. The only difference is that here, the means of the two PMDs can differ by at most $2 c$ in each coordinate which results in the same asymptotic bound for total variation distance.
Moreover, note the rounding procedure doesn't change any probabilities that were originally 0, i.e. $\r_{ij} = 0 \Rightarrow \r^{(1)}_{ij} = 0$. 
\item \label{discstep}
We now discretize all parameters $\r^{(1)}_{ij}$ into multiples of $\lceil \frac {n k} \ve \rceil^{-1}$ to get a new PMD $\r^{(2)}$. This preserves the support of every CRV and makes sure that parameters that were at least $c$ originally remain at least $c-\frac {\ve} {n k}$. Moreover, since $|\hat r_{ij} - \bar r_{ij}| < \frac {\ve} {n k}$, it holds that $\dtv(\r^{(1)}_i,\r^{(2)}_i) < \ve/n$ which implies that $\dtv(\r^{(1)}_{-i},\r^{(2)}_{-i}) < \ve$. This means that overall, for all $i \in [n]$ and $j \in [k]$, $\dtv(\r_{-i},\r^{(2)}_{-i}) < \ve + O( c^{1/2} k^{5/2} \log^{1/2}(1/ck)) < 2\ve$ and $\r_{ij} = 0 \Rightarrow \r^{(2)}_{ij} =0$.
\item By the structural theorem of~\cite{DaskalakisKT15}, the components $\r^{(2)}_i$ of the PMD $\r^{(2)}$ can be partitioned into two PMDs:
\begin{itemize}
  \item a sparse PMD of size $t k^2$: As in step~\ref{discstep}, we can discretize all its probabilities into multiples of $\lceil \frac {t k^3} {\ve} \rceil^{-1}$ to obtain a PMD $\r^{\sparse}$ that is $\ve$-close in total variation distance.
  \item a large PMD of size $n-t k^2$: This PMD $\r^{\lrg}$ is shown in~\cite{DaskalakisKT15} to be approximable within $\ve$ in total variation distance by a discretized Gaussian $g$ (with a structure preserving rounding) that has the same mean and covariance. The Gaussian consists of one or many blocks and has minimum non-zero eigenvalue at least $\frac {tc} {2 k^4}$. Since all the probabilities of the PMD are discretized into multiples of $\lceil \frac {n k} \ve \rceil^{-1}$, the entries of the mean vector of the Gaussian are also multiples of $\lceil \frac {n k} \ve \rceil^{-1}$ and the entries of the covariance matrix are integer multiples of $\lceil \frac {n k} \ve \rceil^{-2}$.
\end{itemize}
Note that the support of every CRV in the PMD $\r^{\sparse} \ast \r^{\lrg}$ is a subset of the support of the corresponding CRV in the PMD of the Nash equilibrium $\r$. Moreover for every CRV $i$ in $\r^{\sparse}$ and $i'$ in $\r^{\lrg}$, it holds that $\dtv(\r_{-i}, \r^{\sparse}_{-i} \ast \r^{\lrg}) < 3 \ve$ and $\dtv(\r_{-i'}, \r^{\sparse} \ast \r^{\lrg}_{-i'}) < 3\ve$. 
\end{enumerate}

After performing the steps above, we have shown that an $O(\ve)$-Nash equilibrium can be found by searching over a limited set of parameters. In particular we require to search over $\r^{\sparse}$ with accuracy $\lceil \frac {t k^3} {\ve} \rceil^{-1}$ and on $\r^{\lrg}$ with accuracy $\lceil \frac {n k} \ve \rceil^{-1}$. The search space unfortunately is still very large since it requires searching over $\r^{\lrg}$ with high accuracy. The main idea to reduce the search space for the problem is to note that the large PMD is approximable by a discretized Gaussian $g$ (with a structure preserving rounding) that has large non-zero eigenvalues, i.e. $\dtv(\r^{\lrg},g) < \ve$. 

For every player $i$ in the sparse PMD, his view about the aggregate strategy of the others is approximately the same as if the large PMD was replaced by the Gaussian, i.e. 
$$\dtv(\r^{\sparse}_{-i} \ast \r^{\lrg},\r^{\sparse} \ast g) < \ve$$ 
Moreover, for every player $i$ that corresponds to a CRV in the large PMD, his view about the aggregate strategy of the others is approximately the same as if the rest of the components in the large PMD were replaced by a Gaussian $g_{-i}$ with the same mean and covariance $\r^{\lrg}_{-i}$, i.e. 
$$\dtv(\r^{\sparse} \ast \r^{\lrg}_{-i},\r^{\sparse} \ast g_{-i}) < \ve$$ 

At this point, the aggregate behavior of all players can be summarized by describing the probabilities of the sparse PMD and providing the mean and covariance of the Gaussian.
However, as shown in Lemma 22 and Lemma 23 of~\cite{DaskalakisKT15}, it is possible reduce the search space by only keeping track of the first $\log(1/\ve)$ moments/power-sums of the sparse PMD. In particular, for a PMD $\p$ let $m_{\a_1,..,\a_k}(\p)$ be the power sum $\sum_i \prod_{j=1}^k \left( \p_{ij} \right)^{\a_j}$. If a PMD $\p^A$ has the same power sums $m_{\a_1,..,\a_k}(\p^A)$ as the PMD $\p^B$ for $\a_1,..,\a_k \in \mathbb{Z}_{\ge 0}$ such that $\sum_{j=1}^k \a_j \le \log(1/\ve)$ and additionally $|\p^A_{i j} - \p^B_{i' j}| \le (4 e k^3)^{-1}$ then $\dtv(\p^{A},\p^{B}) < 2 \ve$. Using this fact, we can partition the CRVs of the sparse PMD into at most $(4 e k^3)^{k}$ smaller components according to the value of the probability in each of the coordinates and replace all CRVs within every partition with a PMD that matches their corresponding power-sums without significant loss in total variation. So knowledge of the power-sums $m_{\a_1,..,\a_k}(\p)$ for every sub-PMD in the partition is sufficient to approximately describe the distribution of the sparse PMD.

With those observations in hand, we proceed to give the algorithm for computing $\ve$-equilibria for anonymous games. To do this, we first guess the mean $\m$ and covariance $\S$ of the Gaussian component as well as all the power-sums $m$ of the sparse PMD. We then try to construct CRVs for every player so that the overall mean and covariance as well as the power-sums match those that we guessed and moreover every player's CRV assigns positive probability mass only to approximately optimal strategies. If we are able to do so, Lemma~\ref{lem:epsnash} implies that this gives an approximate Nash equilibrium. In more detail, the algorithm performs the following steps:

\begin{enumerate}
  \item Guess the mean and covariance of the Gaussian component and 
  the power sums of the sparse PMD. For every guess, we repeat the next steps until a feasible solution is found.
  
  We need to guess the powersums for $(4 e k^3)^k$ different PMDs since CRVs are first clustered according to their value in every coordinate. Since the parameters of the sparse PMD are all multiples of $\lceil \frac {t k^3} {\ve} \rceil^{-1}$, this results in at most $2^{k^{5 k} \log^{k+2} \left(\frac 1 {\ve} \right)}$ distinct power-sum vectors in total\footnote{This upper bound was derived in~\cite{DaskalakisKT15}}. For the gaussian component all entries of the mean and covariance are multiples of $\lceil \frac {n k} {\ve} \rceil$ which requires $\lceil \frac {n k} {\ve} \rceil^{O(k^2)}$ guesses in total.
  
  \item For every player, we need to compute the contribution of his mixed strategy (CRV) to the overall distribution. If that player is to be assigned in the sparse component, its probabilities are all multiples of $\lceil \frac {t k^3} {\ve} \rceil^{-1}$ and we can compute its contribution to the power-sums $m$. Similarly, if that player is to be assigned in the gaussian component its probabilities are all multiples of $\lceil \frac {n k} {\ve} \rceil^{-1}$ and we can easily compute its contribution to the mean and covariance. 
  
  However, not all assignments are feasible. We need to consider only CRVs for that player that assign positive probability mass to coordinates that are approximately best responses to the strategy of other players. Even though we don't know the strategies of the others exactly, we can compute a good approximate description of the players view by subtracting from the power sums $m$ the players contribution (if any) and computing any $PMD$ that matches those power-sums. Similarly, if the player is mapped to the gaussian component we subtract the players mean and covariance from the overall mean $\m$ and covariance $\S$ and compute a discretized Gaussian with the resulting mean and covariance instead. We say that an assignment of a player to a component (sparse or Gaussian) and a specific distribution over strategies is feasible if it approximately maximizes the player's utility $u$ with respect to his approximate view about the strategies of others.
  
  \item To find if there exists a set of feasible strategies that matches the guessed statistic $(m,\m,\S)$, we use dynamic programming. The states of our dynamic program are the following: For any prefix of players, we keep track of the remaining power-sums, mean and covariance we need to account for. We iteratively process players one by one keeping track of which states are reachable. Our estimation is feasible if after processing all players we have accounted for all the power-sums, mean and covariance in our original guess. If we find such a solution, we output the assignment of players to mixed strategies that resulted in this solution.
\end{enumerate}
  
This algorithm is always guaranteed to find a solution $\hat \r$, since the PMD $\r^{\sparse} \ast \r^{\lrg}$ that we got by modifying a Nash equilibrium for the game, satisfies all the constraints we imposed. We now claim that the resulting PMD from this algorithm is an $\ve$-Nash equilibrium. The main ingredient to showing this is applying the CLT we developed in Theorem~\ref{thm:newclt} to show that the view $\hat \r_{-i}$ for every player $i$ is close to the view that was assumed when choosing feasible strategies for every player. Indeed, by the CLT all the CRVs that were mapped in the Gaussian component are approximable by a Gaussian with the same mean and covariance, while CRVs that were mapped in the sparse component have the same power-sums as those that we had guessed.

Applying Lemma~\ref{lem:epsnash} directly shows that this is indeed an $O(\ve)$-Nash equilibrium.

The total runtime of the algorithm is polynomial on the number of states of the above dynamic program. Since there are $\lceil \frac {n k} {\ve} \rceil^{O(k^2)}$ Gaussian parameters in total as well as $2^{k^{5 k} \log^{k+2} \left(\frac 1 {\ve} \right)}$ power sums in total, the overall runtime is $n^{O(k^2)} 2^{\poly(k,\log(1/\ve))^k}$ and the theorem follows.
\end{proof}

\ifnum\camr=0
\section{An $n^{O(k)}$ Non-Proper Cover for PMDs}
\else
\section{An $\MakeLowercase{n}^{O(\MakeLowercase{k})}$ Non-Proper Cover for PMDs}
\fi
\label{sec:body_upperb}
On the road to getting the proper cover described by Theorem~\ref{thm:propercover}, we first {\ifnum\camr=0 show Theorem~\ref{thm:newcover}.
This constructs \else construct \fi}a non-proper cover of the same size. 
\ifnum\tp=1

More details are provided in 
\ifnum\camr=0
Section~\ref{sec:non proper but better cover}.
\else
the full version of this paper.
\fi
\fi

\ifnum\tp=1
\section{Details from Section \ref{sec:body_upperb}}
\fi
\label{sec:non proper but better cover}
The main theorem of this section is the following:
\begin{theorem}
  \label{thm:newcover}
  For all $n, k \in \mathbb{N}$, and $\ve > 0$, there exists a (non-proper) $\ve$-cover, in total variation distance, of the set of all $(n,k)$-PMDs whose size is 
  $$n^{O(k)} \cdot \min \left\{2^{\poly(k/\ve)}, 2^{O(k^{5k}\log^{k+2}(1/\ve))}\right\}.$$
  Moreover, we can efficiently enumerate this cover in time polynomial in its size.
\end{theorem}
This theorem should be contrasted with Theorem~\ref{thm:propercover}, which provides a proper cover of similar size. It should also be contrasted to Theorem 2 of \cite{DaskalakisKT15}, which provides a cover with a leading factor of $n^{k^2}$, so the cover presented here improves the exponent of $n$ from quadratic to linear in the dimension.
This is the correct order of exponential dependence on $k$, as simply counting the number of $(n,k)$-PMDs with deterministic summands gives a lower bound of $n^{\Omega(k)}$. We also show in Section~\ref{sec:lowerb} that the quasi-polynomial dependence on $1/\epsilon$ with an exponent of $\Omega(k)$ cannot be avoided, as we provide an essentially matching lower bound on the cover size.

The starting point for our cover will be Theorem \ref{thm:struct}, stating that every $(n,k)$-PMD is $\ve$-close to the sum of an appropriately discretized Gaussian and a $(\poly(k/\ve), k)$-PMD. 
We generate an $\ve/2$-cover for each and combine them by triangle inequality.

\paragraph{Covering the sparse PMD.}
We cover the sparse PMD component using the same methods as in \cite{DaskalakisKT15}.
The first, naive way of covering this component involves gridding over all $\poly(k/\ve)$ parameters with $\poly(\ve/k)$ granularity.
This results in a cover size of $2^{\poly(k/\ve)}$.

The more sophisticated way of covering this component uses a ``moment matching'' technique.
A result by Roos \cite{Roos02} shows that the probability mass function can be written as the weighted sum of partial derivatives of a standard multinomial distribution. 
When analyzed carefully, his result implies that the lower order moments of the distribution are sufficient to characterize the PMD.
In other words, any two PMDs with identical ``moment profiles'' (which describe these lower order moments) are close in total variation distance, and it suffices to keep only one representative for each moment profile.
This method results in a cover of size $2^{O(k^{5k} \log^{k+2}(1/\ve))}$.
Combining this with the other approach gives a cover of size
$$\min\left\{2^{\poly(k/\ve)},2^{O(k^{5k} \log^{k+2}(1/\ve))} \right\} .$$

For more details, see the proof of Theorem 2 of \cite{DaskalakisKT15}.

\paragraph{Covering the discretized Gaussian.}
To cover the Gaussian component, \cite{DaskalakisKT15} grid over all $O(k^2)$  parameters of the Gaussian component, arguing the effectiveness of the gridding using Proposition \ref{prop:gaussapprox}. This gridding results in the leading factor of $n^{O(k^2)}$ in the size of the cover.
In contrast, we use a spectral covering approach: instead of trying to grid over all parameters of the covariance matrix, we first sparsify it and then match the magnitude of its projection in every direction.
In particular, we establish a cover of the following nature:
\begin{lemma}
\label{lem:covergaussian}
Let $\mathcal{G}_{n,k,\ve}$ be the set of all Gaussians with structure preserving roundings which may arise as a consequence of Theorem \ref{thm:struct} when applied to $(n,k)$-Poisson multinomial random vectors with parameter $\ve$.
Then there exists a set $\mathcal{S}$ of Gaussians with structure preserving roundings of size at most $n^{O(k)} \cdot \left(\frac{k}{\ve}\right)^{O(k^3)}$ with the following properties:

For any $G \in \mathcal{G}_{n,k,\ve}$, there exists a $\hat G \in \mathcal{S}$, such that $G$ and $\hat G$ have the same block structure (i.e., the partition of coordinates), and within each block, have the same pivot coordinate and sum for the mean vector coordinates.
Furthermore, for each block $i$, letting $(\m_i,\S_i)$ and $(\hat \m_i, \hat \S_i)$ be the mean and covariance for the block (excluding the pivot coordinate), we have that for all unit vectors $v$,
\begin{itemize}
\item $|v^T( \m_i - \hat \m_i) | \leq \frac{\ve \s_{iv}}{k}$;
\item $|v^T(\S_i - \hat \S_i)v | \leq \frac{\ve \s_{iv}^2}{2k^{3/2}}$;
\end{itemize}
where $\s_{iv}^2 = \max\{v^T \S_i v, v^T \hat \S_i v\}$.
\end{lemma}
This lemma statement is slightly technical due to the nature of the Gaussians with structure preserving roundings.
It essentially says that we cover the set of Gaussians arising from the structural theorem by matching their block structure exactly, and within each block, matching the moments spectrally.
Plugging these guarantees into Lemma \ref{lem:dtvgaussian} and applying the data processing inequality for total variation distance (Lemma \ref{lem:DPI}) gives the desired closeness.

For simplicity of exposition, for the remainder of this overview section, we assume that the Gaussian's structure preserving rounding consists of a single block, an assumption we do not make in the full proof (described in Section \ref{sec:upperbdetails}). 
By the guarantees of the structural result, in this case, the minimum eigenvalue of the covariance matrix is at least some ${\rm poly}(k/\ve)$. So the goal of our exposition in this section is to produce a cover of Gaussians that may result from Theorem~\ref{thm:struct} and whose covariance matrices have minimum eigenvalue at least ${\rm poly}(k/\ve)$.

Since the mean vector only has $k$ parameters, we can grid over the entries.
Though we require a spectral guarantee, this naive gridding is sufficient.
This gives a set of size $\left(\frac{nk}{\ve}\right)^{O(k)}$, such that, for any Gaussian which may arise from Theorem \ref{thm:struct}, its mean vector is approximated by a mean vector in our set with the approximation guarantees required by Lemma \ref{lem:dtvgaussian}.

Covering the covariance matrix takes more care.
At a high level, our approach views PMDs through the lens of spectral graph theory and exploits the existence of spectral sparsifiers.
Recall the definition of the Laplacian matrix of a graph:
\begin{definition}
  Given an undirected weighted graph $G = (V,E,w)$ on $n$ vertices, its \emph{Laplacian matrix} is an $n \times n$ matrix $L_G$ where
  \begin{align*}
  L_G(i,j) = 
  \begin{cases}
    \sum_{k \neq i} w(i,k) &\text{if}\ i = j \\
    -w(i,j) &\text{if}\ i \neq j \wedge (i,j) \in E, \\
         0 &\text{otherwise}
  \end{cases}
  \end{align*}
\end{definition}

To see the connection to PMDs, we observe that the covariance matrix of a PMD is the Laplacian matrix of a graph defined by the parameters.
For a single $k$-CRV $X$ with parameter vector $\p$, it can be shown that the variance of $X_i$ is $\p(i)(1 - \p(i))$ and the covariance of $X_i$ and $X_j$ is $-\p(i)\p(j)$. 
Since $\sum_{i=1}^k \p(i) = 1$, the covariance matrix is equal to the Laplacian matrix of a graph on $k$ nodes with $w(i,j) = \p(i)\p(j)$.
This can be extended to $(n,k)$-PMDs by observing that the sum of random variables has a covariance matrix equal to the sum of the individual covariance matrices, and a similar statement holds for graphs and the corresponding Laplacian matrices.
We summarize this connection in the following observation:
\begin{observation}
  \label{obs:pmdlaplacian}
  The covariance matrix of an $(n,k)$-Poisson Multinomial Distribution $M^\p$ corresponds to the Laplacian matrix of a graph $G = (V,E,w)$ on $k$ nodes, where the $w(i,j) = \sum_{\ell = 1}^n \p(\ell,i)\p(\ell,j)$.
\end{observation}

At the core of our approach, we use the following celebrated result of Batson, Spielman, and Srivastava \cite{BatsonSS12}, which says that the Laplacian matrix of a graph on $k$ vertices can be spectrally approximated by the Laplacian matrix of a graph with only $O(k)$ edges:
\begin{theorem}[Theorem 1.1 in \cite{BatsonSS12}]
  \label{thm:spectralsparse}
For every $\ve \in (0,1)$, every undirected weighted graph $G = (V,E,w)$ on $n$ vertices contains a weighted subgraph $H = (V,F,\tilde w)$ with $\lceil (n-1)/\ve^2 \rceil$ edges which satisfies
$$(1 - \ve)^2 L_G \preceq L_H \preceq (1+\ve)^2 L_G,$$
where $L_G$ is the Laplacian matrix of the graph $G$.
\end{theorem}

Using this tool, the approach will proceed as follows.
This theorem implies that, for every true covariance matrix $\S$, there exists a matrix $M_1$ with only $O(k)$ entries which preserves every projection up to a multiplicative factor of $1/5$.
We can obtain a matrix $M_2$ with the same sparsity pattern as $M_1$ by guessing which subset of $O(k)$ entries is non-zero, requiring $\exp(k\cdot \log{k})$ guesses.
Furthermore, we can grid over the non-zero entries of $M_2$ to ensure that it approximates every projection of $M_1$ up to a multiplicative factor of $1/25$.
Since $M_1$ has minimum eigenvalue ${\rm poly}(k/\epsilon)$ and maximum entry $O(n)$, gridding requires only $\left({n\cdot k \over \epsilon}\right)^{O(k)}$ guesses, and we get that $M_2$ gives a $1/4$ multiplicative spectral approximation to $\S$.
To make our approximation finer, we will $O(\ve/\sqrt{k})$-cover the set of PSD matrices within a $1/4$-neighborhood of $M_2$.
We first recall the definition of a cover in this context:
\begin{definition}
  Let S be a set of symmetric $k \times k$ PSD matrices. An \emph{$\ve$-cover} of the set $S$, denoted by $S_{\ve}$, is a set of PSD matrices such that for any matrix $A \in S$, there exists a matrix $B \in S_{\ve}$ such that for all vectors $y$: $|y^T (A-B) y| \le \ve y^T A y$.
\end{definition}
Now, if we could $O(\ve/\sqrt{k})$-cover the set of all matrices $1/4$-close to $M_2$, we would obtain an $O(\ve/\sqrt{k})$-approximation to $\S$.
We do so using the following lemma, which provides a method to generate such a cover.
A slight generalization of this statement appeared as Lemma 9 in \cite{DaskalakisKT15}, but we give a slightly simpler proof in Section \ref{sec:psdcoverproof} for completeness.
\begin{lemma}[Lemma 9 in \cite{DaskalakisKT15}]
    \label{lem:psdcover}
      Let $A$ be a symmetric $k \times k$ PSD matrix with minimum eigenvalue at least $1$ and let $S$ be the set of all matrices $B$ such that $|y^T(A - B)y| \le \ve_1 y^T A y$ for all vectors $y$, where $\ve_1 \in [0,1/4]$.
        Then, there exists an $\ve$-cover $S_{\ve}$ of $S$ that has size $|S_{\ve}| \le \left(\frac {k} {\ve} \right)^{O(k^2)}$.
\end{lemma}

Combining the above, we obtain a set of covariance matrices of size $n^{O(k)}\cdot \left(\frac{1}{\ve}\right)^{\poly(k)}$ such that, for any Gaussian which may arise in Theorem \ref{thm:struct}, its covariance matrix is approximated by a covariance matrix in our cover as required by Lemma \ref{lem:dtvgaussian}.

Combining the guarantees obtained for the mean and the covariance matrix, we find that they satisfy both conditions of Lemma \ref{lem:dtvgaussian}.
Therefore, we have described a cover of size $n^{O(k)} \cdot \left(\frac{1}{\ve}\right)^{\poly(k)}$ for all possible Gaussian components.
The proof of Theorem \ref{thm:newcover} is completed by taking the Cartesian product of this Gaussian cover with the cover for the $(\poly(k/\ve),k)$-PMD component.

For more details on covering the Gaussian component, see Section \ref{sec:upperbdetails}.

\subsection{Details on Covering the Gaussian Component}
\label{sec:upperbdetails}
Recall that the Gaussian component will have a structure preserving rounding.
The first step in designing our cover will be to guess the partitioning into blocks.
There are $k$ dimensions, resulting in at most $k!$ different block structures.
In what follows, we will describe how to cover a single block up to accuracy $O(\frac{\ve}{k})$, taking the Cartesian product of the resulting sets will give an $O(\ve)$-cover of the entire Gaussian at the additional cost of $k$ in the exponent.

For a single block which consists of dimensions $S_i$, we must first guess the size parameter $n_i$ and which dimension is to be used as the pivot.
The former is an integer between $0$ and $n$, and guessing it comes at a cost of $n$ in our cover size.
Guessing the latter comes at a $|S_i|$ cost in our cover size.

Recall that our strategy will be to spectrally match the parameters of the true Gaussian.
We will conclude the two distributions are close using the guarantees provided by Lemma \ref{lem:dtvgaussian}.
We describe how to obtain such guarantees for both the mean and covariance matrix separately.

\subsubsection{Covering the Mean Vector of a Block}
We know the mean of the block will be contained in the cube $[0,n_i]^{|S_i|}$.
For some $\a(k,\ve)$ (which for simplicity, we assume divides $n_i$), consider the lattice $\{0, \a, 2\a, \dots, n_i\}^{|S_i|}$, which has $(\frac{n_i}{\a} + 1)^{|S_i|}$ points.
We note that the maximum $\ell_2$ distance between the mean $\m$ and the closest point of this lattice $\hat \m$ is at most $\a \sqrt{k}$, and therefore, for any unit vector $v$, we have that $|v^T (\m - \hat \m)| \leq \a \sqrt{k}$. 
We also know that the minimum variance of any projection the Gaussian is large, in particular, at least $\frac{tc}{2k^4}$, so the standard deviation in any direction $v$ is $\s_v \geq \sqrt{\frac{tc}{2k^4}}$.
Choosing $\a \leq k^5/\ve \leq \ve \s_v / k^{3/2}$ implies that $\a \sqrt{k} \leq \ve \s_v/k$.
This shows that the first condition of Lemma \ref{lem:dtvgaussian} is satisfied to approximate this block up to $\frac{\ve}{k}$ accuracy.
Substituting the value of $\a$, we cover the mean with a set of size at most $\left(\frac{n_i\ve}{k^5} + 1\right)^{|S_i|}$.

\subsubsection{Covering the Covariance Matrix of a Block}
We will use the characterization provided by Observation \ref{obs:pmdlaplacian}, which tells us that the covariance matrix of an $(n,k)$-PMD is the Laplacian matrix of a graph defined by the parameters of the distribution.
Recall from the proof of Theorem \ref{thm:struct} (which appears in \cite{DaskalakisKT15}), the covariance matrix of the Gaussian we are attempting to match is also the covariance matrix of an $(n_i,|S_i|)$-Generalized Multinomial Distribution.
For the remainder of this proof, we let $G$ be the graph defined by this characterization for the covariance matrix of the corresponding $(n_i,|S_i|)$-Poisson Multinomial Distribution.

As a starting point, we use Theorem \ref{thm:spectralsparse}, which shows the existence of spectral sparsifiers. 
In particular it implies that, if given $G$ on $|S_i|$ nodes and we want a subgraph $H$ such that $(1 - 1/5) L_G \preceq L_H \preceq (1+ 1/5) L_G$, there exists an $H$ with at most $110|S_i|$ edges which gives this approximation.
The first step in covering the covariance matrix is to guess which edges are present in the graph.
Since there are $\binom{|S_i|}{2}$ possible edges in the graph, this requires at most
$$\binom{\binom{|S_i|}{2}}{110|S_i|} \leq k^{220k}$$
guesses.

Now that we know which edges are present in the graph, the goal is to guess the weights of these edges.
Ideally, we would like to obtain a graph $M$ with the guarantee that $(1 - 1/25) L_H \preceq L_M \preceq (1 + 1/25)L_H$.
However, this is stronger than we can hope for, since recalling that $L_H$ has a zero eigenvalue, it would require that the diagonals of $L_M$ and $L_H$ are exactly equal.
Instead, we recall that we have a pivot coordinate which will be left out of the Gaussian's covariance matrix, and we only have to match projections which are orthogonal to this direction.
Without loss of generality, assume that the pivot coordinate is $1$.
For any unit vector $v \in \mathbb{R}^k$ orthogonal to $e_1$, we will obtain an $L_M$ such that
$$v^T(L_H - L_M)v \leq \frac{1}{25}v^TL_Hv,$$
which will imply
$$v^T(L_G - L_M)v \leq \frac{1}{4}v^TL_Gv.$$
Further, recall that our structural result implies that $\frac1{25}v^TL_Hv \geq \frac{tc}{100k^4}$,
so it suffices to obtain a graph $M$ such that
$$v^T(L_H - L_M)v \leq \frac{tc}{100k^4}.$$

For a unit vector $v$ and $|S_i| \times |S_i|$ PSD matrices $A$ and $B$,
$$v^T(A - B)v = \sum_{i,j} v_i v_j (A(i,j) - B(i,j)) \leq |S_i|^2 \max_{i,j} |A(i,j) - B(i,j)|.$$
Suppose we guess the edge weights of $M$ such that they are at most $\frac{tc}{100k^7}$ away from those of $H$.
This tells us $\max_{i \neq j} |L_H(i,j) - L_M(i,j)| \leq \frac{tc}{100k^7}$, and since the diagonal entries of $L_M$ are the sums of the off-diagonal entries, $\max_{i} |L_H(i,i) - L_M(i,i)| \leq \frac{tc}{100k^6}$.
This implies that it suffices to additively estimate the edge weights up to accuracy $\frac{tc}{100k^6}$.
Since the maximum entry of $L_G$ is at most $n_i$, the spectral guarantee implies that the maximum entry of $L_H$ is at most $\frac{6n_i}{5}$, and similarly, the maximum edge weight. 
Therefore, gridding over all $110|S_i|$ non-zero edge weights, we define a set with at most
$$\left(\frac{6n_i/5}{tc/100k^7}\right)^{110|S_i|} \leq \left(\frac{3n_i\ve^5}{250k^{11}}\right)^{110|S_i|}$$
candidates.

At this point, we have a PSD matrix $L_M$ which, when projected onto the subspace orthogonal to $e_1$, is $1/4$-spectrally close to the target covariance matrix.
We wish to $\frac{\ve}{2k^{3/2}}$-cover the space of all PSD matrices which are $1/4$-spectrally close to this matrix.
We will use Lemma \ref{lem:psdcover}, which we instantiate with parameter
``$\ve$'' set to $\frac{\ve}{2k^{3/2}}$, allowing us to generate a $\frac{\ve}{2k^{3/2}}$-cover of a $\frac14$-neighborhood of a given PSD matrix with $\left(\frac{k}{\ve}\right)^{O(k^2)}$ candidates.
Since we knew one of the previous candidates was $\frac14$-close to the target, this gives us a matrix which satisfies the second condition of Lemma \ref{lem:dtvgaussian} to approximate this block up to $\frac{\ve}{k}$ accuracy.
The size of this cover is at most
$$ k^{220k} \cdot \left(\frac{3n_i\ve^5}{250k^{11}}\right)^{110|S_i|} \cdot \left(\frac{k}{\ve}\right)^{O(k^2)} = n^{O(|S_i|)}\left(\frac{k}{\ve}\right)^{O(k^2)}.$$

\subsubsection{Putting the Guarantees Together}
\label{subsec:counting size of cover}
At this point, to cover a single block up to accuracy $O(\ve/k)$, we have a set of size at most 
$$n \cdot |S_i| \cdot\left(\frac{n_i\ve}{k^5} + 1\right)^{|S_i|} \cdot n^{O(|S_i|)}\left(\frac{k}{\ve}\right)^{O(k^2)} = n^{O(|S_i|)}\left(\frac{k}{\ve}\right)^{O(k^2)}.$$
Taking the Cartesian product of sets and multiplying by the number of guesses for the block structure of the Gaussian, we get an overall cover of size 
$$k! \cdot \prod_{S_i} \left(n^{O(|S_i|)}\left(\frac{k}{\ve}\right)^{O(k^2)}\right) = n^{O(k)}\left(\frac{k}{\ve}\right)^{O(k^3)} .$$

Combining with the cover for the $(\poly(k/\ve),k)$-PMD component, we obtain an overall cover for $(n,k)$-PMDs of size
$$n^{O(k)} \cdot \min\left\{2^{\poly(k/\ve)}, 2^{O(k^{5k} \cdot \log^{k+2}(1/\ve))} \right\},$$
as desired.

\subsection{Proof of Lemma \ref{lem:psdcover}}
\label{sec:psdcoverproof}
To construct the cover, we will make use of the eigenvalues and eigenvectors of the matrix $A$.
We first show that for any matrix $B \in S$, its eigenvalues are close to the eigenvalues of $A$. 

\begin{proposition} \label{lem:courant}
Let $A,B$ be two symmetric $k \times k$ PSD matrices such that for all vectors $y$ with $\|y\|=1$, $|y^T(A - B)y| \le \ve_1 y^T A y$ for some constant $\ve_1 > 0$. Then for the eigenvalues $\lambda^A_1 \le ... \le \lambda^A_k$ of $A$, and the eigenvalues $\lambda^B_1 \le ... \le \lambda^B_k$ of $B$, it holds that:
$$|\lambda^A_i - \lambda^B_i| \le \ve_1 \lambda^A_i $$
\end{proposition}

\begin{proof}
From Courant's minimax principle, we have that the $i$-th eigenvalue of $A$ is equal to:
$$\lambda^A_i=\max\limits_C\min\limits_{\binom{\| x\| =1}{Cx=0}}  x^T Ax$$
where $C$ is an $(i-1) \times k$ matrix. For the matrix $B$, we have that
$$\lambda^B_i=\max\limits_C\min\limits_{\binom{\| x\| =1}{Cx=0}}  x^T Bx \le \max\limits_C\min\limits_{\binom{\| x\| =1}{Cx=0}}  (1+\ve_1) x^T Ax =  (1+\ve_1) \lambda_i^A$$
Similarly, we have that $\lambda^B_i \ge (1-\ve_1) \lambda_i^A $, so the result follows.
\end{proof}

By computing the eigenvalues $\mu_1 \le \dots \le \mu_k$ of $A$, we have estimates of the eigenvalues $\lambda_1, \dots, \lambda_k$ of $B$ within a multiplicative factor of $1 \pm 2 \ve_1$. 
We can improve our estimates to a better multiplicative factor $1 \pm \ve$ by gridding multiplicatively around each eigenvalue. 
This requires another $\log_{1+\ve}\left(\frac{1+2 \ve_1}{1 - 2\ve_1}\right) = O(1/\ve)$ guesses per eigenvalue. 
So in total, we require $\left( \frac {1} {\ve} \right)^{O(k)}$ guesses for obtaining accurate estimates $\lambda'_1, \dots, \lambda'_k$ of the eigenvalues of $B$.

Once we know (approximately) the eigenvalues of $B$, we will try to guess also its eigenvectors $v_1,\dots,v_k$. We will do this by performing a careful gridding around the eigenvectors of $A$ which we can assume, without loss of generality (by rotating), to be the standard basis vectors $e_1, e_2, \dots,e_k$. So for each eigenvector $v_z$ of $B$, we will try to approximate it by guessing its projections to the eigenvectors of $A$. 

We now bound the projections of eigenvectors of $A$ to eigenvectors of $B$. 
Since we know that  $e^T_i B e_i \le (1+\ve_1) e^T_i A e_i$, we get that $\sum_z \lambda_z (v_z e_i)^2 \le (1+\ve_1) \mu_i$ which implies that $ v_{z,i} \le \sqrt { \frac { 2 \mu_i} { \lambda_z } }$. 
Moreover, since $\lambda_z \ge \max \{ (1-\ve_1)\mu_z , 1\} \geq \max \{ \frac12 \m_z, 1\}$, we know that the projection of $v_z$ to $e_i$ will be smaller than $2 \sqrt {\frac { \mu_i } { \max \{ \mu_z , 1\} } }$.
An additional bound for the projection of $v_z$ to $e_i$ can be obtained by considering the variance of the matrices $A$ and $B$ in the direction $v_z$. 
Since we know that
$v^T_z B v_z \ge (1-\ve_1) v^T_z A v_z$, we get that $\sum_i \mu_i (v_z e_i)^2 \le \frac{\l_z}{1-\ve_1} \le 2  \lambda_z $ which implies that $v_{z,i} \le \sqrt { \frac {2 \lambda_z  } { \mu_i } }$.

We now guess vectors $v'_1,...,v'_k$ that approximate the eigenvectors of $B$ by additively gridding over the projections to each eigenvector of $A$. 
In particular, our candidate guesses for $v'_z \cdot e_i = v'_{z,i}$ will be $ \ell \ve' \min \left\{2 \sqrt {\frac { \mu_i } { \max \{ \mu_z , 1\} } }, 1 \right\}$ with $\ell \in \{0, 1, \dots, 1/\ve'\}$, for a small enough $\ve'$ that only depends on $k$ and $\ve$ . 
This will give us an approximation $v'_z$ for the eigenvector $v_z$, with the guarantee that $|v'_{z,i} - v_{z,i}| \leq \ve' \min \left\{2 \sqrt {\frac { \mu_i } { \max \{ \mu_z , 1\} } }, 1 \right\}$.
This requires $\frac 1 {\ve'}$ guesses for each projection, and thus $\left( \frac 1 {\ve'} \right)^{k^2}$ guesses for all $k^2$ projections. The final covariance matrix we output is then $\hat B = \sum_z \lambda'_z v'_z (v_z')^T$.

We will now show that the covariance matrix $\hat B$ satisfies the property that it is close in all directions to $B$. To do this we will make use of the following lemma from \cite{DaskalakisKT15}.
This roughly states that two PSD matrices spectrally approximate each other in $O(k^2)$ particular directions, then they spectrally approximate each other in every direction.
\begin{lemma}[Lemma 25 from \cite{DaskalakisKT15}]
Let $\S,\hat \S \in \mathbb{R}^{k \times k}$ be two symmetric, positive semi-definite matrices, and let $(\l_1,v_1), \dots, (\l_k,v_k)$ be the eigenvalue-eigenvector pairs of $\S$.
Suppose that
\begin{itemize}
\item For all $i \in [k]$, $\Big|\Big(\frac{v_i}{\sqrt{\l}_i}\Big)^T \Big(\hat \S - \S\Big) \Big(\frac{v_i}{\sqrt{\l}_i}\Big)\Big|  \leq \ve $,
\item For all $i, j \in [k]$, $\Big|\Big(\frac{v_i}{\sqrt{\l}_i} + \frac{v_j}{\sqrt{\l}_j}\Big)^T \Big(\hat \S - \S\Big) \Big(\frac{v_i}{\sqrt{\l}_i} + \frac{v_j}{\sqrt{\l}_j}\Big)\Big| \leq 4\ve$.
\end{itemize}
Then for all $y \in \mathbb{R}^k$, $\Big|y^T \Big(\hat \S - \S\Big) y\Big| \leq 3k\ve y^T \S y$.
\end{lemma}

We will only consider directions $y = \frac {v_z} {\sqrt  \lambda_z}$ for $z \in [k]$ and $ y = \frac {v_z} {\sqrt  \lambda_z} + \frac {v_{z'}} {\sqrt \lambda_{z'}}$ for $z,z' \in [k]$.

We first consider direction $y = \frac {v_z} {\sqrt  \lambda_z}$. We have that:
$$\frac {v_z^T} {\sqrt  \lambda_z} \hat B \frac {v_z} {\sqrt  \lambda_z} = \sum_i \frac {\lambda'_i} {\lambda_z} (v_z^T v'_i)^2 = \sum_i \frac {\lambda'_i} {\lambda_z} (v_z^T v_i + v_z^T (v'_i - v_i) )^2 = \frac {\lambda'_z} {\lambda_z} (1 + v_z^T (v'_z -v_z) )^2 + \sum_{i \neq z} \frac {\lambda'_i} {\lambda_z} (v_z^T (v'_i -v_i) )^2$$

The first term is in the range $[(1-\ve)(1-k \ve')^2,(1+\ve)(1+k \ve')^2]$, which for $\ve' \le \ve/k$, becomes $(1 \pm O(\ve))$. The rest of the terms can be bounded as follows:

\begin{align*}
\frac {\lambda'_i} {\lambda_z} (v_z (v'_i -v_i) )^2 &\le (1+\ve) \frac {\lambda_i} {\lambda_z}  ( \sum_j v_{z,j} (v'_{i,j} -v_{i,j}) )^2 \\
&\le  (1+\ve) \frac {\lambda_i} {\lambda_z} \left( \sum_j \sqrt { 2 \frac {\lambda_z } { \mu_j } } \ve' 2 \sqrt {\frac { \mu_j  } { \max \{ \mu_i   , 1\} } } \right)^2 \\
&\le  (1+\ve) \frac {\lambda_i} {\lambda_z} \left( \sum_j 2 \ve' \sqrt{2 \lambda_z} \sqrt { \frac { 1  } { \max \{ \mu_i , 1\} } } \right)^2 \\
&\le  (1+\ve) \left( \sum_j 2  \ve'   \sqrt { \frac { 2 \lambda_i } { \max \{ \mu_i  , 1\} } } \right)^2 \\
&\le  (1+\ve) \left( 4 k  \ve'   \sqrt { \frac { \mu_i } { \max \{ \mu_i , 1\} } } \right)^2 \\
&\le  (1+\ve) \left( 8 k \ve' \right)^2 \\
&\le \frac {\ve} k
\end{align*}
for $\ve' = O( \sqrt \frac{\ve}{k^3}  )$. This means that $v^T_z \hat B v_z \in (1 - \ve, 1 + \ve) \l_z$. The proof is similar for directions $y = \frac {v_z} {\sqrt  \lambda_z} + \frac {v_{z'}} {\sqrt \lambda_{z'}}$ for $z,z' \in [k]$.

Overall, we can get an estimate $\hat B$ of any matrix $B \in S$ by making at most $\left(\frac {k } {\ve} \right)^{O(k^2)}$ guesses, which implies an $\ve$-cover of this size.

\section{A Proper Cover for PMDs}
\label{sec:body_proper}
\ifnum\tp=1
In this section, we give the main ideas needed to complete the construction of our proper cover, as described by Theorem~\ref{thm:propercover}.

For more details on our proper cover construction, see 
\ifnum\camr=0
Section \ref{sec:proper cover}.
\else
the full version of this paper.
\fi
\fi

\ifnum\tp=1
\section{Details from Section \ref{sec:body_proper}}
\fi
\label{sec:proper cover}
\label{sec:overview of proper cover}
We show how to turn the non-proper cover of Section~\ref{sec:non proper but better cover} into a proper one as described by Theorem~\ref{thm:propercover}, using Theorem~\ref{thm:newclt}. We note that a non-constructive proper cover follows immediately from Theorem \ref{thm:newcover}, since for each element of an improper $\ve/2$-cover that lies within $\ve/2$ of a PMD, we can match it with such a PMD. The resulting set of PMDs defines then a proper $\ve$-cover.
Our focus in this section is to provide an efficient construction of a proper cover.

Our approach will be to enumerate the improper cover of Theorem \ref{thm:newcover} and convert each distribution to a nearby $(n,k)$-PMD.
This cover consists of distributions which are the sum of a Gaussian with a structure preserving rounding and a $(\poly(k/\ve), k)$-PMD. 
Since the $(\poly(k/\ve),k)$-PMD component is already a collection of $k$-CRVs, this part of the cover is already proper, and it suffices to convert the Gaussian component into a nearby $(n - \poly(k/\ve), k)$-PMD.

The main technical lemma we prove is the following, which states that if a discretized Gaussian $G$ is spectrally close to a GMD $\r$, we can obtain a new GMD $\r'$ which is spectrally close to $\r$:
\begin{lemma}
\label{lem:spectralsf}
  Let $\lfloor \mathcal{N} (\m,\S)\rceil$ be a discretized Gaussian and suppose there exists a $(n,k)$-GMD $\r$ with mean $\m^{\r}$ and covariance $\S^{\r}$  such that for all vectors $v$ it holds that $|v^T (\m-\m^{\r})| \le \ve_1 \sqrt{ v^T \S v }$ and $|v^T   (\S-\S^{\r}) v| \le \ve_2 v^T \S v$.
  
  Then, it is possible to compute in time $n^{O(k)}$ a $(n,k)$-GMD $\r'$ with mean $\m^{\r'}$ and covariance $\S^{\r'}$  such that for all vectors $v$ it holds that $|v^T (\m-\m^{\r'})| \le \ve_1 \sqrt{ v^T \S v } + 3k^{2.5} \|v\|_2$ and $|v^T   (\S-\S^{\r'}) v| \le \ve_2 v^T \S v + 3k^3\|v\|_2^2$.
\end{lemma}
We prove this lemma using the Shapley-Folkman lemma \cite{Starr69}, which states that the Minkowski sum of a large number of sets is approximately convex:
\begin{lemma}[Shapley-Folkman lemma]
\label{lem:sf}
Let $S_1, \dots, S_n$ be a collection of sets in $\mathbb{R}^d$, and let $S = \{\sum_{i=1}^n x_i\ |\ x_1 \in S_1,...,x_n \in S_n\}$ be their Minkowski sum.
Then, letting $\textrm{conv}(X)$ denote the convex hull of $X$, every $x \in \textrm{conv}(S) = \sum_{i=1}^n x_i$ where $x_i \in \textrm{conv}(S_i)$ for $i = 1, \dots, n$ and $|\{i\ |\ x_i \not \in S_i\}| \leq d$.
\end{lemma}
With this lemma in hand, the proof of Lemma \ref{lem:spectralsf} proceeds as follows.
Let $\mathcal{M}$ be the set of all possible mean and covariances for a single CRV, and $\mathcal{M}^{\oplus n}$ be the Minkowski sum of $n$ copies of $\mathcal{M}$.
Given a discretized Gaussian with mean and covariance $(\m,\S) \in \mathcal{M}^{\oplus n}$, we would ideally like to find $\{x_1, \dots, x_n\}$ such that $\sum_{i=1}^n x_i = (\m, \S)$.
However, since this set is not convex, this optimization problem is not obviously tractable.
Instead, we convert $(\m, \S)$ to a spectrally close $(\hat \m, \hat \S)$ which lies on the convex hull of $\mathcal{M}^{\oplus n}$, which can be done using a linear program.
At this point, we exploit the ``almost convex'' characterization provided the Shapley-Folkman lemma, and we will iteratively ``peel off'' plausible CRVs.
More specifically, noting that the moment profile is at most $k^2 + k$ dimensional and applying Lemma \ref{lem:sf}, we can use a linear program to find the parameters of a single CRV such that subtracting its moments gives a moment profile which lies on the convex hull of $\mathcal{M}^{\oplus n-1}$.
We repeat $n - k^2 - k$ times until we are left with a point on the convex hull of $\mathcal{M}^{\oplus k^2 + k}$, at which point we may pick the last $k^2 + k$ CRVs arbitrarily.
The proof is completed by arguing that the resulting GMD satisfies the theorem conditions.
For the full proof of Lemma \ref{lem:spectralsf}, see Section \ref{sec:spectralsf}.

We now prove Theorem \ref{thm:propercover}.
As mentioned before, for our starting point, we relate our original PMD $\p$ to the sum of a discretized Gaussian with a structure preserving rounding and a $(\poly(k/\ve'),k)$-PMD using \ref{thm:newcover}, for some $\ve'$ to be set later.
This comes at a cost of $\ve'$ in total variation distance.
The CRVs corresponding to the sparse PMD are already in the form desired for the proper cover, and we ignore them for the remainder of the proof.
We also know that the discretized Gaussian's mean and covariance matrix arose from the mean and covariance matrix of some PMD.
This covariance matrix has a block structure, where each block has a minimum eigenvalue of at least $\frac{k^{15}}{2\ve'^6}$.
At this point, we wish to show that each block of the current PMD $\tilde \p$ is $\ve/k$-close to each block of the PMD after applying the method of Lemma \ref{lem:spectralsf}, $\p'$.
This will be proven by relating a block of $\tilde \p$ and $\p'$ to the corresponding discretized Gaussians using Theorem \ref{thm:newclt}, and arguing that the discretized Gaussians are close using Lemma \ref{lem:dtvgaussian}.

We focus on one block of $\tilde \p$.
The guarantee of our cover, summarized in Lemma \ref{lem:covergaussian}, tells us that the corresponding block of $\p'$ will have a matching pivot and constituent number of CRV's $n_i$.
Therefore, it suffices to consider the corresponding GMDs which exclude the pivot coordinate, namely $\tilde \r$ and $\r'$.
We know that the minimum eigenvalue of this block of $\tilde \r$'s covariance matrix is at least $\frac{k^{15}}{2\ve'^6}$.
The guarantees of Lemma \ref{lem:covergaussian} give us an input to Lemma \ref{lem:spectralsf} with $\ve_1 = \frac{\ve}{k}$ and $\ve_2 = \frac{\ve}{2k^{3/2}}$.
Since the minimum variance of this block of $\tilde \r$ is sufficiently large, the output of Lemma \ref{lem:spectralsf} is a relative spectral approximation to the mean and covariances, with multiplicative $2\ve_1$ and $2\ve_2$ factors, respectively.
We note that this implies that the minimum eigenvalue of this block of  $\r'$'s covariance matrix is at least $\frac{k^{15}}{4\ve'^6}$.

We convert this block of $\tilde \r$ to the corresponding discretized Gaussian using our CLT, Theorem~\ref{thm:newclt}.
Given the aforementioned minimum eigenvalue condition, the cost incurred is at most
$$O\left(\frac{k^{7/2}}{(k^{15}/\ve'^6)^{1/20}} \right) = O(k^{11/4}\ve'^{3/10}).$$
We convert the same block of $\r'$ to a discretized Gaussian in the same way, incurring the same cost.
Finally, we relate the two discretized Gaussians in total variation distance.
As mentioned in the previous paragraph, the means and covariances are spectrally close up to relative accuracy $\frac{2\ve}{k}$ and $\frac{\ve}{k^{3/2}}$.
We plug this guarantee into Lemma \ref{lem:dtvgaussian} and apply the data processing inequality (Lemma \ref{lem:DPI}) to conclude that the two distributions are $O(\ve/k)$-close.
The proof is concluded by setting $\ve' = \ve^{10/3}/k^{25/2} $ and rescaling $\ve$ by a constant factor.

\subsection{Proof of Lemma \ref{lem:spectralsf}}
\label{sec:spectralsf}
  We first argue that rounding all constituent probability vectors in the $(n,k)$-GMD $\r$ so that all their coordinates are integer multiples of $1/n$ to obtain a $(n,k)$-GMD $\hat \r$ approximately preserves the spectral closeness guarantees with the discretized Gaussian. More specifically, for all vectors $v$ it holds that:
$$|v^T (\m-\m^{\hat \r})| \le \ve_1 \sqrt{ v^T \S v } + \sqrt{ k } \|v\|_2 \,\, \textrm{ and } \,\, |v^T   (\S-\S^{\hat \r}) v| \le \ve_2 v^T \S v + k \|v\|_2^2.$$ 

We know that $\|\m^{\r} - \m^{\hat \r}\|_{\infty} \le 1$ and thus $\|\m^{\r} - \m^{\hat \r}\|_{2} \le \sqrt{k}$, so
\begin{align*}
  |v^T (\m-\m^{\hat \r})| &\le |v^T (\m-\m^{\r})| + |v^T (\m^{\hat \r}-\m^{\r})| \\
  &\le \ve_1 \sqrt{v^T \S v} + \|\m^{\r} - \m^{\hat \r}\|_{2} \|v\|_2 \\
  &= \ve_1 \sqrt{v^T \S v} + \sqrt{k} \|v\|_2
\end{align*}

Similarly we have that $\|\S^{\r} - \S^{\hat \r}\|_{\textrm{max}} \le 1$ which implies that $| v^T(\S^{\r} - \S^{\hat \r}) v| \le k \|v\|^2$ for all vectors $v$. Thus, 
$$|v^T   (\S-\S^{\hat \r}) v| \le |v^T   (\S-\S^{\r}) v| + |v^T   (\S^{\hat \r}-\S^{\r}) v| \le \ve_2 v^T \S v + k \|v\|_2^2.$$

At this point, we have shown that there exists a $(n,k)$-GMD with mean and covariance close to that of the discretized Gaussian such that all its constituent probability vectors have coordinates that are integer multiples of $1/n$. Now, for every probability vector $\vec p$ with probabilities that are multiples of $1/n$, consider its moment profile $(\m^{\vec p}, \S^{\vec p})$, where $\m^{\vec p} = \vec p$ and $\S^{\vec p}$ are the mean and covariance of the $k$-CRV with probabilities $\vec p$. Let $\cal{M}$ be the set of all possible moment profiles generated by such probability vectors $\vec p$. Since there are at most $n^{k-1}$ probability vectors $\vec p$ the set $\cal M$ has size at most $n^{k-1}$. Moreover, it is easy to see that for the rounded GMD $\hat \rho$, it holds that $(\m^{\hat \rho}, \S^{\hat \rho}) \in {\cal{M}}^{\oplus n}$ where ${\cal{M}}^{\oplus n} = \{x\ |\ \exists x_1,...,x_n \in {\cal M}, x = \sum_i x_i\}$ denotes the Minkowski addition of $\cal{M}$ with itself $n$ times. This is because the mean and covariance of the GMD is equal to the sum of the means and covariances of its constituent CRVs, which are all in $\cal{M}$ since each CRV has probabilities that are integer multiples of $1/n$.

Naively searching over ${\cal{M}}^{\oplus n}$ for a GMD that satisfies the guarantees of $\hat \r$ is not easy since it would require time that is exponential in $n$. To get a computationally efficient algorithm, we search instead in the set $\textrm{conv} \left( {\cal{M}}^{\oplus n} \right) = \textrm{conv} \left( {\cal{M}} \right)^{\oplus n}$ where, for a set $A$, $\textrm{conv}(A)$ denotes its convex closure, and the equality is a basic property of Minkowski sums.
The reason this is easy is that it is solvable by a linear program as follows:
\begin{itemize}
  \item For $m \in {\cal{M}}$ and $i \in \{1,...,n\}$, we assign the variables $x_{i,m} \ge 0$ that denote whether we want to pick the moment profile $m$ for the $i$-th CRV.
  \item For all $i$, we need that $\sum_{m} x_{i,m} = 1$. This ensures that for all $i$, $\sum_{m} x_{i,m} m \in \textrm{conv}\left( {\cal{M}} \right)$.
  \item We need that the aggregate moment profile $(\hat \m, \hat \S) = \sum_{i,m} x_{i,m} m$ satisfies the closeness constraints with $(\m,\S)$. For all $v$ we require that:
  $$|v^T (\m-\hat \m)| \le \ve_1 \sqrt{ v^T \S v } + \sqrt{ k } \|v\|_2 \,\, \textrm{ and } \,\, |v^T   (\S-\hat \S) v| \le \ve_2 v^T \S v + k \|v\|_2^2.$$ 
\end{itemize}

These are all linear constraints so a solution $(\hat \m, \hat \S) = \sum_{i,m} x_{i,m} m$, can be computed by solving the linear program using the Ellipsoid method. Note that the constraints of the third bullet are infinitely many but can be verified efficiently using a separation oracle. To check the first set of constraints, we can check whether the optimization problem $$\min_{\|v\| \le 1} \ve_1 \sqrt{ v^T \S v } + \sqrt{ k } \|v\|_2 - v^T (\m-\hat \m)$$
has a negative solution. This is a convex optimization problem which can be solved in polynomial time.
To check the second set of constraints, we note that $\ve_2 v^T \S v + k \|v\|_2^2 = v^T (\ve_2 \S + k I) v$.
By setting $A \triangleq (\ve_2 \S + k I)$ and $u \triangleq A^{1/2} v$, we can rewrite the constraints as:
$$\frac{ |u^T \left(A^{-1/2}\right)^T (\S-\hat \S) A^{-1/2} u| } {u^T u} \le 1$$
This is equivalent to checking whether the maximum eigenvalue of the matrix $\left(A^{-1/2}\right)^T (\S-\hat \S) A^{-1/2}$ is greater than $1$.

At this point, we have efficiently computed a solution $(\hat \m, \hat \S) \in \textrm{conv} \left( {\cal{M}} \right)^{\oplus n}$ that satisfies the closeness guarantees and we need to convert it to a solution in the set ${\cal{M}}^{\oplus n}$ that is also appropriately close to $(\m, \S)$ and obtain a GMD with the guarantees of the lemma. By the Shapley-Folkman theorem, it holds that 
$\textrm{conv} \left( {\cal{M}} \right)^{\oplus n} = {\cal{M}}^{\oplus (n-k^2-k)} \oplus \textrm{conv} \left( {\cal{M}} \right)^{\oplus (k^2+k)}$ since ${\cal{M}} \subset \mathbb{R}^{k^2+k}$. We can greedily construct such a solution by iteratively picking points $m_i \in {\cal{M}}$ for $i = 1, \dots,  (n-k^2-k)$ such that $\left( (\hat \m, \hat \S) - \sum_{j = 1}^i m_i \right) \in \textrm{conv} \left( {\cal{M}} \right)^{\oplus (n-i)}$. The Shapley-Folkman theorem for the space $\textrm{conv} \left( {\cal{M}} \right)^{\oplus (n-i)}$, guarantees that for all $i \le (n-k^2-k)$, a point $m_i$ with the required property always exists. Since membership in $\textrm{conv} \left( {\cal{M}} \right)^{\oplus (n-i)}$ can be checked efficiently by writing a linear program similar to the one above, we can efficiently run the above process to generate $(n-k^2-k)$ CRVs. For the remaining $k^2+k$ CRVs, we arbitrarily choose points $m_{n-k^2-k+1},...,m_n \in {\cal{M}}$ to obtain a complete $(n,k)$-GMD $\r'$. We argue next that this GMD satisfies the conditions required by the lemma. 

For any $m,m' \in \textrm{conv} \left( {\cal{M}} \right)$, it holds that $\|m - m'\|_{\infty} \le 1$. Moreover, $(\m^{\r'}, \S^{\r'}) = \sum_{i=1}^n m_i$ and $(\hat \m, \hat \S) = \sum_{i=1}^{(n-k^2-k)} m_i + \sum_{i=1}^{k^2+k} m'_i$. This implies that $\| \m^{\r'}-\hat \m \|_{\infty} \le k^2+k $ and $\| \S^{\r'}-\hat \S \|_{\max} \le k^2+k $. We have that:

\begin{align*}
  |v^T (\m-\m^{\r'})| &\le |v^T (\m-\hat \m)| + |v^T (\m^{\r'}-\hat \m)| \\
  &\le \ve_1 \sqrt{v^T \S v} + \sqrt{k} \|v\|_2 + \|v\|_2 \|\m^{\r'}-\hat \m\|_2  \\
  &\le \ve_1 \sqrt{v^T \S v} + \sqrt{k} \|v\|_2 + (k^2+k) \sqrt{k} \|v\|_2  \\
  &= \ve_1 \sqrt{v^T \S v} + 3k^{2.5} \|v\|_2
\end{align*}

Similarly, $|v^T   (\S-\S^{ \r'}) v| \le |v^T   (\S-\hat \S) v| + |v^T   (\hat \S-\S^{\r'}) v| \le \ve_2 v^T \S v + (k^3+k^2+k^1) \|v\|_2^2.$

\section{A Lower Bound for Covers of PMDs}
\label{sec:body_lowerb}
In this section, we discuss Theorem~\ref{thm:lower-bound}, the lower bound on the size of any $\epsilon$-cover of $(n,k)$ PMDs.
This theorem shows that  it is not possible to get significant improvement on the cover size obtained in Theorem~\ref{thm:propercover}. In particular, the dependence of the size of the cover on $1/\epsilon$ is tight up to a difference of $3$ in the exponent of $\log(1/\epsilon)$. 

It turns out that it is easy to prove a dependence of $O(n^k)$ on the size of any $\epsilon$-cover and  most of the work is involved in showing a lower bound of $T(k,\epsilon)  = 2^{\log^{k-1}(1/\epsilon)}$ on the cover size. Thus, in this overview we only focus on the machinery required to show the lower bound of $T(k,\epsilon)$ on the $\epsilon$-cover size. We remark that prior to our work, for $k=2$ (i.e. PBDs), Diakonikolas, Kane, and Stewart obtained a lower bound of $2^{\log^2(1/\epsilon)}$~\cite{DiakonikolasKS16a}. 

Showing the lower bound on the cover size is equivalent to showing the existence of $T(k,\epsilon)$-many $(n_0,k)$-PMDs which are all $\epsilon$-far from each other where $n_0 \le n$. The usual difficulty in showing cover size lower bounds, is that even if the parameters specifying two PMDs are significantly different, it is not necessarily true that the resulting PMDs are far in total variation distance. In fact, directly arguing that two PMDs are far apart in total variation distance seems difficult. Instead, our strategy is to carefully pick a family of $T(k,\epsilon)$ PMDs and show that for any two distinct PMDs in this set, there is at least one ($k$-dimensional) moment {$\alpha \in \mathbb{Z}^{+k}$} of size $O(\log(1/\epsilon))$ such that the $\alpha^{th}$ moment of the two PMDs are $\epsilon$-far from each other (by size of the moment $\alpha$, we mean $\Vert \alpha \Vert_1$.). 
Usually, gap in moments for two distributions need not translate to significant gap in total variation distance. However, in our setting, we can choose $n_0 \approx \log^k(1/\epsilon)$. Since $n_0$ is small, it is easy to show that if two PMDs differ by $\epsilon$ in one of their moments of size $O(\log (1/\epsilon))$, then they are $\approx \epsilon$ far in total variation distance{\ifnum\camr=0~(Claim~\ref{fac:sa})\fi}.

Note that the $\alpha^{th}$-moment of a PMD  is a multisymmetric polynomial in the parameters of the PMD (i.e. invariant under  permuting its  summands).  Next consider the multidimensional multisymmetric polynomial map where each coordinate in the range corresponds to a moment  of size $O(\log(1/\epsilon))$. Since there are roughly $\Theta_k(\log^k(1/\epsilon))$
moments of size $O(\log(1/\epsilon))$, the dimension of the map is $\Theta_k(\log^k(1/\epsilon))$. 
The problem of showing lower bounds on the cover size is now  equivalent  to showing that the range of this map contains $T(k,\epsilon)$-many points which are $\epsilon$-far from each other. In other words, we need a way to show a lower bound on the metric entropy of this polynomial map. Such problems are usually treated with tools of algebraic geometry and we adopt the same strategy. In particular, rather than directly working over the reals, we change the domain to a finite field  $\mathbb{F}$ of appropriate size and consider the corresponding polynomial map  in $\mathbb{F}$. Once we are in $\mathbb{F}$, we apply an extension of B\'ezout's theorem due to Wooley~\cite{Wooley96}{\ifnum\camr=0~(Theorem~\ref{thm:Wooley})\fi} to show that this map has a large number of points in its range when the underlying domain is $\mathbb{F}${\ifnum\camr=0~(Lemma~\ref{lem:algfinite})\fi}. Because of the special structure of the polynomials involved, it is possible to show that the presence of a large range in a finite field corresponds to an appropriate lower bound on the metric entropy of the map. We remark that the application of B\'ezout's theorem in our context is not straightforward. In particular, to apply the theorem, one needs to reason about the Jacobian of this polynomial map.
Despite being a very natural family of maps, to the best of our knowledge, properties of the corresponding Jacobian have not been previously investigated.
\ifnum\tp=1
For more details on our lower bound, see 
\ifnum\camr=0
Section \ref{sec:lowerb}.
\else
the full version of this paper.
\fi
\fi

\ifnum\tp=1
\section{Details from Section \ref{sec:body_lowerb}}
\fi
\label{sec:lowerb}
\subsection{Details}
We provide the proof of Theorem~\ref{thm:lower-bound}. The proof will use algebraic geometric tools to argue this fact. In particular, the main theorem we will prove will be the following: 
\begin{theorem}\label{thm:lb1}
There are $(m,k)$-PMDs $Z_1, \ldots, Z_\ell$ such that for all $1 \le i < j \le \ell$, $d_{TV}(Z_i, Z_j) \ge \epsilon$ and $\ell = 2^{\tilde{\Omega}(\log^{k-1}(1/\epsilon))})$ where  $m=O(\log^{k-1}(1/\epsilon))$. 
\end{theorem}
We will now prove Theorem~\ref{thm:lower-bound} using Theorem~\ref{thm:lb1}. 
\begin{proof}
Note that by assumption $n > 2m$. It is easy to observe that for any $\alpha= (\alpha_1, \ldots, \alpha_k) \in \mathbb{Z}^k$ such that $\sum \alpha_i = n/2$, there are $k$ CRVs $X_1, \ldots, X_{n/2}$ such that $X_1 + \ldots +X_{n/2}$ is supported on $\alpha$. Now, consider any $\alpha , \beta  \in \mathbb{Z}^k$ such that $\Vert \alpha - \beta \Vert_1 > m$. Then, for any $(m,k)$ PMD $Z_i, Z_j$, the supports of $Z_i + \alpha$ and $Z_j + \beta$ are disjoint. It is now easy to see that we can choose 
$L = (\frac{n}{2 m})^{k-1}$ points $\alpha^{(1)}, \ldots, \alpha^{(L)}$ such that for $1 \le j \le L$, $\sum_{i=1}^k \alpha^{(j)}_i = n/2$ and $\Vert \alpha^{(j)}-\alpha^{(\ell)} \Vert_1 \ge m$ whenever $ j \not = \ell$. Now, let $Z_{i_1}$ and $Z_{i_2}$ be two $(m,k)$ PMDs from Theorem~\ref{thm:lb1}. Then, both $Z_{i_1} +  \alpha^{(j)}$ and $Z_{i_2} +  \alpha^{(\ell)}$ are $(n,k)$ PMDs and further, $d_{TV}(Z_{i_1} +  \alpha^{(j)}, Z_{i_2} +  \alpha^{(\ell)}) \ge \epsilon$. This gives a set of $L \cdot 2^{\tilde{\Omega}(\log^{k-1}(1/\epsilon))})$ $(n,k)$-PMDs which are $\epsilon$-far from each other. 
\end{proof}
Thus, it remains to prove Theorem~\ref{thm:lb1}. The proof of this theorem shall involve a combination of ideas using combinatorics of multisymmetric polynomials and tools from algebraic geometry. In particular, instead of directly arguing about total variation distance of PMDs, we will argue about the moments of PMDs. We first observe that for any $(m,k)$ PMD $Z$, we can associate a  matrix $P_Z \in \mathbb{R}^{m \times (k-1)}$ where the entries of the matrix are non-negative such that the entries of any row sum to at most $1$. The semantics of the matrix are that $Z = X_1 + \ldots + X_n$ where each $X_i$ is an independent CRV with $\Pr[X_i=\mathbf{e}_j] = P_Z[i,j]$ (if $1 \le j <k$) and $\Pr[X_i = \mathbf{e}_k] = 1 - \sum_{j<k} P_Z[i,j]$. Clearly, the distribution of $Z$ is invariant under permuting the rows of $P_Z$. Further, up to permutations of columns, the matrix $P_Z$  associated with such a $Z$ is unique. 
~\\
If $Z$ is a $(m,k)$ PMD and $\alpha \in \mathbb{Z}^{+k}$, then $M_{\alpha}(Z)  = \mathbf{E}[Z^{\alpha}]$ i.e. the $\alpha^{th}$ moment of $Z$. Here $Z^{\alpha}$ is an abbreviation of $\prod_{i=1}^k Z_i^{\alpha_i}$ where $Z_i$ the $i^{th}$ component of $Z$. Thus, $M_{\alpha}(Z)$ is  the $\alpha^{th}$ moment of $Z$. The following will be a very useful observation. 
\begin{observation}
$M_{\alpha}(Z)$ is a multisymmetric polynomial of degree $\Vert \alpha \Vert$ in the variables $\{P_Z[i,j]\}$ where $ 1 \le i \le m$ and $1 \le j < k$. By multisymmetric, we mean the polynomial is invariant under permuting the rows of $P_Z$. 
\end{observation}
While the moment $M_{\alpha}(Z)$ is very natural to consider, for the purposes of proving Theorem~\ref{thm:lb1}, it will be useful to define two more families of multisymmetric polynomials (in the variables $\{P_Z[i,j]\}_{1\le i \le m, 1 \le j < k}$). These polynomials will be the elementary multisymmetric polynomials and the power sum multisymmetric polynomials. 
Also, from now, we will only restrict our attention to $\alpha$ where the last entry is $0$. This is because, for any $\beta \in \mathbb{Z}^{+k}$, $M_{\beta}(Z)$ can be expressed as a polynomial combination of $M_{\alpha}(Z)$ where $\alpha_k=0$. 
\begin{definition}
Let $\alpha \in \mathbb{Z}^{+k}$. Let $\beta \in \mathbb{Z}^{\Vert \alpha \Vert}$ such that the entry $i$ occurs exactly $\alpha_i$ times in $\beta$. While there are multiple choices for such a $\beta$, any canonical choice is good enough. Let $1 \le i_1< \ldots < i_{\Vert \alpha \Vert} \le m$. Then, 
$$
\mathbf{E}_{\alpha}(Z) = \sum_{\sigma : \{1, \ldots, \Vert \alpha \Vert\} \rightarrow [m]} \prod_{j=1}^{\Vert \alpha \Vert} P_Z[\sigma(i_j),\beta_j],
$$
where the sum is taken over all surjections $\sigma$ from $\{1, \ldots, \Vert \alpha \Vert \} \rightarrow [m]  $ where for $i<j$, if $\beta_i = \beta_j$, then $\sigma(i)<\sigma(j)$. 
\end{definition}
The definition above might look a little involved, so to illustrate the concept, we will consider a simple example. For example, if $k=3$ and $\alpha = (1,2,0)$, then 
$$
\mathbf{E}_{\alpha} (Z) = \sum_{j<k, i \not = j,k} P_Z[i,1] P_Z[j,2] P_Z[k,2] . 
$$
Likewise, if $k=4$ and $\alpha = (1,1,1,0)$, then 
$$
\mathbf{E}_{\alpha}(Z) = \sum_{i,j,k \textrm{ all are distinct}} P_Z[i,1] P_Z[j,2] P_Z[k,3] . 
$$
Another family of polynomials which will be very useful in our reasoning will be the family of power sum multisymmetric polynomials. 
\begin{definition}
Let $\alpha \in \mathbb{Z}^{+k}$. Then, 
$$
\mathbf{P}_{\alpha}(Z) = \sum_{1 \le i \le m} \prod_{j=1}^m P_Z[i,j]^{\alpha_j}. 
$$
\end{definition}
We remark that the polynomials we have defined so far i.e. $M_{\alpha}(\cdot)$, $\mathbf{E}_{\alpha}(\cdot)$ and $\mathbf{P}_\alpha(\cdot)$ are well-defined formal polynomials and make sense even if the matrix $P_Z[i,j]$ has entries from field $\mathbb{F}$ whose characteristic is not zero.  This shall be useful for us going forward. For the moment, we will consider the connection between the family $\mathbf{E}_{\alpha}$ and $M_{\alpha}$. We will require the following definition. For two vectors, $\alpha, \beta \in \mathbb{Z}^{+k}$, we say $\alpha \preceq \beta$ if $\alpha_i \le \beta_i$ for $1 \le i \le k$. The first is the following observation. 
\begin{observation}
For any $(m,k)$-PMD $Z$ and $\alpha \in \mathbb{Z}^{+k}$ (with $\alpha_k=0$), $M_{\alpha}(Z)$ can be expressed as a linear combination of $\mathbf{E}_{\beta}(Z)$ where $\beta \preceq \alpha$. 
\end{observation}
\begin{proof}
As we have already observed, $M_{\alpha}(Z)$ is a multisymmetric polynomial in entries of the matrix $P_Z$ at degree is at most $\Vert \alpha \Vert$. To prove that is a linear combination of $\mathbf{E}_{\beta}$ for $\beta \preceq \alpha$, it suffices to make the following observation: Note that $M_{\alpha}(Z)$ is
$$
M_{\alpha}(Z) = \mathbf{E} \left[ \prod_{j=1}^{k-1} \big( \sum_{i=1}^m X_{i,j}\big)^{\alpha_j} \right]. 
$$
Observe that any monomial  where $X_{i,j}$ and $X_{i, \ell}$ appear together with $j \not = \ell$ vanishes under the expectation. Likewise, since $X_{i,j}$ is supported on $\{0,1\}$, hence
$X_{i,j}^{\ell} = X_{i,j}$ for any $\ell \ge 1$. These two observations coupled with each other imply that $M_{\alpha}$ is a linear combination of $\mathbf{E}_{\beta}(\cdot)$ for $\beta \preceq  \alpha$. 
\end{proof}
The next lemma implies bounds on the coefficients of $\mathbf{E}_{\beta}$ in expressing $M_{\alpha}$.  Let us now assume that $M_{\alpha} = \sum_{\beta \preceq \alpha} \gamma_{\beta} \cdot \mathbf{E}_{\beta}$. It is easy to see the following claim. 
\begin{claim}\label{clm:coeff1}
For $\alpha $ with $\alpha_k=0$, we have 
$\gamma_{\alpha} = \prod_{j=1}^{k-1} \alpha_j! $. 
\end{claim}
The next claim is also fairly easy to prove. 
\begin{claim}\label{clm:bound1}
$\sum_{\beta \preceq \alpha} |\gamma_{\beta}| = 2^{O(k)} \cdot \prod_{j=1}^{k-1} \alpha_j!$. 
\end{claim}
\begin{proof}
Let $c_1, \ldots, c_{k-1} \in \mathbb{Z}^{+k}$ be defined as the following: 
\[
c_j = \alpha \cdot I_j. 
\]
In other words, $c_j$ is obtained by a pointwise product of $\alpha$ and the indicator vector of the singleton set $\{j\}$. Now, assume that for $1 \le j \le k-1$, 
\[
M_{c_j}(Z) = \sum_{\beta_j \preceq c_j} \gamma_{\beta_j} \cdot \mathbf{E}_{\beta_j}. 
\]
Then, it is not difficult to see that 
\[
M_{\alpha}(Z) = \sum_{\beta_1 \preceq c_1} \ldots  \sum_{\beta_{k-1} \preceq c_{k-1}}  \mathbf{E}_{\beta} \cdot \prod_{j=1}^{k-1} \gamma_{\beta_j} ,
\]
where $\beta = \beta_1  + \ldots + \beta_{k-1}$. 
Thus, to prove our claim, it suffices to show that for any particular $1 \le j \le k-1$, 
\[
\sum_{\beta_j \preceq c_j} |\gamma_{\beta_j}| = O(1) \cdot \alpha_j!. 
\]
Note that $c_j$ is just $\alpha_j$ at the $j^{th}$ position and zero everywhere else. We introduce the following notation: For any integer $k$, we let $\mathcal{P}(k)$ denote the set of its partitions i.e. a tuple of strictly positive integers summing to $k$ ordered in decreasing sequence. For example, for $k=5$, we have 6 distinct partitions $(5), (4,1), (3,2), (3,1,1), (2,2,1) , (2,1,1,1)$. For any partition $P \in \mathcal{P}(k)$, we use $s(P)$ to denote the number of summands in $P$. For example, for the partition $P=(3,2)$, $s(P)=2$. Further, if $P = (x_1, \ldots, x_k)$ is a partition of $n$, then 
$$
\binom{n}{P} = \binom{n}{x_1 \ \ldots  \ x_k}. 
$$
With these notations in place, it is easy to see that 
\[
\gamma_{\beta_j} = \sum_{P \in \mathcal{P} (\alpha_j): s(P) = \beta_j} \binom{\alpha_j}{P}
\]
Now, it is easy to see that for any integer $x>1$, $(1.4)^x <= x!$. If $1(P)$ denotes the number of $1$ in the partition $P$. With this, we have  
\[
|\gamma_{\beta_j}| \le  \sum_{P \in \mathcal{P} (\alpha_j): s(P) = \beta_j}  \frac{\alpha_j!}{(1.4)^{\alpha_j}} \cdot 1.4^{1(P)}. 
\]
Thus, implies that
\[
\sum_{\beta_j \preceq c_j} |\gamma_{\beta_j}| \le  \sum_{\beta_j \preceq c_j} \sum_{P \in \mathcal{P} (\alpha_j): s(P) = \beta_j}  \frac{\alpha_j!}{(1.4)^{\alpha_j}} \cdot 1.4^{1(P)}. 
\]
Now, note that the total number of partitions of $\alpha_j$ with $t$ ones in it is upper bounded by $|\mathcal{P}(\alpha-t)|$. However, it is a well-known fact in number theory, that
$|\mathcal{P}(\alpha-t) | \le 2^{O(\sqrt{\alpha-t})}$. Thus, 
\[
\sum_{\beta_j \preceq c_j} |\gamma_{\beta_j}| \le  \sum_{x=0}^{\alpha_j} \frac{\alpha_j!}{(1.4)^{\alpha_j}} \cdot 1.4^{x} \cdot 2^{O(\sqrt{\alpha_j -x})} \le \alpha_j ! \cdot \int_{x\ge 0}  1.4^{O(\sqrt{x}) - x} dx = O(\alpha_j!).
\]
This finishes the proof. 
\end{proof}
Thus, using the last two claims, we infer that there is a linear map which given any $\alpha \in \mathbb{Z}^{+k}$ (with $\alpha_k=0$), maps the set $\{ \mathbf{E}_{\beta}(Z)  \}_{\beta \preceq \alpha}$ to the set $\{M_{\beta}(Z) \}_{\beta \preceq \alpha}$. The next lemma bounds the condition number of this map. 
\begin{lemma}\label{lem:moment1}
Let $Z$ and $Z'$ be two $(m,k)$-PMDs such that $|\mathbf{E}_{\beta}(Z) - \mathbf{E}_{\beta}(Z')| \ge \delta$. Then, there exists $\beta_0 \preceq \beta$ such that 
$$
|M_{\beta_0}(Z) -M_{\beta_0}(Z')| \ge \delta \cdot c^{-d}. 
$$
where $c = 2^{O(k)}$ is the constant appearing in Claim~\ref{clm:bound1}.
\end{lemma}
\begin{proof}
Let $c = 2^{O(k)}$ be the constant appearing in Claim~\ref{clm:bound1}. Let $|\beta| = d$ and $i$ be the smallest integer such that there exists a $\beta_0 \preceq \beta$ 
with $|\beta_0| = i$ and 
\[
|\mathbf{E}_{\beta_0}(Z) - \mathbf{E}_{\beta_0}(Z')| \ge \delta \cdot (2c)^{i-d}. 
\]
Note that by assumption, there exists such a $\beta_0$. Next, 
\begin{eqnarray*}
|M_{\beta_0}(Z) -M_{\beta_0}(Z')|  &=&\big| \sum_{\kappa \preceq \beta_0} \gamma_{\kappa} \cdot  \big(\mathbf{E}_{\kappa}(Z) - \mathbf{E}_{\kappa}(Z') \big) \big| \\
&\ge& \gamma_{\beta_0} \cdot \big| \big(\mathbf{E}_{\beta_0}(Z) - \mathbf{E}_{\beta_0}(Z') \big) \big|  - \big| \sum_{\kappa \prec \beta_0} \gamma_{\kappa} \cdot  \big(\mathbf{E}_{\kappa}(Z) - \mathbf{E}_{\kappa}(Z') \big) \big|  \\
\end{eqnarray*}
Applying Claim~\ref{clm:coeff1} and Claim~\ref{clm:bound1}, we get that 
\[
|M_{\beta_0}(Z) -M_{\beta_0}(Z')|  \ge \prod_{i=1}^{k-1} \beta_{0,i}! \bigg(  \big| \big(\mathbf{E}_{\beta_0}(Z) - \mathbf{E}_{\beta_0}(Z') \big) \big|\bigg) -  c \cdot \prod_{i=1}^{k-1} \beta_{0,i}! \max_{\kappa \preceq \beta}  \big| \big(\mathbf{E}_{\kappa}(Z) - \mathbf{E}_{\kappa}(Z') \big) \big|
\]
Again applying the hypothesis on $\beta_0$, we have
\begin{eqnarray*}
|M_{\beta_0}(Z) -M_{\beta_0}(Z')|  &\ge& \prod_{i=1}^{k-1} \beta_{0,i}! \cdot \delta \cdot (2c)^{\Vert \beta_0\Vert  -d} - c \cdot \prod_{i=1}^{k-1} \beta_{0,i}! \cdot \delta \cdot (2c)^{\Vert \beta_0\Vert  -d -1 }  \\
&\ge& \prod_{i=1}^{k-1} \beta_{0,i}! \cdot \delta \cdot c^{\Vert \beta_0\Vert  -d } \ge \delta  \cdot c^{-d}. 
\end{eqnarray*}
\end{proof}
The strategy for the rest of the proof is as follows: Instead of showing Theorem~\ref{thm:lb1}, we will show the following lemma. 
\begin{lemma}\label{lem:moment}
There are $(m,k)$-PMDs $Z_1, \ldots, Z_\ell$ such that  $\ell = 2^{\tilde{\Omega}(\log^{k-1}(1/\epsilon))})$,  $m=O(\log^{k-1}(1/\epsilon))$ and for every $1 \le i< j \le \ell$, 
there exists some $\alpha \in \mathbb{Z}^{+k}$ such that $\alpha_k=0$, $\Vert \alpha \Vert = \tilde{O}(\log(1/\epsilon))$ and $|M_{\alpha}(Z_i) - M_{\alpha}(Z_j)| \ge \epsilon$. 
\end{lemma}
To see why it suffices to prove Lemma~\ref{lem:moment}, we have the following claim. 
\begin{claim}\label{fac:sa}
Let $Z_1, Z_2$  be two $(m,k)$-PMDs and $\alpha \in \mathbb{Z}^k$ such that $\big| M_{\alpha}(Z_1) - M_{\alpha}(Z_2) \big|\ge \delta$. Then,
$d_{TV}(Z_1,Z_2) \ge \delta \cdot m^{-\Vert \alpha \Vert_1}$. 
\end{claim} 
\begin{proof}
Assume towards a contradiction that $d_{TV}(Z_1, Z_2) < \delta \cdot m^{-\Vert \alpha \Vert_1}$. By definition, this means that there is a coupling $(Z'_1, Z'_2)$ such that the marginal $Z'_1$ is distributed as $Z_1$, the marginal $Z'_2$ is distributed as $Z_2$ and $\Pr[Z'_1 \not = Z'_2] < \delta \cdot m^{-\Vert \alpha \Vert_1}$. 
As the support of both $Z_1$ and $Z_2$ is confined in the box $[0,m]^{k}$, it easily follows that $
\big| M_{\alpha}(Z_1)- M_{\alpha}(Z_2)\big| < \Pr[Z'_1 \not = Z'_2] \cdot m^{\Vert \alpha \Vert_1} < \delta$. This results in a contradiction, thus completing the proof. 
\end{proof}

In light of Lemma~\ref{lem:moment1}, it instead suffices to prove the following lemma. 
\begin{lemma}\label{lem:alg1} 
There are $(m,k)$-PMDs $Z_1, \ldots, Z_\ell$ such that  $\ell = 2^{\tilde{\Omega}(\log^{k-1}(1/\epsilon))})$,  $m=O(\log^{k-1}(1/\epsilon))$ and for every $1 \le i< j \le \ell$, 
there exists some $\alpha \in \mathbb{Z}^{+k}$ such that $\alpha_k=0$, $\Vert \alpha \Vert = \tilde{O}(\log(1/\epsilon))$ and $|\mathbf{E}_{\alpha}(Z_i) - \mathbf{E}_{\alpha}(Z_j)| \ge \epsilon$. 
\end{lemma}
The rest of the proof is towards proving Lemma~\ref{lem:alg1}. As we said before, the proof is going to involve use of algebraic geometry tools. In fact, to prove Lemma~\ref{lem:alg1}, instead of considering the matrices $P_Z$ to be real-valued matrices, we will instead first show an equivalent version of Lemma~\ref{lem:alg1} over a finite field $\mathbb{F}$ of appropriate size. This change to finite fields will make it easier to apply tools of algebraic geometry. In particular, we will prove the following lemma. 
\begin{lemma}\label{lem:algfinite}
For any integer $d \in \mathbb{N}$ and any finite field $\mathbb{F}$ of size $2 \cdot d  \cdot m(k-1)$, there are $\ell$ matrices $A_1, \ldots, A_{\ell}$  in $\mathbb{F}^{m \times (k-1)}$ where $\ell = 2^{\tilde{\Omega}(\log^{k-1}(1/\epsilon))})$, $m=O(\log^{k-1}(1/\epsilon))$ and for $1 \le i < j \le \ell$, there exists $\alpha \in \mathbb{Z}^{+k}$ where $\Vert  \alpha \Vert \le d$, $\alpha_k=0$ and 
\[
\mathbf{E}_{\alpha}(A_i) \not = \mathbf{E}_{\alpha}(A_j). 
\] 
\end{lemma}
Before, we prove Lemma~\ref{lem:algfinite}, let us see why it implies Lemma~\ref{lem:alg1}. To get PMD matrices from the matrices $A_1, \ldots, A_{\ell}$, we use the following map. 
\[
A_{i} \mapsto P_{Z_i} \ \textrm{where} \ P_{Z_i} [j, j' ] = \begin{cases} \frac{A_{i}[j,j']}{2 k |\mathbb{F}|} & \ \textrm{if}\  j'<k \\
              1 - \sum_{j''<k} P_{Z_i} [j, j'' ]& \ \textrm{if} \ j= k  \end{cases}
\]
It is easy to see that this operation defines legitimate $(m,k)$ PMD matrices. Further, note that $\mathbf{E}_{\alpha}$ is a homogenous polynomial of degree $\Vert  \alpha \Vert$. 
Thus, it is easy to see that if $\mathbf{E}_{\alpha}(A_i) \not = \mathbf{E}_{\alpha}(A_j)$, then 
\[
\big| \mathbf{E}_{\alpha}(Z_i) - \mathbf{E}_{\alpha}(Z_j) \big| \ge \frac{1}{|\mathbb{F}|^{\Vert \alpha \Vert}}.
\]
By choosing $|\mathbb{F}|$ to be a field of size $O( k \cdot \log(1/\epsilon))$, we immediately see that it implies the bounds in Lemms~\ref{lem:alg1}.  Thus, all that remains to be proven here is Lemma~\ref{lem:algfinite}. To prove this, let us set $d = \tilde{O}(\log(1/\epsilon))$ and let $\mathcal{S} = \{ \alpha \in \mathbb{Z}^{+k} : \alpha_k=0 \ \textrm{ and } \Vert \alpha \Vert \le d\}$. 
We now define the map $\mathbf{E}_{\mathcal{S}}: \mathbb{F}^{m \times  (k-1)} \rightarrow \mathbb{F}^{\mathcal{S}}$ which is a multidimensional map indexed by $\mathcal{S}$ where the coordinate for $\alpha \in \mathbb{Z}^{+k}$ is $\mathbf{E}_{\alpha} (\cdot)$. Note that Lemma~\ref{lem:algfinite} amounts to showing a lower bound on the entropy of this map. 
%
%
~\\
We will need the notion of Jacobian of a map which is defined next. 
\begin{definition}
Let $\mathbb{F}$ be any field and let $M: \mathbb{F}^n \rightarrow \mathbb{F}^m$. Then, the Jacobian of $M$, denoted by $J_M$ is the $m \times n$ matrix in $F(x_1, \ldots, x_n)$ where
the $(i,j)^{th}$ entry is given by $\partial M_i(x_1,\ldots, x_n)/\partial _j$ where $M_i$ denotes the $i^{th}$ coordinate of the map $M$. 
\end{definition}
~\\ For us, the utility of Jacobian will come from its role in the following theorem. 
\begin{theorem}\label{thm:Wooley}~\cite{Wooley96}
Let $\mathbb{F}$ be a prime field of size $p$. Let $k$ and $d$ be integers. Let $M: \mathbb{F}^s \rightarrow \mathbb{F}^s$ be such that any coordinate is a polynomial map of degree at most $d$. 
 For $a \in \mathbb{F}^s$, let
\[
N_a = \big| \big\{ c \in \mathbb{F}^s:  M(c)=a \ \textrm{and} \ J_M(c) \not =0  \big\}\big|
\]
Then, for every $a \in \mathbb{F}^s$, $N_a \le d^s$. 
\end{theorem}
~\\ As a consequence, we have the following corollary. 
\begin{corollary}\label{corr:Wooley}
Let $\mathbb{F}$ be a prime field of size $p$. Let $s$ and $d$ be integers. Let $M: \mathbb{F}^{s'} \rightarrow \mathbb{F}^s$ be such that any coordinate is a polynomial map of degree at most $d$. Let us assume that $\mathsf{rank}( J_M ) =s$. If $|\mathbb{F}|>2 \cdot d \cdot s$, then, $|\mathsf{Range}(M)| \geq |\mathbb{F}|^s /2 d^s$. 
\end{corollary}
\begin{proof}
Since $\mathsf{rank}(J_M) =s$, it means that there is a submatrix of size $s \times s$ (call it $J'_M$) such that $\det J_{M'} \not = 0$ (here the determinant is evaluated over the field of rational functions over $\mathbb{F}$). Let the $s$ columns correspond to the set of variables $\mathcal{L}$. Since the determinant is a low degree polynomial (of degree at most $d \cdot s$ in each variable), hence there is choice  of the variables outside $\mathcal{L}$ to some values in $\mathbb{F}$ such that $\det J_{M'} \not =0$. 
~\\
For this setting of variables, let $M': \mathbb{F}^{\mathcal{L}} \rightarrow \mathbb{F}$ be the map restricted to the variables in $\mathcal{L}$. 
Since $ \det J_{M'} \not =0$, hence $J_{M'}$ is a non-zero polynomial of degree at most $d \cdot s$. Since $|\mathbb{F}| >2 \cdot d \cdot s$, hence by Schwartz-Zippel lemma, 
the set $C = \{c : \det J_M(c) \not =0 \}$ has size at least $|\mathbb{F}|^s/2$. Applying Theorem~\ref{thm:Wooley}, we get the stated claim. 
\end{proof}
To show a lower bound on the entropy of $M_{\mathcal{S}}$, we will apply Corollary~\ref{corr:Wooley}. To apply this, we need to prove that $\det J_{M_{\mathcal{S}}} \not = 0$.  
It is not clear how to show this, so we introduce an intermediate map. 
\begin{definition} For $\alpha \in \mathbb{Z}^k$ (with $\alpha_k=0$), we define the map $\mathbf{P}_{\alpha} : \mathbb{F}^{m \cdot (k-1)} \rightarrow \mathbb{F}$ where $$ \mathbf{P}(A) \mapsto \sum_{i=1}^n \prod_{j=1}^{k-1} A[i,j]^{\alpha_j} .$$ The family $\{ \mathbf{P}_{\alpha} \}$ is usually referred to as power-sum multisymmetric polynomials. 
\end{definition}
The idea here will be that we will relate the family $\mathbf{P}_{\alpha}$ and $\mathbf{E}_{\alpha}$ and then argue about the Jacobian of a map defined in terms of $\mathbf{P}_{\alpha}$. The following relation between $\mathbf{E}_{\alpha}$ and $\mathbf{P}_{\alpha}$  was established in Dalbec~\cite{Dalbec99}. 

\begin{proposition}
$$
\Vert \alpha \Vert \cdot \mathbf{E}_{\alpha}  + \sum_{ {\alpha = \beta + \gamma,} \atop{ \beta, \gamma \not =0} } (-1)^{\Vert \beta \Vert} \binom{\Vert \beta \Vert}{\beta} \mathbf{P}_\beta \cdot \mathbf{E}_\gamma + (-1)^{\Vert \alpha \Vert} \binom{\Vert \alpha \Vert}{\alpha} \mathbf{P}_{\alpha} =0. 
$$
\end{proposition}
Using induction, the following lemma is immediate. 
\begin{lemma}\label{lem:rel1}
For any $\alpha$, if either the characteristic of $\mathbb{F}$ is $0$ is more than $\Vert \alpha \Vert$,  there exist $\mathbf{Q}_{\alpha}$ such that 
$$
\mathbf{P}_{\alpha} = \mathbf{Q}_{\alpha} \bigg( \bigg\{ \mathbf{E}_{\beta}\bigg\}_{\beta \preceq \alpha} \bigg)
$$
\end{lemma}
Rather than considering the map $\mathbf{E}_{\alpha}$, we will consider a restricted version of it. In particular, choose some matrix $A \in \mathbb{F}^{m \cdot (k-2)}$. The exact choice will be specified later. However, given $x \in \mathbb{F}^m$, we can consider a matrix $A_x \in \mathbb{F}^{m \cdot (k-1)}$ which is obtained by concatenating $x$ with $A$ where $x$ is the last row of $A_x$ whereas the first $(k-1)$ rows are formed by $A$. Thus, fixing this choice of $A$, for every $\alpha \in \mathbb{Z}^{+k}$ with $\alpha_k=0$, we can define the map 
$$
\mathbf{E'}_{\alpha} : x \mapsto  \mathbf{E}_{\alpha} (A_x). 
$$
Our aim will be to argue that the map $\mathbf{E'}_{\mathcal{S}}: \mathbb{F}^m \rightarrow \mathbb{F}^{\mathcal{S}}$ (defined analogously to $\mathbf{E}_{\alpha}$) has full rank and thus Corollary~\ref{corr:Wooley} is applicable here. Note that, we can also define the map $\mathbf{P'}_\alpha: \mathbb{F}^{m} \rightarrow \mathbb{F}$ and  $\mathbf{P'}_{\mathcal{S}}: \mathbb{F}^{m} \rightarrow \mathbb{F}^{\mathcal{S}}$ analogously. 
As a consequence of Lemma~\ref{lem:rel1}, we have that for $\mathbf{P'}_{\alpha}(z_1, \ldots, z_m)$ and for $1 \le j \le m$, 
$$
\frac{\partial \mathbf{P'}_{\alpha}}{\partial z_j}= \sum_{\gamma \preceq \alpha}\frac{\partial  \mathbf{Q}_{\alpha} \bigg( \bigg\{ \mathbf{E'}_{\beta}^{}\bigg\}_{\beta \preceq \alpha} \bigg)}{\partial \mathbf{E'}_{\gamma}^{}} \cdot \frac{\partial \mathbf{E'}_{\gamma}^{}}{\partial z_j}.
$$
Consider the field $\mathbb{K} = \mathbb{F}(z_1, \ldots, z_m)$ (i.e. the field of rational functions over $\mathbb{F}$ in the variables $z_1,\ldots, z_m$). If $J_{\mathbf{P'}_{\mathcal{S}}}$ and $J_{\mathbf{E'}_{\mathcal{S}}}$ are the Jacobians of the maps ${\mathbf{P'}_{\mathcal{S}}}$ and ${\mathbf{E'}_{\mathcal{S}}}$ respectively, then this immediately implies that there exists a matrix $B \in \mathbb{K}^{\mathcal{S} \times m}$ such that 
\[
J_{\mathbf{P'}_{\mathcal{S}}} = B \cdot J_{\mathbf{E'}_{\mathcal{S}}}.
\]
Immediately, we have that $\mathsf{rank}(J_{\mathbf{P'}_{\mathcal{S}}}) \le \mathsf{rank}(J_{\mathbf{E'}_{\mathcal{S}}})$. Thus, to show a lower bound on $\mathsf{rank}(J_{\mathbf{E'}_{\mathcal{S}}})$, it suffices to show a lower bound on $\mathsf{rank}(J_{\mathbf{P'}_{\mathcal{S}}})$. We next show the following claim. 
\begin{claim}
Over $\mathbb{K}$, $\mathsf{rank}(J_{\mathbf{P'}_{\mathcal{S}}}) \ge |\mathcal{S}|/k^k$ provided $|\mathbb{F}|> d/k$. 
\end{claim}
\begin{proof}
Consider the row in $J_{\mathbf{P'}_{\mathcal{S}}}$ corresponding to $\alpha \in \mathbb{Z}^{+k}$ where $\alpha_{k} =0$. Let us denote it by $J_{\mathbf{P'}_{\alpha}}$. It is given by 
\[
J_{\mathbf{P'}_{\alpha}} = \alpha_k \cdot \bigg[ z_1^{\alpha_{k-1}} \cdot \prod_{j=1}^{k-2} A[1,j]^{\alpha_j}  \ldots  \quad \ldots z_m^{\alpha_{k-1}} \cdot \prod_{j=1}^{k-2} A[m,j]^{\alpha_j}  \bigg]
\]
Now, consider the $m$ points in $\mathbb{F}^{k-1}$ given by $A_z$ where $z=(z_1, \ldots, z_m)$.  Call these points $y_1, \ldots, y_m$. Then, up to the scaling factor $\alpha_k$, $J_{\mathbf{P'}_{\alpha}} $, 
is simply the evaluation of the monomial $\mathbf{y}^{\alpha'}$ where $\alpha'= \alpha - \mathbf{e}_k$ at the points $y_1, \ldots, y_m$. Thus, if we restrict our attention to those $\alpha$ such that $\alpha_k \not =0 $, then these rows constitute the multivariate interpolation matrix for the monomials given by such $\alpha$'s at the points $y_1, \ldots, y_m$. 
We would like to prove the non-singularity of this multivariate interpolation matrix. Let us look at the subset of $\mathcal{S}$ such that $1 \le \alpha_{k-1} \le d/k$ 
and $0 \le \alpha_k < d/k$. While this is a subset of $\mathcal{S}$, note that the size of this subset is at least $|\mathcal{S}|/k^k$. Further, now, let us assume our points $y_1, \ldots, y_m$ are obtained as follows: Choose some subset $L$ of $\mathbb{F}$ of size at least $d/k$ and consider the $(d/k)^k$ obtained by taking a direct product of these points.  The interpolation matrix is then a $k$-fold tensor product of the univariate interpolation matrix at $L$ (of degree $d/k$). If all the points in $L$ are distinct and non-zero, then the univariate interpolation matrix is the Vandermonde matrix which has a non-zero determinant. This will imply that its $k$-fold tensor product has a non-zero determinant concluding the proof.
\end{proof}
This implies that $\mathsf{rank}(J_{\mathbf{E'}_{\mathcal{S}}}) \ge d^k/k^k$. This means that we can choose a square  submatrix of size $(d/k)^k \times (d/k)^k$ of $J_{\mathbf{E'}_{\mathcal{S}}}$ of full rank. This means there is a subset of $\mathcal{S}$ of size $d^k/k^k$ (call it $\mathcal{S}'$) such that  
\[
\mathsf{rank}(J_{\mathbf{E'}_{\mathcal{S'}}}) = |\mathcal{S'}|. 
\]
Now, applying Corollary~\ref{corr:Wooley} to the map $\mathbf{E'}_{\mathcal{S'}}$, we see that the range of the map has size $2^{\tilde{\Omega}(\log^{k-1}(1/\epsilon))}$. This immediately proves Lemma~\ref{lem:algfinite}.

\section{A Fourier-Based Learning Algorithm for PMDs}
\label{sec:body_learning}
In this section, we discuss Theorem~\ref{thm:learning}, our learning result for PMDs.
Our technique crucially uses Fourier analysis.
We note that the recent work of Diakonikolas, Kane, and Stewart \cite{DiakonikolasKS16a} also uses Fourier analysis to learn $k$-SIIRVs, i.e. sums of independent integer valued random variables taking values in $\{0, 1, \ldots, k-1\}$. We note that our use of Fourier analysis is somewhat different from theirs. In particular, \cite{DiakonikolasKS16a} use the Fourier transform over some discrete group $\mathbb{Z}_m$ for an appropriately chosen $m$. In contrast, we do the usual Fourier analysis over $\mathbb{Z}^k$. It turns out doing Fourier analysis over $\mathbb{Z}^k$ (rather than a finite group) avoids many problems and may be viewed as the natural domain for Fourier analysis for such problems.

We believe the application of Fourier analysis to learn such structured distributions is interesting in its own right and might have application in the future towards obtaining learning algorithms for other interesting classes of distributions. In particular, the recent work on the population recovery problem~\cite{WigdersonY12, MoitraS13, LovettZ15} may also be viewed as an example of use of Fourier analysis towards learning of structured distributions.

We now give a high level description of our learning algorithm. The $(n,k)$-PMD $Z$, that we are aiming to learn is supported on $\mathbb{Z}^k$ and hence the Fourier transform $\widehat{Z}$ is defined for every $\xi \in [-1,1]^k$ as $\widehat{Z}(\xi) = \mathbf{E}[e^{i \cdot \pi \cdot \langle \xi, Z \rangle}]$. 
While our actual algorithm does not perform Fourier inversion explicitly, it resembles Fourier inversion fairly closely. 
For the moment, assume that we are performing Fourier inversion. 
It immediately becomes clear that a vanilla Fourier inversion will not work --
this is because the Fourier transform is supported on $[-1,1]^k$ which is an uncountable set and thus we cannot evaluate $\widehat{Z}(\cdot)$ at all points of the support.  Rather what we show is that the Fourier transform of a PMD decays exponentially around any point of the form $\{-1,0, 1\}^k$. In particular, if $\Sigma$ is the covariance matrix of the PMD, then we show that for $\xi \in [-1/2,1/2]^k$, 
$$
|\widehat{Z}(\xi)| = e^{-\Theta(1) \xi^T  \Sigma \xi}. 
$$
{\ifnum\camr=0 Refer to Corollary~\ref{corr:FourierSupportLem} for the precise bounds. \fi}Similar exponential decay of Fourier transform is also true around the other points of the form $\{-1,0,1\}^k$. Let us use $V = \prod_{i=1}^k (1+ \sigma_i)$ where $\sigma_i^2$ are the eigenvalues of $\Sigma$. It is not difficult to show that all but an $\epsilon$-fraction of the mass of $Z$ falls on a set of size $V \cdot \log^k (1/\epsilon)${\ifnum\camr=0~(Lemma~\ref{lem:Chernoffmult})\fi}.  On the other hand, using the exponential decay of the Fourier transform, we have the following crucial claim:  We identify a region $\mathcal{S} \subseteq [-1,1]^k$ of volume $\log^k(1/\epsilon)/V$ such that 
\begin{equation}\label{eq:11}
\int_{\xi \not \in \mathcal{S}} |\widehat{Z}(\xi)|^2 d \xi \le \widetilde{O}_k\left(\frac{\epsilon}{V}\right). 
\end{equation}
{\ifnum\camr=0 Refer to Claim~\ref{cl:Fouriererror} for the precise bounds. \fi}Also, in this informal description, we use $\widetilde{O}$ to hide the dependence on $k$ as well as the polylogarithmic factors of $1/\epsilon$.  This implies that if $H$ is another function such that  $|\widehat{H}(\xi) - \widehat{Z}(\xi) | \le \epsilon$ inside $\mathcal{S}$ and $0$ outside $\mathcal{S}$, then 
\begin{equation}\label{eq:22}
\int_{\xi \in [-1,1]^k} |\widehat{H}(\xi) - \widehat{Z}(\xi)|^2 \le \widetilde{O}_k\left(\frac{\epsilon}{V}\right).  
\end{equation}
By using Plancherel's identity and Cauchy-Schwarz, it immediately follows that 
$ \sum_{z \in \mathbb{Z}^k} |H(z) - Z(z)| \le \widetilde{O}(\epsilon)$.
In other words, if we perform Fourier inversion by estimating $\widehat{Z}$ pointwise to error $\epsilon$ within $\mathcal{S}$ and setting it to be $0$ outside $\mathcal{S}$, then the $\ell_1$ distance between our hypothesis and $Z$ is $\widetilde{O}(\epsilon)$. We remark that the factor $1/V$ that we get in (\ref{eq:11}) and (\ref{eq:22}) is crucial for our algorithm to succeed. The only detail we have not specified is how to approximate $\widehat{Z}$ to error $\epsilon$ inside $\mathcal{S}$. Note that $\mathcal{S}$ still has infinitely many points. However, what we show is that there is a carefully chosen grid $\mathcal{S}_{\mathsf{grid}}$ of size $\widetilde{O}_k((1/\epsilon)^{k})$ such that estimating $\widehat{Z}(\xi)$ on $\mathcal{S}_{\mathsf{grid}}$ to error $\epsilon$ suffices to estimate $\widehat{Z}(\xi)$ on $\mathcal{S}$ (to error $2\epsilon$). This is done by assigning the estimate of $\widehat{Z}$ of the nearest grid point. This uses the choice of the grid points in $\mathcal{S}$ along with the Lipschitz property of the Fourier transform. Note that since we are evaluating the Fourier transform at $(1/\epsilon)^k$ points to error $\epsilon$, we need $\tilde{O}_k(1/\epsilon^2)$ samples. 

One caveat that remains to be discussed is that we have not commented on the time complexity of the Fourier inversion algorithm. In the actual algorithm, we do not perform Fourier inversion out of concerns of time complexity and the fact that the resulting measure obtained from Fourier inversion while computable need not be samplable. Instead, we use the structural characterization of PMDs from \cite{DaskalakisKT15} to decompose $Z \approx G + S$ where $G$ is a discretized Gaussian and $S$ is a $(\mathsf{poly}(k/\epsilon), k)$ PMD~(Theorem~\ref{thm:struct}). {Using samples from $Z$, we can spectrally approximate its covariance matrix, which then gives us a good handle on the covariance matrix of $G$, as $S$ has small size. In particular, we can construct a $(1/\ve)^{O(k)}$-size  spectral cover for the covariance matrix of $G$ using the covariance matrix of $Z$. So we can assume that $G$ is essentially known, and the challenge is to uncover $S$, using samples from $Z$. Of course, $Z$ is not actually equal to $G+S$, but if our overall algorithm uses $\ell=\tilde{O}_k(1/\ve^2)$ samples, and we have approximate equality of $Z$ and $G+S$ to within variation distance $O(1/\ell^2)$, say, then we can pretend that $Z$ is actually equal to $G+S$ for the purposes of our analysis{\ifnum\camr=0~(Claim~\ref{clm:convenientform})\fi}. So knowing $G$, and getting samples from $G+S$ we need to uncover $S$. We follow a linear programming approach to find the probability density of $S$.} We enforce constraints on this density so that the Fourier transform of $G +S$  approximately matches the empirical Fourier transform of $Z$. Our choice of the error and points at which we evaluate $\widehat{Z}$ and enforce this constraint is informed by the discussion above. What is crucial here is that the Fourier transform of $S$ is a linear function of its probability density and thus we are left to solve a system of linear constraints. 
\ifnum\tp=1
For further details on our learning algorithm, see 
\ifnum\camr=0
Section \ref{sec:learning}.
\else
the full version of this paper.
\fi
\fi

\ifnum\tp=1
\section{Details from Section~\ref{sec:body_learning}}
\else
\subsection{Fourier Properties of PMDs}
\fi

\label{sec:learning}

 The main idea behind learning PMDs is to look at the Fourier spectrum of PMDs. Specifically, we will prove two structural results about PMDs. One is that the Fourier spectrum of PMDs (roughly) has an exponential decay around the origin. The second result we will prove is the Fourier spectrum is a Lipschitz function and thus to estimate the Fourier spectrum in the entire domain, it suffices to compute it at a few points. Combining these two results along with standard statements on Fourier inversion show that if we construct a hypothesis distribution which approximates the Fourier spectrum of the target PMD at the chosen points and also exhbits a similar exponential decay in the Fourier spectrum, then the hypothesis distribution is close to the target PMD. While the condition on Fourier decay is not algorithmically easy to impose, we show that using some ideas from \cite{DaskalakisKT15}, the problem of imposing these constraints reduces to linear programming. 
We will first quickly review the notion of Fourier spectrum of integer valued distributions. 
\begin{definition}
For a random variable $Z$ supported in $\mathbb{Z}^k$ and $\xi \in [-1, 1]^k$, we define
$$
\widehat{Z}(\xi) = \mathbf{E}_{z \sim Z} [e^{i \cdot \pi \cdot \xi \cdot z}].
$$
We note that the reason to restrict $\xi \in [-1,1]^k$ is because the Fourier spectrum of distributions supported on $\mathbb{Z}^k$ is periodic with the fundamental period being the box $[-1,1]^k$.  
\end{definition}
Let us now recall the setting: $P = Z_1 + \ldots + Z_n$ where $Z_i$ are independent random variables supported on 
$\{\mathbf{e}_1, \ldots, \mathbf{e}_k\}$. Also for $1 \le i \le n$ and $1 \le j \le k$, let $p_{ij} = \Pr[X_i = j]$. To specify the next lemma, for any $\xi \in [-1,1]^k$, we will need to define an associated vector $\zeta  \in [-1,1]^k$. 
For any $\xi \in [-1,1]$, define the associated  $\zeta$ as follows: 
\[
\zeta = \begin{cases}  \xi & \textrm{ if } \xi \in [-1/2,1/2] \\
                 1 - \xi &  \textrm{ if } \xi \in [1/2,1]\\
                -1 -  \xi  &  \textrm{ if } \xi \in [-1,-1/2] \end{cases}
\] 
For $\xi \in [-1,1]^k$, we obtain $\zeta \in [-1/2,1/2]^k$, by doing the above operation coordinatewise. Let $\mathcal{B}_{\ell_p}(z,r)$ denote the $\ell_p$ ball 
of radius $r$ around $z$. 
To put it succinctly, for $\xi \in \mathcal{B}_{\ell_{\infty}}(z,1/2)$ where $z \in \{-1,0,1\}^k$, 
we obtain $\zeta = (\xi-z) \circ (-1)^{z}$. Here $(-1)^z$ denotes the vector in $\{-1,0,1\}^k$ where the $i^{th}$ coordinate is $(-1)^{z_i}$  and $\circ$ denotes the Hadamard (the coordinate-wise) product of two vectors. 
For every $\xi \in [-1,1]$, we call it Type~1 if $\xi \in [-1/2,1/2]$, Type~2 if $\xi \in [1/2,1]$ and Type~3 otherwise. 
\begin{claim}\label{clm:abc}
Let $X$ be a CRV with covariance matrix $\Sigma$. Then, $|\widehat{X}(\xi)|^2 \le 1-  \frac 1 5 \cdot \zeta^{T} \cdot \Sigma\cdot  \zeta$. 
\end{claim}
\begin{proof}
\[
\widehat{X}(\xi) = \sum_{j=1}^k p_j \cdot e^{i \cdot \pi \cdot \xi_j} = \sum_{j=1}^k p_j \cdot \cos(\pi \cdot \xi_j) +  i \cdot p_j \cdot  \sin (\pi \cdot \xi_j). 
\]
This implies
\begin{eqnarray*}
|\widehat{X}(\xi)|^2 &=& \sum_{j=1}^k p_j^2 + 2 \cdot \sum_{1 \le i < j \le k} p_i \cdot p_j \cdot \cos (\pi (\xi_i - \xi_j)). \\
&=& \sum_{j=1}^k p_j^2 + 2 \cdot \sum_{1 \le i < j \le k} p_i \cdot p_j \cdot \bigg( 1 - 2 \sin^2\bigg(\frac{\pi (\xi_i - \xi_j)}{2}\bigg)\bigg) \\
&=& 1 - 4 \sum_{1 \le i \le j \le k} 4 p_i \cdot p_j \cdot \sin^2 \bigg(\frac{\pi (\xi_i - \xi_j)}{2}\bigg). 
\end{eqnarray*}
We will first show that for every $i, j$,
\begin{equation}\label{eq:ij}
\sin^2 \bigg(\frac{\pi (\xi_i - \xi_j)}{2}\bigg) \ge \frac 1 {5} \cdot (\zeta_i - \zeta_j)^2. 
\end{equation}
To prove this, we do a simple case analysis, and use the inequality $\sin^2(\p x/2) \ge x^2/5$ for $|x| \le 3/2$:
\begin{itemize}
\item If both $\xi_i$ and $\xi_j$ are of the same type, then note that $| \xi_i - \xi_j | = | \zeta_i - \zeta_j | \le 1$ which gives the required inequality.

\item If $\xi_i$ and $\xi_j$ are type $2$ and $3$, then note that $|\xi_i - \xi_j| = |2 - (\zeta_i - \zeta_j)|$. This implies 
that 
$$
\sin^2 \bigg(\frac {\pi (\xi_i - \xi_j)} 2 \bigg) = \sin^2  \bigg(\frac {\pi (\zeta_i - \zeta_j)} 2\bigg).
$$
Noting that $|\zeta_i - \zeta_j| \le 1$ gives the required inequality.

\item If $\xi_i$ is of type $1$ and $\xi_j$ is of type 2, then note that the maximum value that $|\xi_i - \xi_j| $ can take is $ 3/2$. On the other hand, notice that $|\xi_i - \xi_j| \ge |\zeta_i - \zeta_j|$. These two facts immediately imply that  
$$
\sin^2  \bigg(\frac{\pi (\xi_i - \xi_j)}{2}\bigg) \ge \frac 1 {5} \cdot (\zeta_i - \zeta_j)^2. 
$$
The exact same situation holds if $\xi_i$ is of type $1$ and $\xi_j$ is of type 3. 
\end{itemize}
Having shown (\ref{eq:ij}), see that this implies that 
$$
|\widehat{X}(\xi)|^2 \le 1 -  \frac 15  \sum_{1 \le i \le j \le k} p_i p_j (\zeta_i - \zeta_j)^2. 
$$
However, 
$$
 \sum_{1 \le i \le j \le k} p_i p_j (\zeta_i - \zeta_j)^2 = \zeta^T \cdot \Sigma \cdot \zeta. 
$$
This finishes the proof. 
\end{proof}
As a corollary, we have the following. 
\begin{corollary}\label{corr:FourierSupportLem}
 For any $(n,k)$-PMD $P$ with covariance matrix $\Sigma$, we have that for any $\xi \in [-1,-1]^k$:
  $$|\widehat{P}(\xi)|^2 \le \exp \left(- \frac 15 \zeta^T \cdot \Sigma \cdot \zeta \right). $$
\end{corollary}
\begin{proof}
This follows simply by noticing that for a PMD $P = X_1 + \ldots + X_n$, 
$$
\widehat{P}(\xi)  = \prod_{j=1}^n \widehat{X_j}(\xi). 
$$
Using Claim~\ref{clm:abc}, we have
$$
|\widehat{P}(\xi)   |^2 \le \prod_{j=1}^n \bigg( 1- \zeta^T \cdot \Sigma_i \cdot \zeta \bigg),
$$
where $\Sigma_i$ is the covariance matrix of $X_i$. Using the inequality, $1 - x^2 \le e^{-x^2/2}$ (for $|x| \le 1$), we have
$$
| \widehat{P}(\xi)   |^2 \le \prod_{j=1}^n e^{- \frac 15 \cdot \zeta^T \cdot \Sigma \cdot \zeta}. 
$$
\end{proof}
\begin{lemma}\label{lem:fourierlipschitz}
Let $X$ be a random variable supported on $\mathbb{R}^k$ with mean $\mu$ and covariance matrix $\Sigma$. Then the Fourier transform $\widehat{X}$ is 
Lipschitz in the following sense: 
\[
 | \widehat X(\xi) - \widehat X(\xi') | \le \pi \cdot (\xi-\xi') \cdot (\Sigma + \mu^T \cdot \mu) \cdot (\xi-\xi').
\]
\end{lemma}
\begin{proof}
\begin{eqnarray*}
\big| \mathbf{E}[e^{i \cdot \pi \cdot  \xi \cdot X}] -\mathbf{E}[e^{i \cdot  \pi \cdot \xi' \cdot X}] \big| = \big| \mathbf{E}[e^{i \cdot  \pi \cdot \xi \cdot X} \cdot \big(e^{i \cdot \pi \cdot (\xi - \xi') \cdot X} -1\big)] \big| \le \mathbf{E}[\big|\big(e^{i \cdot  \pi \cdot (\xi - \xi') \cdot X} -1\big)\big|]
\end{eqnarray*}
It is easy to observe that for any $\theta \in \mathbb{R}$, $|e^{i \cdot \theta} -1| \le \theta^2$. Applying this to the above inequality, we have
\[
\big| \mathbf{E}[e^{i \cdot \pi \cdot \xi \cdot X}] -\mathbf{E}[e^{i \cdot  \pi \cdot \xi' \cdot X}] \big| \le \mathbf{E}[|\pi \cdot (\xi-\xi') \cdot X|^2] = \pi \cdot (\xi-\xi') \cdot (\Sigma + \mu^T \cdot \mu) \cdot (\xi-\xi').
\]
\end{proof}
We also have the following variant of the above lemma which will be useful for us. 
\begin{lemma}\label{lem:fourierlipschitz1}
Let $X$ and $Y$ be two distributions in $\mathbb{R}^k$ with the same mean $\mu$ and covariance $\Sigma$. 
If for a point $\xi \in \mathbb{R}^k$, $|\widehat{X}(\xi) - \widehat{Y}(\xi)| \le \epsilon$, then 
$$
\left| \widehat{X}(\xi + \zeta) - \widehat{Y}(\xi + \zeta)\right| \le \epsilon +2 \zeta^T \cdot \Sigma \cdot \zeta.
$$ 
\end{lemma}
\begin{proof}
To prove this, note that 
$$
\left| \widehat{X}(\xi + \zeta) - \widehat{Y}(\xi + \zeta)\right| =  \left| \widehat{X'}(\xi + \zeta) - \widehat{Y'}(\xi + \zeta)\right| 
$$
where $X'$ and $Y'$ are the centered random variables obtained by centering $X$ and $Y$. 
Likewise, 
$$
\left| \widehat{X}(\xi ) - \widehat{Y}(\xi )\right| =  \left| \widehat{X'}(\xi ) - \widehat{Y'}(\xi )\right| 
$$
However, by Lemma~\ref{lem:fourierlipschitz}, we have 
$$
\left| \widehat{X'}(\xi ) -\widehat{X'}(\xi + \zeta)\right|  \le \zeta^T \cdot \Sigma \cdot \zeta. 
$$
Applying the same for the $\widehat{Y'}$ and applying triangle inequality, we get the claim. 
\end{proof}

We now state the Plancherel identity in this setting. In particular, we have the following easy claim (which can be found in any standard text on Fourier analysis). 
\begin{claim}\label{clm:Plancherel}
Let $F: \mathbb{Z}^k\rightarrow \mathbb{R}$. Then, 
\[
\int_{\xi_1 \in [-1,1]} \ldots \int_{\xi_k \in [-1,1]} |\widehat{F}(\xi)|^2 d\xi_1 d\xi_2 \ldots d\xi_k= \sum_{z \in \mathbb{Z}^k} |F(z)|^2. 
\]
\end{claim}
In our setting, $F = G-H$ where $G$ and $H$ are probability measures supported on $\mathbb{Z}^k$. The following easy consequence of Claim~\ref{clm:Plancherel} will be useful for us. 
\begin{corollary}\label{corr:Fourinv}
Let $F, H: \mathbb{Z}^k \rightarrow [0,1]$ be a probability distributions such that for some $\mathcal{S} \subseteq \mathbb{Z}^k$, $\Pr[F \not \in \mathcal{S}], \Pr[H \not \in \mathcal{S}] \le \ve$. Then  
\[
\dtv(F,H) \le \ve + \sqrt{|\mathcal{S}|} \cdot \sqrt{\bigg( \int_{\xi_1 \in [-1,1]} \ldots \int_{\xi_k \in [-1,1]} |\widehat{F - H}(\xi)|^2 d\xi_1 d\xi_2 \ldots d\xi_k\bigg)}
\]
\end{corollary}
\begin{proof}
By Claim~\ref{clm:Plancherel}, 
\[
\sum_{z \in \mathcal{S}} |F(z) - H(z)|^2 \le \int_{\xi_1 \in [-1,1]} \ldots \int_{\xi_k \in [-1,1]} |\widehat{F-H}(\xi)|^2 d\xi_1 d\xi_2 \ldots d\xi_k= \sum_{z \in \mathbb{Z}^k} |F(z)-H(z)|^2.
\]
Applying Cauchy-Schwarz inequality, 
\[
\sum_{z \in \mathcal{S}} |F(z) - H(z)| \le \sqrt{|\mathcal{S}|} \cdot \sqrt{\bigg( \int_{\xi_1 \in [-1,1]} \ldots \int_{\xi_k \in [-1,1]} |\widehat{F - H}(\xi)|^2 d\xi_1 d\xi_2 \ldots d\xi_k\bigg)}
\]
\end{proof}
The above corollary demonstrates that to learn the PMD to error $\ve$, it suffices to produce another distribution $H$ whose Fourier spectrum is very close to the Fourier spectrum of $F$ (the ``very small" is quantified
by the effective support of $F$). 

\begin{lemma}[Lemma 8 from \cite{DaskalakisKT15}]\label{lem:MCestimate}
  Given sample access to a $(n,k)$-PMD $X$ with mean $\m$ and covariance matrix $\S$, there exists an algorithm which can produce estimates $\hat \mu$ and $\hat \S$ such that with probability at least $9/10$ for every vector $y$:
  $$|y^T(\hat \mu - \mu)| \leq \ve \sqrt{y^T \S y} ~~ \text{and} ~~ |y^T(\hat \S - \S)y| \leq \ve y^T \S y \sqrt{1 + \frac{y^T y}{y^T \S y}} $$
The sample and time complexity are $O(k^4/\ve^2)$.
\end{lemma}
~\\
The following is guaranteed by the multidimensional Chernoff bound. 
\begin{lemma}\label{lem:Chernoffmult}
Let $X$ be a $(n,k)$-PMD with mean $\mu$ and covariance matrix $\Sigma$. Let $L_{r} = \{ z: (z-\mu) \cdot (z-\mu)^t \preceq r \cdot \Sigma\}$.  
For $r = O(\log (1/\ve) + \log k)$, 
$$
\Pr[X \not \in L_r] \le \ve. 
$$
\end{lemma}
This implies the following corollary. 
\begin{lemma}\label{lem:ellipsoidlattice}
Let $\mu \in \mathbb{R}^k$ and $\Sigma$ be a PSD matrix with eigenvalues $\sigma_1^2 \ge \ldots \ge \sigma_k^2 \ge 0$.  
Let $L_{r} = \{ z: (z-\mu) \cdot (z-\mu)^t \preceq r \cdot \Sigma\}$. The total number of points of $\mathbb{Z}^k$ which lie in $L_r$ is bounded by
$ \prod_{i=1}^k \left( 2 \sigma_i \sqrt{ k r} + 1 \right)$.
\end{lemma}

\begin{proof}
Let the eigenvectors of $\Sigma$ be $v_1, \ldots, v_k$ with the corresponding eigenvalues $\sigma_1^2 ,\ldots , \sigma_k^2$. Consider any two distinct $x,y \in L_r$. Since $x$ and $y$ are distinct, hence
there must be some $1 \le i \le k$ such that the projection of $x$ and $y$ along $v_i$ is separated by $k^{-1/2}$.  
~\\
Let us denote the projection of $x$ along $v_i$ by $x_i$. Then the condition of lying in $L_r$ implies that $|x_i| \le r \cdot \sigma_i$. It is then easy to see that if the number of integer points in $L_r$ is more 
than 
$ \prod_{i=1}^k \left( 2 \sigma_i \sqrt{ k r} + 1 \right)$, then there must be $2$ points $x$ and $y$ and some $1\le i\le k$, $|x_i  - y_i | \le k^{-1/2}$. 
\end{proof}

Given a $(n,k)$ PMD $Z$, let $\widehat{\mu}$ and $\widehat{\Sigma}$ be the empirical mean and covariance matrices obtained from Lemma~\ref{lem:MCestimate}. For technical reasons, instead of working with $\widehat{\Sigma}$, 
we create a new PSD matrix $\tilde{\Sigma}$ which is obtained as follows: $\widehat{\Sigma}$ and $\tilde{\Sigma}$ have the same eigenvectors. If $\widehat \sigma_i^2$ is the eigenvalue of $\widehat{\Sigma}$ corresponding to $v_i$, 
then the corresponding eigenvalue $\tilde{\sigma}^2$ of $\tilde{\Sigma}$ is $(1-3 \ve) \cdot \widehat \sigma_i^2$. Further, after this operation, if a particular eigenvalue of $\tilde{\Sigma}$ is smaller than $\ve$, we modify that singular value to make it $0$.  Doing this operation ensures that 
\[
|y^T(\tilde \S - \S)y| \leq \ve y^T \S y \sqrt{1 + \frac{y^T y}{y^T \S y}} \ \textrm{ and } \tilde \S \preceq \S
\]
which implies that for all eigenvalues $\tilde{\s}_i^2 \le \s_i^2$.
Note that to learn the PMD, one possible strategy is to evaluate the Fourier transform of a $(n,k)$-PMD in the region $(\xi_1, \ldots, \xi_k) \in [-1,1]^k$ and then perform a Fourier inversion. Unfortunately, this is too expensive for us. 
Instead, we show that the Fourier transform only needs to be evaluated in a very small region.

\begin{definition}
For a point $z \in \{-1,0, 1\}^k$, define $C_{z,r}$ as
$$C_{z,r} =   \{y : \sum \tilde{ \s}^2_i \left( \tilde{v}_i \cdot ((-1)^z \circ  (y-z)) \right)^2 \le r \}$$
and
$R_{z}$ as  $$R_{z}= \mathcal{B}_{\ell_{\infty}}(z,1/2) \cap [-1,1]^k.$$
\end{definition}
Note that $[-1,1]^k$ can be partitioned into the regions $R_{z}$ (for $z \in \{-1,0,1\}^k$). In other words, 
\begin{eqnarray*}
[-1,1]^k = \cup_{z \in \{-1,0,1\}^k} R_{z}. 
\end{eqnarray*}

\begin{claim}\label{cl:Fouriererror}
Let $S_r = \cup_{z \in \{-1,0,1\}^k} ( R_{z} \cap C_{z,r} ) $ and let $\overline{S}_r = [-1,1]^k \setminus S_r$.  Then, 
$$
\int_{(\xi_1, \ldots, \xi_k) \in \overline{S}_r} |\widehat{Z}(\xi)|^2 d\xi_1 \ldots d\xi_k \le
  e^{-r/10} 
\prod_{i=1}^k  \frac{1}{\max \bigg\{\tilde{\sigma_i}, \frac{1}{k} \bigg\}}
$$
\end{claim}
\begin{proof}
\begin{eqnarray*}
\int_{(\xi_1, \ldots, \xi_k) \in \overline{S}_r} |\widehat{Z}(\xi)|^2 d\xi_1 \ldots d\xi_k &=& 
\sum_{z \in \{-1,0,1\}^k} 
\int_{(\xi_1, \ldots, \xi_k) \in R_{z} \setminus C_{z,r}} |\widehat{Z}(\xi)|^2 d\xi_1 \ldots d\xi_k \\
%
%
\end{eqnarray*}
We now individually bound each of the summands. Fix any particular $z$. 
Using Corollary~\ref{corr:FourierSupportLem}
\[
\int_{(\xi_1, \ldots, \xi_k) \in R_{z} \setminus C_{z,r}} |\widehat{Z}(\xi)|^2 d\xi_1 \ldots d\xi_k  \le \int_{(\xi_1, \ldots, \xi_k) \in R_{z} \setminus C_{z,r}} e^{-\frac 15 \zeta^T \cdot \Sigma \cdot \zeta} d\xi_1 \ldots d\xi_k 
\]
where $\zeta= (-1)^z \circ (\xi-z)$. To bound this, 
note that since $\tilde{\Sigma} \preceq \Sigma$ and $R_z \subset B_{\ell_2(z,\sqrt{k}/2)}$, we get
\[
\int_{(\xi_1, \ldots, \xi_k) \in R_z \setminus C_{z,r}} e^{-\frac 15 \zeta^T \cdot \Sigma \cdot \zeta} d\xi_1 \ldots d\xi_k \le \int_{(\xi_1, \ldots, \xi_k) \in B_{\ell_2(z,\sqrt{k}/2)} \setminus C_{z,r}} e^{-\frac 15 \cdot \zeta^T \cdot \tilde{\Sigma} \cdot \zeta} d\xi_1 \ldots d\xi_k
\]
Using the fact that $\ell_2$ balls are invariant under rotation, the right hand integral becomes
\[
\int_{\sum_i {\tilde \sigma}_i^2 w_i^2 > r; \ (w_1,...,w_k) \in B_{\ell_2(0,\sqrt{k}/2)}}  e^{-\frac 15 \cdot \sum {\tilde \sigma}_i^2 w_i^2} dw_1 \ldots dw_k
\]

Since, $B_{\ell_2(0,\sqrt{k}/2)} \subset B_{\ell_{\infty}(0,\sqrt{k}/2)}$, this is upper bounded by
\[
\int_{\sum {\tilde \sigma}_i^2 w_i^2 > r; |w_i| \le \sqrt{k}/2}  e^{-\frac 15 \cdot \sum {\tilde \sigma}_i^2 w_i^2} dw_1 \ldots dw_k
\]
To upper bound this integral, let $Y_1 = \{j: \tilde{\sigma}_j \le 1/k\}$. Then, 
\begin{eqnarray*}
\int_{\sum {\tilde \sigma}_i^2 w_i^2 > r; |w_i| \le \sqrt{k}/2}  e^{-\frac 15 \cdot \sum {\tilde \sigma}_i^2 w_i^2} dw_1 \ldots  dw_k &\le& \int_{\sum {\tilde \sigma}_i^2 w_i^2 > r; |w_i| \le \sqrt{k}/2}  e^{-\frac 15 \cdot \sum_{i \not \in Y_1} {\tilde \sigma}_i^2 w_i^2} dw_1 \ldots  dw_k \\ 
&\le& \int_{\sum_{i \not \in Y_1} {\tilde \sigma}_i^2 w_i^2 > r/2; |w_i| \le \sqrt{k}/2}  e^{-\frac 15 \cdot \sum_{i \not \in Y_1} {\tilde \sigma}_i^2 w_i^2} dw_1 \ldots  dw_k
\end{eqnarray*}
The last inequality uses that $r>2$. This integral is now easily seen to be bounded by 
\[
\prod_{i \in Y_1} k \cdot e^{-r/10} \cdot \prod_{i \not \in Y_1}  \frac{1}{\tilde{\sigma_i}}. 
\]
This is exactly the same bound as stated in the claim. 
\end{proof}
\begin{claim} \label{cl:Fouriervolume}
Let $S_r = \cup_{z \in \{-1,0,1\}^k} ( R_{z} \cap C_{z,r} ) $. Then, 
$$
\int_{(\xi_1, \ldots, \xi_k) \in S_r}  d\xi_1 \ldots d\xi_k \leq  3^k \prod_{i=1}^k \min \bigg\{ \sqrt{k}, \frac{ 2 \sqrt{r}}{ \tilde{\sigma}_i } \bigg\}
$$
\end{claim}
\begin{proof}
Doing the exact same calculation as in the proof of Claim~\ref{cl:Fouriererror}, 
\[
\int_{(\xi_1, \ldots, \xi_k) \in {S}_r}  d\xi_1 \ldots d\xi_k  \le 3^k \cdot  \int_{\sum_i {\tilde \sigma}_i^2 w_i^2 \leq r; |w_i| \le \sqrt{k}/2}   dw_1 \ldots dw_k
\] 
By using the same manipulation as before, we can upper bound this integral by $3^k \prod_{i=1}^k \min \bigg\{ \sqrt{k}, \frac{ 2 \sqrt{r}}{ \tilde{\sigma}_i } \bigg\}$. 
\end{proof}

\subsection{Learning algorithm for PMDs} 

Theorem~\ref{thm:struct}, the structure theorem from \cite{DaskalakisKT15}, allows us to assume that the PMD $Z$ is essentially a discretized Gaussian $G$ convolved with a sparse PMD $S$ where the sparse PMD is supported on only $\poly(k/\ve)$ summands.  

By setting `$\ve$' from the Theorem statement to be  $\ve^{10}$, we get that $d_{TV}(Z,G+S) \le \ve^{10}$. Because our subsequent learning algorithm will take $\ll O(\ve^{-10})$ samples, we assume that we are getting samples from $G + S$ instead of $Z$ and that $Z = G+S$. Furthermore, using the following claim from~\cite{DaskalakisKT15}, we can get a spectral estimate with accuracy $\ve^{10}$ of the mean and covariance of the Gaussian $G$ by guessing the partition of coordinates in the covariance matrix of the Gaussian and going through all elements of the spectral cover of PSD matrices around a fine estimate $\widehat{S}$ for $\S$ obtained using $k/\ve^2$ samples from Lemma~\ref{lem:MCestimate}.

\begin{claim}[Lemma 9 from \cite{DaskalakisKT15}]
Let A be a symmetric $k  \times k$ PSD matrix with minimum eigenvalue 1 and let $\mathcal{S}$ be the
set of all matrices $B$ such that $|y^T \cdot (  A - B) \cdot y | \le \ve_1 y^T \cdot A \cdot y + \ve_2  y^T y $ 
where $\ve_1 \in [0,1/4)$ and $\ve_2 \in[0, \infty)$. Then, there exists a cover $\mathcal{S}_{\ve}$ of size 
$ (k \cdot (1+\ve_2)/\ve)^{k^2}$ such that any $B \in \mathcal{S}$  is $\ve$-spectrally close to some element in the cover.
\end{claim}

The spectral closeness translates to closeness $\ve^{10}$ in total variation distance between Gaussians (Lemma~\ref{lem:dtvgaussian}) and again since we will be taking $\ll O(\ve^{-10})$ samples in the learning algorithm, we can assume that the gaussian $G$ has exactly the mean $\m_G$ and covariance $\S_G$  we guessed. 

Similarly, we can assume that the sparse-PMD has known mean and covariance $\mu_S$ and $\S_S$. This is because any PMD with $n'$ summands is $\ve^{10}$-close in total variation to a PMD where all the probabilities are rounded to multiples of $\lceil n' k /\ve^{10} \rceil^{-1}$. This fact follows from union-bounding all the errors of the individual summands. Since $n' = \poly(k/\ve)$ for the sparse PMD, all coordinates are multiples of $\poly(\ve/k)$, which implies that the mean and covariance coordinates are also multiples of $\poly(\ve/k)$ and we can guess them exactly using $\poly(k/\ve)^{k^2}$ guesses. Again, since this sparse PMD is $\ve^{10}$ close and we will be getting much fewer samples, we can assume that the sparse PMD has exactly the mean and covariance we guessed.


At this point, we have argued the following: 
\begin{claim}\label{clm:convenientform}
The PMD $Z$ is equal to the sum of a discretized Gaussian $G$ and a sparse PMD $S$ with $\poly(k/\ve)$ summands. The mean and covariance of the Gaussian $(\m_G,\S_G)$ and of the sparse PMD $(\m_S,\S_S)$ are known, which implies that the mean and covariance of the overall PMD $Z$ is equal to $(\m,\S) = (\m_S,\S_S) + (\m_G,\S_G)$.
\end{claim}
Our learning algorithm attempts to recover the sparse PMD in order to learn the overall distribution $Z$. However, imposing the condition that the distribution we are trying to estimate is a sparse PMD will involve solving non linear equations making the computation intractable.
Rather, we will seek to learn a sparse distribution $S'$ supported on $[0, T]^k$ where $T = \poly(k/\ve)$. 

To learn this distribution, we will attempt to estimate its Fourier Transform. We will be mostly interested in points on the grid: $$\mathcal{V} = \bigg\{ \alpha_1 \cdot \frac{\ve}{k^{2k} \cdot 6^k} \cdot \frac{ \vec  v_1}{{\max\{1,\s_1\}}} + \dots + \alpha_k \cdot \frac{\ve}{k^{2k} \cdot 6^k} \cdot \frac{ \vec v_k}{{\max\{1,\s_k\}} } : \alpha_i \in \mathbb{Z} \bigg\}$$ where  $(\vec  v_i, \s^2_i)$ are the eigenvector, eigenvalue pairs of the matrix $ \Sigma$. From Corollary~\ref{corr:FourierSupportLem}, we know that the Fourier transform decays exponentially as we move away from $\{-1,0,1\}^k$, and in particular Claim~\ref{cl:Fouriererror} bounds the total mass contained at a distance at least $r$ from all the points.
For our purposes, we set $r = O(k \log k + k \log (1/\ve))$ and perform the following steps to learn the sparse distribution $S'$.

\begin{enumerate}
\item Create variables $p_{\alpha}$ for every $\alpha \in [0,T]^k$ with the constraints $0 \le p_{\alpha} \le 1$ and 
$\sum_{\alpha \in T^k} p_{\alpha}=1$.
\item \label{step:grid} Let $\mathcal{A}_1 = \cup_{z \in \{-1,0,1\}^k} \{\xi : \sum { \s}^2_i \left( {v}_i \cdot ((-1)^z \circ  (\xi-z)) \right)^2 \le r \}$. Let $\mathcal{V}_1$ be the points of the grid $\mathcal{V}$ that lie in $\mathcal{A}_1$. For each of those points, get an estimate $\widehat{Z}_{est}$ of $\widehat{Z}$ such that $|\widehat{Z}_{est} - \widehat{Z}| <  \frac{\ve}{6^k \cdot k^{2k}}$ and then impose linear constraints on $\{p_{\alpha}\}$ so that $|Re[\widehat{S'}(\xi) \cdot \widehat{G}(\xi) - \widehat{Z}_{est}(\xi) ]| \le  \frac{\ve}{6^k \cdot k^{2k}}$ and $|Im[ \widehat{S'}(\xi) \cdot \widehat{G}(\xi) - \widehat{Z}_{est}(\xi) ] | \le  \frac{\ve}{6^k \cdot k^{2k}}$. 
\item Let $\s_{G,i}^2, \vec v_{G,i}$ be the eigenvalues and eigenvectors of $\S_G$\footnote{
We note that, since we may have eigenvalues which are both large and small in magnitude, a naive eigendecomposition algorithm would incur a cost which depends on $n$.
However, as we only require the eigenvalues and eigenvectors approximately, this cost can be avoided by applying an appropriate power-iteration method.
The cost in terms of $k$ and $1/\ve$ is dominated by the other steps in our algorithm.} and consider the set:
$$\mathcal{A}_2 = \cup_{z \in \{-1,0,1\}^k} \{\xi : \sum { \s}^2_{G,i} \left( {\vec v}_{G,i} \cdot ((-1)^z \circ  (\xi-z)) \right)^2 \le \frac r 2  \wedge \sum { \s}^2_i \left( \vec {v}_i \cdot ((-1)^z \circ  (\xi-z)) \right)^2 > r \} $$
Construct a grid of points in $[-1,1]^k$ with a spacing of $\frac{\ve^{2k}}{k^{2k} \cdot 6^k}$ in every direction. Let $\mathcal{V}_2$ be the subset of these points which fall in $A_2$. For all these points impose the conditionss that $|Re[\widehat{S'}(\xi)]| \le e^{- \zeta^T \Sigma_S \zeta}$ and $|Im[\widehat{S}(\xi)]| \le e^{- \zeta^T \Sigma_S \zeta}$ that follow from Corollary~\ref{corr:FourierSupportLem}. 

\item Finally, add the constraints $\sum_\a p_\a \a  = \m_S$ and $\sum_\a p_\a (\a-\m_S) (\a-\m_S)^T = \S_S$
\end{enumerate}

\noindent
Note that in Step~\ref{step:grid}, $\mathcal{V}_1$ has size at most $\left( \frac {\sqrt{r} \cdot k^{2k} \cdot 6^k} {\ve} \right)^k$. If we naively estimated every Fourier coefficient in $\mathcal{V}_1$ the number of samples would be too high because every Fourier coefficient requires $\log(1/\d)/{\ve^2}$ samples to learn with accuracy $\ve$ and probability of failure $1-\d$. However, we can instead take $O( k \log(r/\ve) / \ve^2)$ samples and reuse the same samples to compute all the required Fourier coefficients. Since the probability of error is very small a simple union bound among all of the coefficients, shows that with at least constant probability all of them can be estimated within $\ve$.

To complete the learning algorithm, we repeat the steps above for each of the guessed mean and covariance matrices $(\m_G,\S_G),(\m_S,\S_S)$. We then perform a hypothesis selection algorithm to choose a distribution within $O(\ve)$ from each of the distributions we obtain. We made $O( \poly(k/\ve)^{k^2} )$ guesses, and thus obtained $O( \poly(k/\ve)^{k^2} )$ candidate hypotheses. Applying the following tournament theorem for hypothesis selection from~\cite{DaskalakisK14}, we can select a good estimate in $ O\left( \left( \frac k {\ve} \right)^2 \log(k/\ve) \right)$ samples in $O( \poly(k/\ve)^{k^2} )$ runtime.

\begin{theorem}[Theorem 19 of \cite{DaskalakisK14}] \label{thm:tournament}
There is an algorithm {\tt FastTournament}$(X, {\cal H},\ve,\delta)$, which is given sample access to some distribution $X$ and a collection of distributions ${\cal H}=\{H_1,\ldots,H_N\}$ over some set ${\cal D}$, access to a PDF comparator for every pair of distributions $H_i, H_j \in {\cal H}$, an accuracy parameter $\ve >0$, and a confidence parameter $\delta >0$.  The algorithm makes
{$O\left({\log {1/ \delta} \over \ve^2} \cdot \log N\right)$} draws from each of $X, H_1,\ldots,H_N$ and returns some $H \in {\cal H}$ or declares ``failure.''  If there is some $H^* \in {\cal H}$ such that $\dtv(H^*,X) \leq \ve$ then with probability at least $1-\delta$ the distribution $H$ that {\tt
FastTournament} returns satisfies $\dtv(H,X) \leq {512} \ve.$ The total number of operations of the algorithm is {$O\left( {\log{1 / \delta} \over \ve^2} \left(N \log N + \log^2 {1 \over \delta}\right) \right)$}.
Furthermore, the expected number of operations of the algorithm is {$O\left( {N\log{N /\d} \over \ve^2}\right)$}.
\end{theorem}

\medskip
\noindent
{\bf Proof of correctness:}

We first show that there is a solution to $\{p_{\alpha}\}$ which satisfies all the constraints. 
 Indeed, if we set the sparse distribution $S'$ to be equal to the distribution $S$ of the sparse PMD we defined above, we get:
\begin{enumerate}
  \item $\sum_{\alpha \in T^k} p_{\alpha}=1$ since $S$ is a probability distribution supported on $[0,T]^k$.
  \item The constraint $|Re[\widehat{S'}(\xi) \cdot \widehat{G}(\xi) - \widehat{Z}_{est}(\xi) ]| \le \frac{\ve}{6^k \cdot k^{2k}}$ is satisfied since for $S'=S$, $$|Re[\widehat{S}(\xi) \cdot \widehat{G}(\xi) - \widehat{Z}_{est}(\xi) ]| = |Re[\widehat{Z}(\xi) - \widehat{Z}_{est}(\xi) ]| \le |\widehat{Z}(\xi) - \widehat{Z}_{est}(\xi)| \le \frac{\ve}{6^k \cdot k^{2k}}.$$
  The derivation for the constraint on the imaginary part is identical.
  \item From Corollary~\ref{corr:FourierSupportLem}, the sparse PMD satisfies $|\widehat{S}(\xi)| \le e^{- (1/5) \cdot \zeta^T \Sigma_S \zeta}$ everywhere in $[-1,1]^k$. This condition implies the imposed constraints which are only evaluated in few points.
  \item The distribution $S$ has mean $\mu_S$ and covariance $\S_S$, so the last constraint is satisfied.
\end{enumerate}

We now  prove that any feasible solution $\{p_{\alpha}\}$ to the above system of constraints defines a distribution $S'$ such that $\dtv(S + G, S' +G) \le \epsilon$. To show this, we divide the space $[-1,1]^k$ into three parts: $\mathcal{A}_1$, 
$\mathcal{A}_2$ and $\mathcal{A}_3= [-1,1]^k \setminus (\mathcal{A}_1 \cup \mathcal{A}_2)$. 

\begin{claim}\label{clm:F1}
\[
\int_{\xi \in \mathcal{A}_1} |\widehat{S+G}(\xi) - \widehat{S'+G}(\xi)|^2 d \xi   = O \left( \frac{\epsilon^2 \cdot r^{k/2}}{k^{3k} \cdot \prod_{i=1}^k \max\{{\sigma_i}, 1\}}\right)
\]
\end{claim}
\begin{proof}
Consider any point $\xi$ in $\mathcal{A}_1$. Then, note that there is some $\xi' \in \mathcal{V}$ such that
for $1 \le i \le k$, $\langle \xi - \xi' , \vec v_{i} \rangle \le \frac{\epsilon}{k^{2k} \cdot 6^k \cdot {\max\{1, \sigma_i\}}}$. Applying Lemma~\ref{lem:fourierlipschitz1}, we get that 
$$
|\widehat{S+G}(\xi) - \widehat{S'+G}(\xi)| \le \frac{\epsilon \cdot \sqrt{k}}{6^k \cdot k^{2k}} +  |\widehat{S+G}(\xi') - \widehat{S'+G}(\xi')| \le \frac{\epsilon \cdot 2 \sqrt{k}}{6^k \cdot k^{2k}}. 
$$
Applying Claim~\ref{cl:Fouriervolume},  we have
$$
\int_{\xi \in \mathcal{A}_1} |\widehat{S+G}(\xi) - \widehat{S'+G}(\xi)|^2 d \xi   \le \max_{\xi \in \mathcal{A}_1} |\widehat{S+G}(\xi) - \widehat{S'+G}(\xi)|^2 \cdot \int_{\xi \in \mathcal{A}_1} d \xi  = O \left( \frac{\epsilon^2 \cdot r^{k/2}}{k^{3k} \cdot \prod_{i=1}^k \max\{{\sigma_i}, 1\}}\right). 
$$
This finishes the proof. 
\end{proof}

\begin{claim}\label{clm:F2}
\[
\int_{\xi \in \mathcal{A}_2} |\widehat{S}(\xi) - \widehat{S'}(\xi)|^2 d \xi   = O \left( \frac{\epsilon^2 \cdot r^{k/2}}{k^{3k} \cdot \prod_{i=1}^k \max\{{\sigma_i}, 1\}}\right)
\]
\end{claim}
\begin{proof}
Note that $A_2$ is a subset of the set
$$B_2 = \cup_{z \in \{-1,0,1\}^k} \{\xi : \sum { \s}^2_{G,i} \left( {\vec v}_{G,i} \cdot ((-1)^z \circ  (\xi-z)) \right)^2 \le \frac r 2\}.$$ 
We bound the volume of the set $B_2$. To do this, we again apply Claim~\ref{cl:Fouriervolume}, and get that 
$$
\int_{\xi \in B_2} d\xi = 3^k \cdot \frac{r^{k/2} \cdot k^{k/2}}{\prod_{i=1}^k \max \{\sigma_{G,i}, 1\}}.
$$
Note that for any point $\xi \in A_2$, we there is a point $\xi'$ such that $\Vert \xi - \xi' \Vert_2 \le \frac{\epsilon^{2k}}{k^{2k} \cdot 6^k}$ and that 
$|\widehat{S'} (\xi')| \le e^{-(1/5) \cdot \zeta'^T \Sigma_S \cdot \zeta'}$. Since the variance of $\Sigma_S$ is at most $\poly(k/\epsilon)$ in every direction, we get that 
$$
|\widehat{S'} (\xi)| \le e^{-(1/5) \cdot \zeta^T \Sigma_S \cdot \zeta} + \frac{\epsilon^{2k}}{k^{2k} \cdot 6^k}. 
$$
This implies that 
$$
\int_{\xi \in A_2} |\widehat{S}(\xi) - \widehat{S'}(\xi)|^2 \le \int_{\xi \in A_2} 2 \cdot |\widehat{S}(\xi)|^2 + 2 \cdot |\widehat{S'}(\xi)|^2 d\xi
$$
By applying Claim~\ref{cl:Fouriervolume} to bound the volume of the set $A_2 \subseteq B_2$ and using the fact that $|\widehat{S}(\xi)|^2$ is at most $e^{-r/20}$, we get that the first integral is at most 
\begin{eqnarray*}
\int_{\xi \in A_2} 2 \cdot |\widehat{S}(\xi)|^2 &\le& e^{-r/20}  \cdot \prod_{i=1}^{k} \frac{1}{\max\{ \sigma_{G,i}, 1/k\}} \\
&\le& e^{-r/20} \cdot \poly(k/\epsilon)^k \cdot \prod_{i=1}^{k} \frac{1}{\max \{\sigma_{i}, 1/k\}} 
\end{eqnarray*}
The last inequality uses the fact that whenever $\sigma_{G,i} \le {\sigma}_i$, it must imply that all the variance comes from $S$ and thus ${\sigma_i} \le \poly(k/\epsilon)$. By plugging the value of $r$,  we get that 
$$
\int_{\xi \in A_2} 2 \cdot |\widehat{S}(\xi)|^2 \le \epsilon^k  \cdot \prod_{i=1}^k \frac{1}{\max \{ {\sigma}_i, 1 \}}. 
$$
The calculation for the second integral is similar. 
\begin{eqnarray*}
\int_{\xi \in A_2}  2 \cdot |\widehat{S'}(\xi)|^2 d\xi &\le& \int_{\xi \in A_2} e^{-(1/5) \cdot \zeta^T \Sigma_S \cdot \zeta}  + \int_{\xi \in A_2} \frac{\epsilon^{2k}}{k^{2k} \cdot 6^k} d\xi \\
&\le& \epsilon^k  \cdot \prod_{i=1}^k \frac{1}{\max \{ {\sigma}_i, 1 \}} + \int_{\xi \in A_2} \frac{\epsilon^{2k}}{k^{2k} \cdot 6^k} d\xi\\
&\le& \epsilon^k  \cdot \prod_{i=1}^k \frac{1}{\max \{ {\sigma}_i, 1 \} } + \int_{\xi \in B_2} \frac{\epsilon^{2k}}{k^{2k} \cdot 6^k} d\xi
\end{eqnarray*}
Here the first inequality follows by exactly the same calculation we did for the first integral whereas the second inequality uses that $A_2 \subseteq B_2$.  
Now, that we had derived that 
$$
\int_{\xi \in B_2} d\xi = 3^k \cdot \frac{r^{k/2} \cdot k^{k/2}}{\prod_{i=1}^k \max \{\sigma_{G,i}, 1\}}.
$$
However, $\max\{\sigma_{G,i} , 1 \} \ge \epsilon^{\Theta(1)} \cdot \max \{ {\sigma_i}, 1\}$ (because the variance of $S$ is at most $\poly(1/\epsilon)$ in any direction. This implies that 
$$
\int_{\xi \in B_2} d\xi = \left(\frac{3}{\epsilon} \right)^k \cdot \frac{r^{k/2} \cdot k^{k/2}}{\prod_{i=1}^k \max \{{\sigma_{i}}, 1\}}.
$$
This implies that 
$$
\int_{\xi \in A_2}  2 \cdot |\widehat{S'}(\xi)|^2  \le  \epsilon^k  \cdot \prod_{i=1}^k \frac{1}{\max \{ {\sigma}_i, 1 \}}. 
$$
\end{proof}
\begin{claim}\label{clm:F3}
\[
\int_{\xi \in \mathcal{A}_3} |\widehat{S+G}(\xi) - \widehat{S'+G}(\xi)|^2 d \xi   = O \left( \frac{\epsilon^2 \cdot r^{k/2}}{k^{3k} \cdot \prod_{i=1}^k \max\{{\sigma_i}, 1\}}\right)
\]
\end{claim}
\begin{proof}
Note that $\widehat{S'+G }(\xi) = \widehat{G}(\xi) \cdot \widehat{S' }(\xi) $. Thus, 
$|\widehat{S'+G }(\xi) |^2 \le | \widehat{G}(\xi)|^2$. 
Applying Claim~\ref{cl:Fouriererror} and noting that $$A_3 \subseteq  \cup_{z \in \{-1,0,1\}^k} \{\xi : \sum { \s}^2_{G,i} \left( {\vec v}_{G,i} \cdot ((-1)^z \circ  (\xi-z)) \right)^2  > \frac r 2 \}$$ 
we obtain that 
$$
\int_{\xi \in A_3} |\widehat{G}(\xi)|^2 d \xi = e^{-r/10} \cdot k^k \cdot \prod_{i=1}^k \frac{1}{\max\{1, \sigma_{G,i} \}}
$$
Again using the fact that the variance of $S$ in any direction is at most $\poly(k/\epsilon)$, 
$$
\int_{\xi \in A_3} |\widehat{G}(\xi)|^2 d \xi \le  e^{-r/10} \cdot \poly(k/\epsilon)^k \cdot \prod_{i=1}^k \frac{k}{\max\{1, {\sigma}_{i} \}}
$$
Plugging in the value of $r$, we get
that
$$
\int_{\xi \in A_3} |\widehat{G}(\xi)|^2 d \xi \le  \epsilon^k \cdot \prod_{i=1}^k \frac{1}{\max\{1, {\sigma}_{i} \}}
$$
This immediately implies the claim. 
\end{proof}
Combining Claim~\ref{clm:F1}, Claim~\ref{clm:F2} and Claim~\ref{clm:F3}, we get that 
$$
\int_{\xi \in [-1,1]^k} |\widehat{S+G}(\xi) -\widehat{S'+G}(\xi)|^2 d\xi = \epsilon^2 \cdot (k \log(1/\epsilon))^{O(k)} \cdot \prod_{i=1}^{k} \frac{1}{\max\{{\sigma}_i,1 \}}. 
$$
We now apply Corollary~\ref{corr:Fourinv} to derive that 
$$
d_{TV}(S+G, S'+G) \le \epsilon \cdot  (k \log(1/\epsilon))^{O(k)} \cdot \sqrt{\prod_{i=1}^{k} \frac{1}{\max\{{\sigma}_i,1 \}} \cdot \prod_{i=1}^k (2 \sigma_i  \sqrt{k r} +1)}. 
$$
This is at most $d_{TV}(S+G, S'+G) \le \epsilon \cdot (k \log(1/\epsilon))^{O(k)}$. Setting $\ve$ to be  $\frac{\ve'} {\poly(k,\log(1/\ve'))^k}$, we complete the proof of Theorem~\ref{thm:learning}. 

\section{Open Problems}

A number of interesting questions regarding Poisson Multinomial distributions are left open by this work and \cite{DiakonikolasKS16c}.
We outline a few of them here. 
\begin{enumerate}
\item {\bf The complexity of learning Poisson Multinomials.}
This work and \cite{DiakonikolasKS16c} both give algorithms for learning PMDs. 
The sample and time complexities are polynomial in $1/\ve$ and exponential in $k$.
Meanwhile, \cite{DaskalakisKT15} gives an algorithm with a sample complexity polynomial in both parameters, but the time complexity is exponential in $k$ and $1/\ve$.
Is there an algorithm for learning PMDs with sample and time complexities both polynomial in $k$ and $1/\ve$?


\item {\bf Exploring the connection between Poisson Multinomials and Laplacian matrices.}
In this work, we described a cover for the set of $(n,k)$-PMDs of size $O_{k,\ve}(n^{O(k)})$.
Our construction relied crucially on 
\ifnum\camr=0
Observation~\ref{obs:pmdlaplacian} (which states that the covariance matrix of a PMD is Laplacian) 
\else
the fact that the covariance matrix of a PMD is Laplacian
\fi
and spectral sparsification results for Laplacian matrices.
With this connection in hand, can one derive other results for PMDs using the wealth of literature on Laplacian matrices?

\item {\bf A tighter central limit theorem.}
\cite{ValiantV11} proves a central limit theorem between an $(n,k)$-GMD and a discretized Gaussian with the same mean and covariance, upper bounding their total variation distance by $O(k^{4/3}\s^{-1/3} \log^{2/3} n)$, where $\s^2$ is the smallest eigenvector of the covariance matrix of the GMD.
Both this paper and \cite{DiakonikolasKS16c} qualitatively improve this bound by removing the dependence on $n$, while keeping the dependence on $k$ and $1/\s$ still polynomial.
How well can a GMD be approximated by a discretized Gaussian?
In one dimension, the answer is $\Theta(1/\s)$ \cite{ChenGS10}, which implies a the answer for multiple dimensions is at least $\Omega(\sqrt{k}/\s)$.
\cite{DiakonikolasKS16c} achieves this dependence on $1/\s$ (up to log factors), but the optimal dependence on $k$ is currently unknown.

\item {\bf Sums of independent integer random vectors.}
  Poisson Multinomial distributions are the natural multivariate generalization of Poisson Binomial distributions, which have now been explored in this paper and other recent works \cite{DaskalakisKT15, DiakonikolasKS16c}.
  However, we currently have minimal understanding of any multivariate analogue of sums of independent integer random variables (i.e., SIIRVs, the object of study in \cite{BarbourC02, DaskalakisDOST13, DiakonikolasKS16a}), which we will denote as \emph{vector SIIRVs} (VSIIRVs).
  The natural definition of such an object is not immediately clear; one potential definition of an $(n,k,d)$-VSIIRV may be as the sum of $n$ independent random vectors in $\mathbb{N}^d$, where each is a distribution over all positive lattice points at $\ell_1$ distance at most $k$ from the origin.
  We note that an $(n,1,d)$-VSIIRV is an $(n,d)$-PMD, so these objects generalize PMDs at well.
  An interesting line of study would be to obtain structural, covering, and learning results for VSIIRVs.
\end{enumerate}

\bibliographystyle{alpha}
\bibliography{biblio}

\newcommand{\etalchar}[1]{$^{#1}$}
\newcommand{\noopsort}[1]{} \newcommand{\printfirst}[2]{#1}
  \newcommand{\singleletter}[1]{#1} \newcommand{\switchargs}[2]{#2#1}
\begin{thebibliography}{DDO{\etalchar{+}}13}

\bibitem[Bar88]{Barbour88}
Andrew~D. Barbour.
\newblock {S}tein's method and {P}oisson process convergence.
\newblock {\em Journal of Applied Probability}, 25:175--184, 1988.

\bibitem[B{\'C}02]{BarbourC02}
Andrew~D. Barbour and {\'C}ekanavi{\'c}ius.
\newblock Total variation asymptotics for sums of independent integer random
  variables.
\newblock {\em The Annals of Probability}, 30(2):509--545, 2002.

\bibitem[Ben05]{Bentkus05}
Vidmantas Bentkus.
\newblock A {L}yapunov-type bound in {R}d.
\newblock {\em Theory of Probability \& Its Applications}, 49(2):311--323,
  2005.

\bibitem[Ber41]{Berry41}
Andrew~C. Berry.
\newblock The accuracy of the {G}aussian approximation to the sum of
  independent variates.
\newblock {\em Transactions of the American Mathematical Society},
  49(1):122--136, 1941.

\bibitem[Blo99]{Blonski99}
Matthias Blonski.
\newblock Anonymous games with binary actions.
\newblock {\em Games and Economic Behavior}, 28(2):171--180, 1999.

\bibitem[Blo05]{Blonski05}
Matthias Blonski.
\newblock The women of {C}airo: Equilibria in large anonymous games.
\newblock {\em Journal of Mathematical Economics}, 41(3):253--264, 2005.

\bibitem[BSS12]{BatsonSS12}
Joshua~D. Batson, Daniel~A. Spielman, and Nikhil Srivastava.
\newblock Twice-{R}amanujan sparsifiers.
\newblock {\em SIAM Journal on Computing}, 41(6):1704--1721, 2012.

\bibitem[BSST13]{BatsonSST13}
Joshua~D. Batson, Daniel~A. Spielman, Nikhil Srivastava, and Shang-Hua Teng.
\newblock Spectral sparsification of graphs: Theory and algorithms.
\newblock {\em Communications of the ACM}, 56(8):87--94, 2013.

\bibitem[CDO15]{ChenDO2015}
Xi~Chen, David Durfee, and Anthi Orfanou.
\newblock On the complexity of {N}ash equilibria in anonymous games.
\newblock In {\em Proceedings of the 47th Annual ACM Symposium on the Theory of
  Computing}, STOC '15, pages 381--390, New York, NY, USA, 2015. ACM.

\bibitem[CDT09]{ChenDT09}
Xi~Chen, Xiaotie Deng, and Shang-Hua Teng.
\newblock Settling the complexity of computing two-player {N}ash equilibria.
\newblock {\em Journal of the ACM}, 56(3):14:1--14:57, 2009.

\bibitem[CGS10]{ChenGS10}
Louis~H.Y. Chen, Larry Goldstein, and Qi-Man Shao.
\newblock {\em Normal approximation by {S}tein’s method}.
\newblock Springer, 2010.

\bibitem[CST14]{ChenST14}
Xi~Chen, Rocco~A. Servedio, and Li~Yang Tan.
\newblock New algorithms and lower bounds for monotonicity testing.
\newblock In {\em Proceedings of the 55th Annual IEEE Symposium on Foundations
  of Computer Science}, FOCS '14, pages 286--295, Washington, DC, USA, 2014.
  IEEE Computer Society.

\bibitem[Dal99]{Dalbec99}
John Dalbec.
\newblock Multisymmetric functions.
\newblock {\em Beitr\"age zur Algebra und Geometrie}, 40(1):27--51, 1999.

\bibitem[DDO{\etalchar{+}}13]{DaskalakisDOST13}
Constantinos Daskalakis, Ilias Diakonikolas, Ryan O'Donnell, Rocco~A. Servedio,
  and Li~Yang Tan.
\newblock Learning sums of independent integer random variables.
\newblock In {\em Proceedings of the 54th Annual IEEE Symposium on Foundations
  of Computer Science}, FOCS '13, pages 217--226, Washington, DC, USA, 2013.
  IEEE Computer Society.

\bibitem[DGP09]{DaskalakisGP09}
Constantinos Daskalakis, Paul~W. Goldberg, and Christos~H. Papadimitriou.
\newblock The complexity of computing a {N}ash equilibrium.
\newblock {\em SIAM Journal on Computing}, 39(1):195--259, 2009.

\bibitem[DK14]{DaskalakisK14}
Constantinos Daskalakis and Gautam Kamath.
\newblock Faster and sample near-optimal algorithms for proper learning
  mixtures of {G}aussians.
\newblock In {\em Proceedings of the 27th Annual Conference on Learning
  Theory}, COLT '14, pages 1183--1213, 2014.

\bibitem[DKS16a]{DiakonikolasKS16c}
Ilias Diakonikolas, Daniel~M. Kane, and Alistair Stewart.
\newblock The {F}ourier transform of {P}oisson multinomial distributions and
  its algorithmic applications.
\newblock In {\em Proceedings of the 48th Annual ACM Symposium on the Theory of
  Computing}, STOC '16, New York, NY, USA, 2016. ACM.

\bibitem[DKS16b]{DiakonikolasKS16a}
Ilias Diakonikolas, Daniel~M. Kane, and Alistair Stewart.
\newblock Optimal learning via the {F}ourier transform for sums of independent
  integer random variables.
\newblock In {\em Proceedings of the 29th Annual Conference on Learning
  Theory}, COLT '16, 2016.

\bibitem[DKT15]{DaskalakisKT15}
Constantinos Daskalakis, Gautam Kamath, and Christos Tzamos.
\newblock On the structure, covering, and learning of {P}oisson multinomial
  distributions.
\newblock In {\em Proceedings of the 56th Annual IEEE Symposium on Foundations
  of Computer Science}, FOCS '15, Washington, DC, USA, 2015. IEEE Computer
  Society.

\bibitem[DP88]{DeheuvelsP88}
Paul Deheuvels and Dietmar Pfeifer.
\newblock {P}oisson approximations of multinomial distributions and point
  processes.
\newblock {\em Journal of multivariate analysis}, 25(1):65--89, 1988.

\bibitem[DP07]{DaskalakisP07}
Constantinos Daskalakis and Christos~H. Papadimitriou.
\newblock Computing equilibria in anonymous games.
\newblock In {\em Proceedings of the 48th Annual IEEE Symposium on Foundations
  of Computer Science}, FOCS '07, pages 83--93, Washington, DC, USA, 2007. IEEE
  Computer Society.

\bibitem[DP08]{DaskalakisP08}
Constantinos Daskalakis and Christos~H. Papadimitriou.
\newblock Discretized multinomial distributions and {N}ash equilibria in
  anonymous games.
\newblock In {\em Proceedings of the 49th Annual IEEE Symposium on Foundations
  of Computer Science}, FOCS '08, pages 25--34, Washington, DC, USA, 2008. IEEE
  Computer Society.

\bibitem[DP09]{DaskalakisP09}
Constantinos Daskalakis and Christos~H. Papadimitriou.
\newblock On oblivious {PTAS}'s for {N}ash equilibrium.
\newblock In {\em Proceedings of the 41st Annual ACM Symposium on the Theory of
  Computing}, STOC '09, pages 75--84, New York, NY, USA, 2009. ACM.

\bibitem[DP15]{DaskalakisP15b}
Constantinos Daskalakis and Christos~H. Papadimitriou.
\newblock Approximate {N}ash equilibria in anonymous games.
\newblock {\em Journal of Economic Theory}, 156:207--245, 2015.

\bibitem[Ess42]{Esseen42}
Carl-Gustaf Esseen.
\newblock On the {L}iapounoff limit of error in the theory of probability.
\newblock {\em Arkiv f\"or matematik, astronomi och fysik}, 28A(2):1--19, 1942.

\bibitem[Kal05]{Kalai05}
Ehud Kalai.
\newblock Partially-specified large games.
\newblock In {\em Proceedings of the 1st International Workshop on Internet and
  Network Economics}, WINE '05, pages 3--13, Berlin, Heidelberg, 2005.
  Springer.

\bibitem[Loh92]{Loh92}
Wei-Liem Loh.
\newblock {S}tein's method and multinomial approximation.
\newblock {\em The Annals of Applied Probability}, 2(3):536--554, 08 1992.

\bibitem[LZ15]{LovettZ15}
Shachar Lovett and Jiapeng Zhang.
\newblock Improved noisy population recovery, and reverse {B}onami-{B}eckner
  inequality for sparse functions.
\newblock In {\em Proceedings of the 47th Annual ACM Symposium on the Theory of
  Computing}, STOC '15, pages 137--142, New York, NY, USA, 2015. ACM.

\bibitem[Mil96]{Milchtaich96}
Igal Milchtaich.
\newblock Congestion games with player-specific payoff functions.
\newblock {\em Games and Economic Behavior}, 13(1):111--124, 1996.

\bibitem[MS13]{MoitraS13}
Ankur Moitra and Michael Saks.
\newblock A polynomial time algorithm for lossy population recovery.
\newblock In {\em Proceedings of the 54th Annual IEEE Symposium on Foundations
  of Computer Science}, FOCS '13, pages 110--116, Washington, DC, USA, 2013.
  IEEE Computer Society.

\bibitem[Roo02]{Roos02}
Bero Roos.
\newblock Multinomial and {K}rawtchouk approximations to the generalized
  multinomial distribution.
\newblock {\em Theory of Probability \& Its Applications}, 46(1):103--117,
  2002.

\bibitem[She10]{Shevtsova10}
I.G. Shevtsova.
\newblock An improvement of convergence rate estimates in the {L}yapunov
  theorem.
\newblock {\em Doklady Mathematics}, 82(3):862--864, 2010.

\bibitem[SS11]{SpielmanS11}
Daniel~A. Spielman and Nikhil Srivastava.
\newblock Graph sparsification by effective resistances.
\newblock {\em SIAM Journal on Computing}, 40(6):1913--1926, 2011.

\bibitem[ST11]{SpielmanT11}
Daniel~A. Spielman and Shang-Hua Teng.
\newblock Spectral sparsification of graphs.
\newblock {\em SIAM Journal on Computing}, 40(4):981--1025, 2011.

\bibitem[Sta69]{Starr69}
Ross~M Starr.
\newblock Quasi-equilibria in markets with non-convex preferences.
\newblock {\em Econometrica}, 37(1):25--38, 1969.

\bibitem[VdV00]{VanderVaart00}
A.~W Van~der Vaart.
\newblock {\em Asymptotic statistics}, volume~3.
\newblock Cambridge University Press, 2000.

\bibitem[VV10]{ValiantV10a}
Gregory Valiant and Paul Valiant.
\newblock A {CLT} and tight lower bounds for estimating entropy.
\newblock {\em Electronic Colloquium on Computational Complexity (ECCC)},
  17(179), 2010.

\bibitem[VV11]{ValiantV11}
Gregory Valiant and Paul Valiant.
\newblock Estimating the unseen: An $n/\log n$-sample estimator for entropy and
  support size, shown optimal via new {CLT}s.
\newblock In {\em Proceedings of the 43rd Annual ACM Symposium on the Theory of
  Computing}, STOC '11, pages 685--694, New York, NY, USA, 2011. ACM.

\bibitem[Woo96]{Wooley96}
Trevor~D. Wooley.
\newblock A note on simultaneous congruences.
\newblock {\em Journal of Number Theory}, 58(2):288--297, 1996.

\bibitem[WY12]{WigdersonY12}
Avi Wigderson and Amir Yehudayoff.
\newblock Population recovery and partial identification.
\newblock In {\em Proceedings of the 53rd Annual IEEE Symposium on Foundations
  of Computer Science}, FOCS '12, pages 390--399, Washington, DC, USA, 2012.
  IEEE Computer Society.

\end{thebibliography}
\appendix
\section{Proof of Lemma~\ref{lem:dtvgaussian}}
\label{sec:tvgaussian}
  We instead prove that $\dtv(X,Y) \leq \ve\sqrt{k}$ when $|v^T(\m_1 - \m_2)| \leq \ve \sqrt{k} s_v$ and 
  $|v^T(\S_1 - \S_2)v| \leq \frac{\ve s_v^2}{2}$, which we can see is equivalent to the lemma statement by a rescaling.

  Without loss of generality, assume that $\S_1$ and $\S_2$ are full rank.
  If not, the guarantees in the statement ensure that their nullspace is identical, and we can project to a lower dimension such that the resulting matrices are full rank. 

  First, we note that the assumptions in the lemma statement can be converted to be in terms of the minimum of the two variances, instead of the maximum.
  Define $\s_v^2 = \min\{v^T\S_1v, v^T\S_2v\}$.
  The second assumption can be rearranged to see that $(1 - \frac{\ve}{2})s_v^2 \leq \s_v^2$.
  Plugging this back into the second assumption gives that
  $$|v^T(\S_1 - \S_2)v| \leq \frac{\ve s_v^2}{2} \leq \frac{\ve \s_v^2}{2(1 - \frac{\ve}{2})} \leq \ve \s_v^2,$$
  where the last inequality holds for $\ve \leq 1$ (otherwise, the lemma's conclusion is trivial).
  Similarly, the second assumption also implies $\sqrt{1 - \frac{\ve}{2}}s_v \leq \s_v$, when plugged into the first assumption gives
  $$|v^T(\m_1 - \m_2)| \leq \ve\sqrt{k} s_v \leq \frac{\ve}{\sqrt{1 - \frac{\ve}{2}}}\sqrt{k} \s_v \leq \sqrt{2}\ve\sqrt{k}\s_v.$$
  For the remainder of the proof, we will use these guarantees instead of the ones in the lemma statement.

  We recall the standard formula for KL-divergence between two Gaussian distributions.
  Let $\{\l_i\}$ be the eigenvalues of $\S_2^{-1/2}\S_1\S_2^{-1/2}$.
  \begin{align*}
    \dkl\left(X||Y\right) 
    &= \frac12 \left((\m_2 - \m_1)^T\S_2^{-1}(\m_2 - \m_1) + \Tr(\S_2^{-1/2}\S_1\S_2^{-1/2}) - \ln\left(\det\left(\S_2^{-1/2}\S_1\S_2^{-1/2}\right)\right) - k\right) \\
    &= \frac12 \left((\m_2 - \m_1)^T\S_2^{-1}(\m_2 - \m_1) + \sum_{i=1}^k \left(\l_i - \ln \l_i - 1\right)\right)
  \end{align*}

  We bound the divergence induced by differences in the means and covariances separately.
  We start with the means.
  Note that
  $$|v^T(\m_2 - \m_1)| \leq \sqrt{2}\ve\sqrt{k} \s_v \Rightarrow \frac{|v^T(\m_2 - \m_1)|}{\sqrt{v^T\S_2v}} \leq \sqrt{2}\ve\sqrt{k}.$$
  Substituting $u = \S_2v$ gives
  $$\frac{|u^T \S_2^{-1} (\m_2 - \m_1)|}{\sqrt{u^T\S_2^{-1}u}} \leq \sqrt{2}\ve\sqrt{k}.$$
  We let $u = \m_2 - \m_1$, giving
  $$\sqrt{(\m_2 - \m_1)^T \S_2^{-1} (\m_2 - \m_1)} \leq \sqrt{2}\ve\sqrt{k},$$
  which implies
  $$(\m_2 - \m_1)^T \S_2^{-1} (\m_2 - \m_1) \leq 2\ve^2k.$$

  Now we bound the divergence induced by differences in the covariances.
  We bound the eigenvalues of $\S_2^{-1/2}\S_1\S_2^{-1/2}$.
  Note that
  $$|v^T(\S_1 - \S_2)v| \leq \ve \s_v^2 \Rightarrow \frac{1}{1 + \ve} \leq \frac{v^T\S_1v}{v^T\S_2v} \leq 1 + \ve.$$
  Substituting $u = \S_2^{1/2}v$ makes the latter condition equivalent to
  $$\frac{1}{1 + \ve} \leq \frac{u^T\S_2^{-1/2}\S_1\S_2^{-1/2}u}{u^Tu} \leq 1 + \ve.$$
  The Courant-Fischer Theorem implies that $\frac{1}{1 + \ve} \leq \l_i \leq 1 + \ve$ for all $i$.
  
  At this point, we note that $x - \ln x - 1 \leq (1 - x)^2$ for all $x \geq 1$.
  This implies 
  $$\sum_{i=1}^k \left(\l_i - \ln \l_i - 1\right) \leq \sum_{i=1}^k (1 - \l_i)^2 \leq \ve^2k.$$

  Thus, $\dkl(X||Y) \leq 2\ve^2k$.
  Applying Pinsker's inequality gives $\dtv(X,Y) \leq \ve\sqrt{k}$, as desired.

\end{document}